\documentclass[11pt]{article}

\usepackage{float} 
\usepackage{amsmath}
\usepackage{amssymb}
\usepackage{array}
\usepackage{bbm}
\usepackage{makeidx}
\usepackage{hyperref}
\usepackage{cleveref}
\usepackage{caption}
\usepackage{subcaption}
\usepackage{amsthm}
\usepackage{color}
\usepackage{mathtools}
\usepackage{graphicx} 
\usepackage[LGR,T1]{fontenc}
\usepackage{mathrsfs}
\usepackage{lipsum}
\usepackage{booktabs}
\usepackage{xcolor}
\usepackage[title]{appendix}
\hypersetup{
    colorlinks,
    linkcolor={blue!80!black},
    citecolor={green!50!black},
}
\usepackage{apptools}
\usepackage{tikz}
\usetikzlibrary{arrows, positioning, automata}
\usepackage{comment}
\usepackage[shortlabels]{enumitem}

\usepackage[noend]{algorithmic}
\usepackage[ruled,vlined]{algorithm2e}

\usepackage[mathscr]{euscript}
\DeclareSymbolFont{rsfs}{U}{rsfs}{m}{n}
\DeclareSymbolFontAlphabet{\mathscrsfs}{rsfs}

\def\R{{\mathbb R}}

\def\dim{{\rm dim}}
\def\nulls{{\rm null}}

\def\NM{\Phi_{\rm NM}}

\def\oZ{\overline{Z}}
\def\obG{\overline{\boldsymbol G}}

\def\cuS{\mathscrsfs{S}}
\def\de{{\rm d}}

\def\<{\langle}
\def\>{\rangle}

\def\pmax{p_{\max}}

\def\sgood{\mbox{\tiny \rm good}}

\def\salg{\mbox{\tiny \rm alg}}
\def\sHD{\mbox{\tiny \rm HD}}

\def\prob{{\mathbb P}}

\usepackage{mathtools}

\newtheoremstyle{myremark} 
    {\topsep}                    
    {\topsep}                    
    {\rm}                        
    {}                           
    {\bf}                        
    {.}                          
    {.5em}                       
    {}  

\newtheorem{claim}{Claim}[section]
\newtheorem{lemma}[claim]{Lemma}

\newtheorem{theorem}{Theorem}
\newtheorem{proposition}[claim]{Proposition}
\newtheorem{corollary}[claim]{Corollary}
\newtheorem{definition}[claim]{Definition}

\theoremstyle{myremark}
\newtheorem{remark}{Remark}[section]

\allowdisplaybreaks
\usepackage[top=1in,bottom=1in,left=1in,right=1in]{geometry}
\usepackage[utf8]{inputenc}

\def\ed{\stackrel{{\mathrm d}}{=}}
\def\ind{\mathbbm{1}}

\def\top{\intercal}

\def\bA{{\boldsymbol A}}
\def\bB{{\boldsymbol B}}
\def\bE{{\boldsymbol E}}
\def\obE{\overline{\boldsymbol E}}

\def\hbF{\widehat{{\boldsymbol F}}}

\def\cA{{\mathcal A}}

\def\bbJ{{\mathbb J}}

\def\cE{\mathcal{E}}
\def\cJ{\mathcal{J}}
\def\cG{\mathcal{G}}

\def\cO{\mathcal{O}}

\def\GOE{\mbox{GOE}}

\def\Vol{{\rm Vol}}

\def\S{{\mathbb S}}

\def\Sol{{\sf Sol}}
\def\Ts{{\sf T}}

\def\cE{\mathcal{E}}

\def\CC{\mathbb{C}}

\def\ZZ{\mathbb{Z}}

\def\calE{\mathcal{E}}

\def\calL{\mathcal{L}}

\def\calN{\mathcal{N}}

\def\calS{\mathcal{S}}

\def\bA{\mathbf{A}}
\def\bB{\mathbf{B}}
\def\bC{\mathbf{C}}
\def\bD{\mathbf{D}}
\def\bE{\mathbf{E}}
\def\bF{\mathbf{F}}
\def\bG{\mathbf{G}}

\def\bI{\mathbf{I}}

\def\bK{\mathbf{K}}

\def\bM{\mathbf{M}}

\def\bO{\mathbf{O}}
\def\bP{\mathbf{P}}
\def\bQ{\mathbf{Q}}
\def\bR{\mathbf{R}}
\def\bS{\mathbf{S}}
\def\bT{\mathbf{T}}
\def\bU{\mathbf{U}}
\def\bV{\mathbf{V}}
\def\bW{\mathbf{W}}
\def\bX{\mathbf{X}}

\def\bZ{\mathbf{Z}}

\def\ba{\boldsymbol{a}}
\def\bb{\boldsymbol{b}}

\def\be{\boldsymbol{e}}

\def\bg{\boldsymbol{g}}
\def\bh{\boldsymbol{h}}

\def\bs{\boldsymbol{s}}

\def\bu{\boldsymbol{u}}
\def\bv{\boldsymbol{v}}
\def\bw{\boldsymbol{w}}
\def\bx{\boldsymbol{x}}
\def\by{\boldsymbol{y}}
\def\bz{\boldsymbol{z}}

\def\bA{\boldsymbol{A}}
\def\bB{\boldsymbol{B}}
\def\bC{\boldsymbol{C}}
\def\bD{\boldsymbol{D}}
\def\bE{\boldsymbol{E}}
\def\bF{\boldsymbol{F}}
\def\bG{\boldsymbol{G}}

\def\bI{\boldsymbol{I}}

\def\bK{\boldsymbol{K}}

\def\bM{\boldsymbol{M}}

\def\bO{\boldsymbol{O}}
\def\bP{\boldsymbol{P}}
\def\bQ{\boldsymbol{Q}}
\def\bR{\boldsymbol{R}}
\def\bS{\boldsymbol{S}}
\def\bT{\boldsymbol{T}}
\def\bU{\boldsymbol{U}}
\def\bV{\boldsymbol{V}}
\def\bW{\boldsymbol{W}}
\def\bX{\boldsymbol{X}}

\def\bZ{\boldsymbol{Z}}

\def\normal{{\mathsf{N}}}
\def\bcH{{\boldsymbol {\mathcal H}}}

\def\bSigma{\boldsymbol{\Sigma}}

\def\bDel{\boldsymbol{\Delta}}

\def\bgamma{\boldsymbol{\gamma}}
\def\bvphi{\boldsymbol{\varphi}}

\def\bzero{\boldsymbol{0}}
\def\bfone{\boldsymbol{1}}
\def\balpha{\boldsymbol{\alpha}}

\def\GOE{{\sf GOE}}
\def\fix{{s}}

\def\oC{{\overline{C}}}

\def\reals{{\mathbb R}}
\def\naturals{{\mathbb N}}

\def\sT{{\sf T}}

\def\Ball{{\sf B}}
\def\Cap{{\sf Cap}}

\def\id{{\boldsymbol I}}

\def\bfzero{{\boldsymbol 0}}
\def\rP{{\rm P}}

\def\ed{\stackrel{{\rm d}}{=}}
\def\op{{\mbox{\rm\tiny op}}}

\newcommand\hF{\hat F}

\renewcommand{\P}{\mathbb{P}}
\newcommand{\E}{\mathbb{E}}

\newcommand{\eps}{\varepsilon}

\newcommand{\argmax}{\operatorname{argmax}}
\newcommand{\argmin}{\operatorname{argmin}}

\newcommand{\sign}{\operatorname{sign}}

\newcommand{\diag}{\operatorname{diag}}

\newcommand{\Lip}{\operatorname{Lip}}

\newcommand{\RN}[1]{%
  \textup{\uppercase\expandafter{\romannumeral#1}}%
}

\newcommand{\am}[1]{\noindent{{\color{blue}\textbf{\#\#\# AM:} \textsf{#1} \#\#\#}}}

\newcommand{\RNum}[1]{\uppercase\expandafter{\romannumeral #1\relax}}



\makeatletter
\newcommand*{\rom}[1]{\expandafter\@slowromancap\romannumeral #1@}
\makeatother

\begin{document}

\title{On Smale’s 17th problem over the reals}

\author{Andrea Montanari\thanks{Department of Statistics and Department of Mathematics, Stanford University} 
	\and 
	Eliran Subag\thanks{Department of Mathematics, Weizmann Institute of Science}
}
\pagenumbering{arabic}

\date{December 9, 2024}

\maketitle

\begin{abstract}
We consider the problem of efficiently solving a system of $n$ non-linear equations in $\reals^d$.
Addressing Smale's 17th problem stated in 1998, we consider a setting whereby the $n$ equations are random homogeneous polynomials of arbitrary degrees. In the complex case and for $n= d-1$, Beltr\'{a}n and Pardo proved the existence of an efficient randomized algorithm and Lairez recently showed it can be de-randomized to produce a deterministic efficient algorithm. Here we consider the real setting, to which previously developed 
methods do not apply. We describe a polynomial time algorithm that finds solutions (with high probability)
for $n= d -O(\sqrt{d\log d})$ if the maximal degree is bounded by $d^2$ and for $n=d-1$ if the maximal degree is larger than $d^2$.
\end{abstract}

\tableofcontents

\section{Introduction and main result}
In 1998 Steve Smale published a list of `Mathematical problems for the next century' 
\cite{smale1998mathematical}. Quoting from the original, his 17th problem asked:
\begin{quote}
\emph{``Can a zero of $n$ complex polynomial equations in $n$ unknowns be found approximately, on the average, 
in polynomial time with a uniform algorithm?''}
\end{quote}

\noindent The precise setting of the problem, originally introduced in the so-called B\'{e}zout  series \cite{ShubSmaleI,ShubSmaleII,ShubSmaleIII,ShubSmaleIV,ShubSmaleV} of Shub and Smale, is as follows. First, the assumption that the polynomials are homogeneous is made. Hence, for the set of solutions to be discrete, if the number of equations is $n$, the number of variables needs to be $d=n+1$, and one either has to restrict to solutions in the projective space as in the B\'{e}zout  series, or equivalently to solutions on the unit sphere as we shall in the current work.
Second, the question is a probabilistic one. Namely, it concerns a random polynomial system $\bF(\bx) :=(F_1(\bx),\dots,F_n(\bx))$ consisting of independent polynomials 
\begin{equation}\label{eq:Fi}
	F_{i}(\bx):=\sum_{k_1+\cdots+k_d=p_i}\Big(\frac{p_i!}{k_1!\cdots k_d!}\Big)^{\frac12}a^{(i)}_{k_1,\ldots,k_d}x_1^{k_1}\cdots x_d^{k_d}\,,
\end{equation}
where $\bx=(x_1,\ldots,x_d)\in \CC^d$,
$a^{(i)}_{k_1,\ldots,k_d}$ are i.i.d. complex standard Gaussian variables and $p_1,\ldots,p_n$ is some sequence of degrees.  
By classical results of Abel and Galois, polynomials generically do not have closed-form solutions. 
The 17th problem therefore asks for an algorithm that finds an `approximate solution'. The latter is defined as a point on the unit sphere such that (projected) Newton's method started from it converges immediately and quadratically fast to a solution, see Section \ref{subsec:apxsol} for a formal definition. The model of computation allows storing real numbers and complexity is measured by the number of read/writes and elementary operations with real numbers (formally, a Blum-Shub-Smale machine \cite{BSS}). To avoid technicalities we will also assume that the square root of a positive real number can be computed in a single time unit. The algorithm receives as input the coefficients of the system $\bF$ and is required to have a polynomial time complexity on average in $N:=\sum_{i=1}^{n}\binom{d+p_i-1}{p_i}$, the total number of coefficients. Roughly speaking, a uniform algorithm is an algorithm that can be implemented as a single program, whereas for non-uniform algorithms the implementation may depend on the input (in our setting, on $d$, $n$ and $p_1,\ldots,p_n$). For the precise definition, see \cite{BCSS} and \cite[Section 1.3]{BeltranPardo1}.

In its original form, Smale's 17th problem was recently solved after several major breakthroughs \cite{BeltranPardo1,BeltranPardo2,BeltranShub,BurgisserCucker,Lairez1} (see more in Section \ref{subsec:complexSmale17} below). We, however, will be interested in the real version of it. Indeed, after stating the problem in \cite{smale1998mathematical}, Smale also mentioned that the same problem is very interesting, and even more difficult, in the real case. Namely, when the coefficients $a^{(i)}_{k_1,\ldots,k_d}$ are assumed to be real Gaussian variables instead of complex and one searches for solutions on the unit sphere in $\R^d$ instead of $\CC^d$. We note that in the real case the specific choice of the combinatorial factor above makes  $F_i(\bx)$ rotationally invariant in law. Namely, $F_i(\bx)\overset{\calL}{=}F_i(O\bx)$ as processes for any orthogonal matrix $O$, where $\overset{\calL}{=}$ denotes equality in law. Henceforth, we shall assume this setting. 
Moreover, we will sometime consider an under-determined variant of Smale's 17th problem in the real case, by allowing  the number of equations $n$ to be strictly smaller than $d-1$. We still keep the notation $n$ for the number of equations and $d$ for the number of variables. 
We denote by $\pmax=\max_{i\le n}p_i$ the maximal degree, by $\ba:=\Big(\big(\frac{p_i!}{k_1!\cdots k_d!}\big)^{\frac12} a^{(i)}_{k_1,\ldots,k_d}\Big)_{i\le n,\,k_1+\cdots+k_d=p_i}$ the set of real Gaussian coefficients, and by $\S^{d-1}\subset \R^d$ the unit sphere. 
From now on, we shall always assume that $p_i\ge 2$ for all $i$, since for $p_i=1$ the zero set of $F_i(\bx)$ is the orthogonal space to $(a^{(i)}_1,\ldots,a^{(i)}_d)$ and, thanks to rotational invariance, via a pre-processing stage by a change of basis we may remove $F_i(\bx)$ from the system and reduce the number of variables by one.

Our main results are polynomial time algorithms to find approximate solutions. Their input is the set of random coefficients $\ba$, and possibly some additional parameters. Their output is either a point $\bx^{\salg}\in\S^{d-1}$ that belongs to the unit sphere, or the value FALSE. Denoting by  $\Sol:=\{\bx\in\S^{d-1}:\,\bF(\bx)=\bzero\}$ the set of exact solutions, $\{\Sol\neq\varnothing\}$ is the event that the system has at least one solution.  We say that an algorithm solves $\bF(\bx)=\bzero$ on an event $\calE$ if on $\calE\cap\{\Sol\neq\varnothing\}$ its output is an approximate solution  $\bx^{\salg}\in\S^{d-1}$ and on $\calE\cap\{\Sol=\varnothing\}$ its output is FALSE. We say that an algorithm solves $\bF(\bx)=\bzero$ with probability  at least $1-\eps$ if there exists an event $\calE$ such that $\P(\calE)\ge 1-\eps$ and the algorithm solves $\bF(\bx)=\bzero$ on $\calE$.

\begin{theorem}\label{thm:main0}
There exist absolute constants $A,C,\delta_0$ such that, for any 
 $\delta\in (0,\delta_0)$ and any $\delta'>0$ the following holds.\footnote{\label{ft:delta0}The only assumption on $\delta_0$ is that it is sufficiently small so that the condition $Ce^{-d/C}<\delta<\delta_0$ implies that $d$ is large enough so that  $\lfloor d- A(d\log d)^{1/2}\rfloor\ge1$. } Suppose that:
	\begin{enumerate}
		\item If $\pmax<d^{2}$ and $Ce^{-d/C}<\delta$, then $\eps=Ce^{-d/C}$ and $n=\lfloor d- A(d\log d)^{1/2}\rfloor$.
		\item If $\pmax\ge d^{2}$ and $\pmax^{-d}<\delta$, then $\eps =\pmax^{-d}$ and $n=d-1$.
		\item Otherwise, $\eps = \delta$ and $n=d-1$.
	\end{enumerate}

	\noindent Then, there exists an algorithm which solves $\bF(\bx)=\bzero$ with probability at least $1-\eps$
 (hence with probability at least $1-\delta$) whose real complexity $\chi$ is bounded, for an appropriate constant $C_0(\delta,\delta')>0$, by
	\[
	\chi\le C_0(\delta,\delta')N^{5.5+\delta'}\,.
	\]
\end{theorem}

Let us emphasize that the user can specify an arbitrarily small
failure probability $\delta>0$, and this theorem guarantees that the algorithm solves the problem with probability at least $1-\delta$
for any $d$ and $\pmax$. However, the theorem also implies that the success  probability will actually be significantly closer to one for
for large enough $d$ or $\pmax$. For large $d>d_0(\delta)=\max\{t:Ce^{-t/C}\ge \delta\}$ and moderate $\pmax<d^{2}$, the system is under-determined with $n=\lfloor d- A(d\log d)^{1/2}\rfloor$  and otherwise it is determined $n=d-1$.

It was recently proved in  \cite{SubagConcentrationPoly} that a solution exists with probability at least $1-C'/\sqrt d$ for some universal $C'>0$ (previously, this was known in the i.i.d. case $p_i\equiv p$  with $n=d-1$ \cite{Wschebor}). Hence, the algorithm of Theorem \ref{thm:main0} outputs an approximate solution with such probability.\footnote{Note that in contrast to the complex case of Smale's problem, in which the number of solutions is a.s. the  B\'{e}zout number $\prod_{i=1}^{d-1} p_i$, (for $n=d-1$) in the real case with positive probability there is no solution. Indeed, if $p_i$ are even and all the coefficients $\ba$ are positive and $\bF_i(\bx)>0$ for all $\bx\neq\bzero$.}

Theorem \ref{thm:main0} will be proved in Section \ref{sec:BF} by combining the following results about two algorithms we introduce below. For $a,b\in\R$, we denote $a\wedge b:=\min\{a,b\}$ and $a\vee b:=\max\{a,b\}$.

\begin{theorem}[Hessian Descent]\label{thm:main1}
For some absolute constants $A,C,C_0$, assuming that $\pmax\le d^2$ and $n=\lfloor d- A(d\log d)^{1/2}\rfloor\ge1$, there exists an algorithm such that:
\begin{enumerate}
\item[$(i)$] The real complexity $\chi$ of the algorithm is bounded by
\begin{align}
\chi \le C_0 Nd^{9/2}\pmax^5  \log(\pmax)^2\, .\label{eq:BoundComplexity}
\end{align}

\item[$(ii)$] The output is an approximate solution $\bx^{\salg}\in\S^{d-1}$ with probability at least $1-Ce^{-d/C}$ (and, in particular, it solves $\bF(\bx)=\bzero$ with the same probability).
\end{enumerate}
\end{theorem}

\begin{theorem}[Multi-Scale Search]\label{thm:main2}
Let  $\pmax,d\ge 2$  be general and assume that $n= d- 1$. For some absolute constants $C,C_0$ and any $u_1\ge 1\ge u_2>0$ and $u_3>0$, there exists an algorithm such that:
\begin{enumerate}
\item[$(i)$] The real complexity $\chi$ of the algorithm is bounded by
\begin{align}
\chi\le N u_3\big(u_1 C_0\pmax\sqrt{d\log(\pmax)} \big)^{d-1}\big(
\log u_1\vee \log(1/u_2) \vee d\log(\pmax) 
\big)
\, .\label{eq:BoundComplexity2}
\end{align}

\item[$(ii)$] The algorithm solves $\bF(\bx)=\bzero$ with probability at least
\begin{equation}\label{eq:MSS_complexitybound}
1-2\exp\Big(-\frac{1}{C}d(u_1-1)^2\Big)- C u_2 -C\frac{1}{u_3}\Big(
\log u_1\vee \log(1/u_2) \vee d\log(\pmax) 
\Big).
\end{equation}

\end{enumerate}
\end{theorem}

The two algorithms will be described in Sections \ref{subsec:HDintro} and \ref{subsec:Bruteforceintro} below and the theorems will be proved in Sections \ref{sec:HD} and \ref{sec:MMS}. The event $\calE$ on which the algorithms solve the random system $\bF(\bx)=\bzero$ will be described in Section \ref{sec:goodevent}.
 
To the best of our knowledge, the real case of Smale's problem is completely open, with a few  exceptions.
Namely,  \cite{rojas2005solving} develops an algorithm
that finds all the roots of a univariate polynomial ($d=1$) consisting
of 3 monomials; \cite{lasserre2008semidefinite} develops a semidefinite programming characterization which in general does not yield a polynomial-time algorithm;  \cite{ergur2023polyhedral} proposes
a polyhedral homotopy method for sparse polynomials which is not known to yiels a polynomial time algorithm.

\subsection{Approximate solutions}\label{subsec:apxsol}

Since we work on the unit sphere $\S^{d-1}:=\{\bx\in\R^d:\,\|\bx\|=1\}$, we define Newton's method by taking a step in an orthogonal direction and projecting back to the sphere. This definition is basically a variant of Newton's Method in the projective space, defined by Shub \cite{Shub1993}
for polynomial systems of $n=d-1$ equations in $d$ complex variables. Here we also allow for $n<d-1$.

Given a matrix $\bA\in \reals^{n_2\times n_1}$ 
or a linear operator $\bA:V_1\to V_2$, with $\dim(V_i) =n_i$
(we will often identify the two objects), for $n_1\le n_2$ we denote by 
$\sigma_{\min}(\bA) := \min\{\|\bA\bx\|:\, \|\bx\|=1\}$
the minimum singular value of $\bA$. For $n_1>n_2$, 
$\sigma_{\min}(\bA) :=\sigma_{\min}(\bA^{\sT})$.

We denote by $\bD\bF(\bx):=(\partial_jF_i(\bx))_{ij}\in\R^{n\times d}$ the differential or Jacobian matrix and by $\Ts_{\bx}\subset \R^d$ the orthogonal space to $\bx$.
Given a point $\bx\in\S^{d-1}$ we define Newton's operator as follows. The matrix $\bD\bF(\bx)$ defines a linear operator and its restriction to a subspace is defined as an operator in the obvious way.
If $\sigma_{\min}(\bD\bF(\bx)\big|_{\Ts_{\bx}})>0$, let $\bv\in\R^d$ be the unique solution to 
\begin{equation}
\label{eq:Newton}
\bF(\bx)+\bD\bF(\bx)\bv=\bzero \hspace{.4cm}\mbox{subject to}\hspace{.4cm}\bv\perp \bx,\,\bv\perp\ker(\bD\bF(\bx))\,,
\end{equation}
where $\ker(\bA):=\{\by: \bA\by=\bzero\}$, and define
\begin{equation}
\label{eq:Newton2}
\NM(\bx):=\frac{\bx+\bv}{\|\bx+\bv\|_2}\,.
\end{equation}
Otherwise, if $\sigma_{\min}(\bD\bF(\bx)\big|_{\Ts_{\bx}})=0$, arbitrarily define 
\[
\NM(\bx):=\bx\,.
\]

\begin{definition}[Approximate solution] \label{def:apxsol}
Let $\bx\in\S^{d-1}$ and define the sequence $\bx^0=\bx$, $\bx^{i+1}=\NM(\bx^i)$ by Newton's method. We say that $\bx$ is  an approximate solution of the system $\bF$ if there exists $\bz\in\S^{d-1}$ such that $\bF(\bz)=\bzero$ and $\|\bx^{i}-\bz\|_2\le 2^{1-2^i}\|\bx^0-\bz\|_2$ for all $i\ge0$.
\end{definition}


\subsection{\label{subsec:energy}Optimization}
Define the ``energy function''
\begin{equation}\label{eq:H}
H(\bx) := \frac{1}{2}\|\bF(\bx)\|^2\, .
\end{equation}
Of course, $\bF(\bx)=\bzero$ if and only if $H(\bx)=0$. In other words, the solutions of the polynomial system are exactly the minimizers of the energy function. This simple observation allows us to change perspective and try to algorithmically minimize a real-valued  function over the sphere $\S^{d-1}$, instead of working with a vector-valued function. 

Problems not very different from this have been studied recently in the context of spin glass models. In the terminology of spin glass theory, $F_i(\bx)$ as in \eqref{eq:Fi} is the energy function of the spherical $p$-spin model (with $p_i=p$), up to normalization. For the spin glass model, it is customary to consider $F_i(\bx)$ as a function on the sphere of radius $\sqrt{d}$ and multiply it by $d^{-(p-1)/2}$. More generally, for mixed $p$-spin the energy is defined as a linear combinations of the pure $p$-spin energies with deterministic coefficients. Given a model, those coefficients are fixed and one is interested in the $d\to\infty$ asymptotics. 

In \cite{subag2018following} an algorithm for minimizing the energy of a spherical mixed $p$-spin model by Hessian descent was proposed and analyzed. The algorithm constructs a discrete path $\bx_0,\ldots,\bx_k$ from the origin to the sphere with $\|\bx_i\|=\sqrt{i\cdot d/k}$, whose increments at each step $\bx_{i+1}-\bx_i$ are orthogonal to the position $\bx_i$, and are approximately in the span of the eigenvectors corresponding to the minimal eigenvalues of the Hessian of the energy. Related Approximate Message Passing algorithms for Ising models, which use the same energy functions but on the hyper-cube $\{+1,-1\}^d$, were proposed and studied in \cite{montanari2019optimization,el2021optimization}. Both types of algorithms \cite{subag2018following,montanari2019optimization,el2021optimization} were proved to be optimal within a certain large class of appropriately defined Lipschitz algorithms in \cite{huang2021tight}.

In light of the close connections to the spherical spin glass models, it is natural to wonder whether the Hessian descent algorithm of \cite{subag2018following} can be adapted to minimize $H(\bx)$. In an earlier work \cite{montanari2023solving} we studied this question for $F_i(\bx)$ that correspond  to general mixed models. 
That is, the $F_i(\bx)$ studied in \cite{montanari2023solving}
are given by  linear combinations of the functions as in \eqref{eq:Fi}. The functions $F_i(\bx)$ were also allowed to have a zero order term. Namely, one of the terms in the linear combination was given by $\gamma_0 W_i$ where $\gamma_0\ge0$ and $W_i\sim \normal(0,1)$. We assumed there that the functions $F_i(\bx)$ are i.i.d. and our goal was only to find points $\bx_*$ such that $H(\bx_*)=o(n)$ asymptotically, while the  typical value is $H(\bx)=\Theta(n)$ for most points on the sphere. 
The objective $H(\bx_*)=o(n)$ is better suited
for the methods developed for algorithmic optimization of spin glasses \cite{subag2018following,montanari2019optimization,el2021optimization}.
We also note that the same model was recently studied from a statistical physics viewpoint in \cite{fyodorov2022optimization,urbani2023continuous,
urbani2023dynamical}. Again, nonrigorous physics tools are better suited to characterize the asymptotics of $\min_{\bx\in\S^{d-1}}H(\bx)$ up to $o(n)$
corrections. 

In the present work, in Theorem \ref{thm:main1} we adapt the Hessian descent algorithm to the problem of finding approximate solutions.

\subsection{Hessian descent: large $d$ and moderate $\pmax\le d^2$
}
\label{subsec:HDintro}

In the present paper, the functions $F_i(\bx)$ are homogeneous,
and therefore changing the norm $\|\bx\|$ does not change the problem (apart from a scaling factor).
As mentioned above, the earlier papers \cite{subag2018following,montanari2023solving}
construct a path in the interior of the unit ball, hence
effectively changing the energy landscape as $\|\bx\|$ increases.
This is irrelevant in the homogeneous case. Instead, we may take an orthogonal step in $\R^d$ and project back to the unit sphere $\S^{d-1}$.

Moreover, here we wish to achieve energy values much smaller than  in the spin glass optimization problems, and we therefore have to take steps which depend on the value of the energy at each step and run the algorithm with a much larger number of steps. 
Below $C_1,c_0>0$ are  absolute constants which will be determined in Lemma \ref{rmk:MaxHessian} and the proof of Theorem \ref{thm:main1}.  The algorithm we use for moderate $\pmax$ ($\pmax\le d^2$) (and
in the proof of Theorem \ref{thm:main1}) is given in the pseudo-code
\ref{algo:HD} below. \vspace{.3cm}

\begin{algorithm}[H]
	\caption{Hessian descent}\label{algo:HD}
	\SetAlgoLined
	\KwIn{The coefficients $\ba:=\Big(\big(\frac{p_i!}{k_1!\cdots k_d!}\big)^{\frac12} a^{(i)}_{k_1,\ldots,k_d}\Big)_{i\le n,\,k_1+\cdots+k_d=p_i}$, number of iterations $k$}
	\KwOut{$\bx^{\sHD}\in \S^{d-1}$ or FALSE}
	Initialize $\bx_0=(1,0,\ldots,0)\in\S^{d-1}$\,\;
	\For{$i\in\{0,\ldots,k-1\}$}{
		Find   $\bv_i\perp \bx_i$, $\|\bv_i\|=1$, s.t.:  $\langle\nabla^2H(\bx_i),\bv_i^{\otimes2}\rangle\leq {\displaystyle\frac{1}{2}\min_{\bu\perp\bx_i,\|\bu\|=1}}\langle\nabla^2H(\bx_i),\bu^{\otimes2}\rangle$\,\;\vspace{.2cm}
		$\displaystyle \delta_i = \bigg( \frac{1}{30C_1}\frac{1}{\pmax^4\log(\pmax)} \frac{ \sqrt{ (d-n) H(\bx_i)}}{d}\wedge \frac1\pmax\bigg)^{1/2}$\,\;
        $s_i = \sign\left(H(\bx_i+ \delta_i\bv_i)- H(\bx_i- \delta_i\bv_i)\right) $\,\;
        $\by_{i+1} = \bx_i- s_i\delta_i\bv_i$\,\;
	$\bx_{i+1} = \frac{\by_{i+1}}{\|\by_{i+1}\|}$\,\;
        \If{$H(\bx_{i+1})\le \pmax^{-c_0d}$}{\KwRet $\bx^{\sHD}=\bx_{i+1}$\,\;}
	}
	\KwRet FALSE\,;
\end{algorithm}\vspace{.3cm}

Above we use the convention that $\sign(0)=1$.
\begin{remark}\label{rem:find}
In the pseudo-code above we allowed ourselves to be imprecise, to improve  readability. Specifically, we did not explain how we `find' the vector $\bv_i$ or when it is even possible. More accurately, our algorithm will use a sub-routine which we define in Section \ref{sec:subroutines}. On an event of overwhelmingly high probability, uniformly in $\bx_i$, the sub-routine outputs a vector $\bv_i$ as required.
\end{remark}

To explain the basic idea, for arbitrary $\bv$ we let $\by:=\bx_i+\bv$ and write the Taylor expansion 
\[
H(\by)=H(\bx_i)+\sum_{j=1}^3\frac{1}{j!}\big\<\bv^{\otimes j},\nabla^j H(\bx_i)\big\> + \frac{1}{4!} \big\<\bv^{\otimes 4},\nabla^4 H(\bz_i)\big\>\,,
\]
where $\bz_i$ is a point in the segment connecting $\by$ to $\bx_i$. 
In Lemma \ref{rmk:MaxHessian} we will derive global bounds on the derivatives of $H(\bx)$ of any order, which hold with overwhelmingly high probability. Given such bounds, assuming that $\|\bx-\bx_i\|$ is small enough the forth order error term can be made small. By flipping $\bv$ to $-\bv$, if needed, we can make sure that the contribution of the first and third order terms to the expansion is non-positive. This is essentially the role of the sign $s_i$ in the pseudo-code. Hence,  with the notation above, up to a small error, $H(\by_{i+1})-H(\bx_i)$ is bounded from above by
\begin{align*}
\frac12\langle\nabla^2H(\bx_i),\bv_i^{\otimes2}\rangle\leq \frac{1}{4}\min_{\bu\perp\bx_i,\|\bu\|=1}\langle\nabla^2H(\bx_i),\bu^{\otimes2}\rangle\, .
\end{align*}
Since $F_i(\bx)$ are homogeneous, projecting $\by_{i+1}$ to $\bx_{i+1}$ can only reduce the energy.

The algorithm of Theorem \ref{thm:main1} runs the Hessian descent algorithm for $k=C d^{3/2}\pmax^4\log(\pmax)^2$ iterations, for some absolute constant $C$. 
We will prove that the algorithm outputs an approximate solution 
when a certain `good' event $\cE$ holds. The event $\cE$
is defined in Section \ref{sec:goodevent} by: $(i)$~an upper bounds on the derivatives of $H(\bx)$ of different orders and certain Lipschitz constants uniformly over $\S^{d-1}$; $(ii)$~a uniform upper bound on the smallest eigenvalue of $\nabla^2H(\bx)$ restricted to the orthogonal space; and $(iii)$~a lower bound on the singular value of $\bD\bF(\bx)$ uniformly over the solutions of the system $\bF(\bx)=\bzero$. 

Using $(iii)$, we will show quantitatively in Section \ref{sec:Newton} that in order for $\bx$ to be an approximate solution, it is sufficient 
for it to be close enough to some exact solution $\balpha$. In our analysis of the Hessian descent algorithm in Theorem \ref{thm:Algorithm}, we will prove bounds for $H(\bx_i)$ and $\|\bx_i-\bx_j\|$,
which holds for any $i\ge 1$ 
(even if we produce an infinite sequence $\bx_i$ by running the Hessian descent indefinitely). These bounds imply that $\bx_{\infty}=\lim_{i\to\infty}\bx_i$ is well-defined and is a solution $\bF(\bx_{\infty})=\bzero$, and that for $k$ as above,  $\|\bx_k-\bx_{\infty}\|$ is small enough to conclude that $\bx_k$ is an approximate solution.

\subsection{\label{subsec:Bruteforceintro}Multi-Scale Search: large $d$ and $\pmax\ge d^2$ or bounded $\pmax\le p_0$, $d\leq d_0$}

Under certain regularity conditions on the system $\bF(\bx)$, (summarized in Definition \ref{def:good2}) in the proof of Theorem \ref{thm:main2} we shall see that to find an approximate solution or determine that no such exists, it is enough to check the values of $\|\bF(\bx)\|$ and $\sigma_{\min}(\bD\bF(\bx)\big|_{\Ts_{\bx}})$ over a sufficiently fine grid 
\begin{equation}\label{eq:Ndknet}
	\calN_{d,k} := \left\{\bx=(x_1,\ldots,x_d)\in[-1,1)^d:\, 2^k x_i\in \mathbb Z,\,\forall i\le d \right\}\,.
\end{equation}
Basically, any exact solution  in $\S^{d-1}$ can be detected at some point $\bx\in\calN_{d,k}$ close to it, whose projection to the sphere is in turn an approximate solution.
However, for large $\pmax$ or $d$  we must choose a very small step size $2^{-k}$, making the net size $|\calN_{d,k}|=(2/2^{-k})^d$ too large to employ exhaustive search over $\calN_{d,k}$ in polynomial time complexity.\footnote{Of course, this strategy can be made more efficient by working with a net on $\S^{d-1}$ directly. However, this improvement  only affects constants in the complexity analysis.}

In the pseudo-code below we propose a different Multi-Scale Search (MSS) Algorithm. Instead of scanning the whole grid, we start from a coarse grid and refine it gradually, keeping at each stage only the  blocks that have a chance of containing a solution. 

The algorithm searches for a solution in the block $\bx+[0,2^{-k}]^d$. To search for an approximate solution on the sphere $\S^{d-1}$, it needs to be initialized with $\bx=(-1,\ldots,-1)$ and $k=-1$ so that we search over $\bx+[0,2]^d=[-1,1]^d$. The correctness of the algorithm will be proved on a certain event $\calE$ (see Definition \ref{def:good2}). This event includes a bound on the Lipschitz constant of $\bx\mapsto \|\bF(\bx)\|$ on $\S^{d-1}$. The algorithm takes as input a parameter $\calL$ that should be chosen to match this bound. As in the pseudo-code, suppose that $D:=\mbox{diam}(\bx+[0,2^{-k}]^d)$, so that the IF condition $\|\bF(\hat\bx)\|>D\cdot \calL$ implies that $\|\bF(\bz)\|>0$ for all $\bz$ in $\big(\bx+[0,2^{-k}]^d\big)\cap \S^{d-1}$, the intersection of the block with the sphere. If instead the complementing condition $\|\bF(\hat\bx)\|\le D\cdot \calL$ holds, the block is divided to sub-blocks of half the side-length, and a search is performed within each block by a recursive call to Multi-Scale Search. The value of $k_0$ determines the finest step (or block) size that we consider $2^{-k_0}$ and must be large enough. Finally, $\calL$ and $\calS$ need be values such that last IF condition with $k=k_0$ implies that $\hat\bx$ is an approximate solution. A sufficient condition for this will be derived in Corollary \ref{cor:apxsol_largep}.

We finish with two technical remarks. First, in the
pseudo-code, $s_{\min}(\bD\bF(\hat\bx)\big|_{\Ts_{\hat\bx}})$ is an approximation of $\sigma_{\min}(\bD\bF(\hat\bx)\big|_{\Ts_{\hat\bx}})$ which will be made precise in Section \ref{sec:subroutines}, where we define a sub-routine that computes it.\footnote{Indeed, the minimal singular value $\sigma_{\min}(\bA)$ of a matrix $\bA$ is equal to the minimal solution of the characteristic polynomial of $\bA\bA^{\top}$. Hence, we can only approximate it.} Second, the definition of $\hat\bx$  in the pseudo-code is simply a way to choose a point in  $\big(\bx+[0,2^{-k}]^d\big)\cap \S^{d-1}$. Any other procedure that produces a point in the same set would result the correctness of the algorithm on $\calE$.

\vspace{.3cm}

\begin{algorithm}[H]
	\caption{MSS (Multi-Scale Search)}\label{algo:MMS}
	\SetAlgoLined
	\KwIn{The coefficients $\ba$, $\bx\in[-1,1]^{d}$, parameters $\calL,\calS>0$, $k,k_0\in\ZZ$}
	\KwOut{Approximate solution $\bx^{\salg}\in \S^{d-1}$, or FALSE if no solution exists (w.h.p.)}
	\If{$\big(\bx+[0,2^{-k}]^d\big)\cap \S^{d-1}=\varnothing$}{\KwRet FALSE\,\;}
	$D=\sqrt d 2^{-k}$\vspace{.1cm}\tcp*{diameter of the block $\bx+[0,2^{-k}]^d$}
	$\tilde\bx=\argmin\big\{\|\by\|:\,\by\in\bx+\{0,2^{-k}\}^d\big\}$\vspace{.1cm} \tcp*{corner of blcok closest to origin}
	$\hat\bx=\tilde\bx/\|\tilde\bx\|$\vspace{.1cm}\tcp*{its projection to $\S^{d-1}$}
	\uIf{$k< k_0$, $\|\bF(\hat\bx)\|>D\cdot \calL$}
	{\KwRet FALSE\,\;}
	\uElseIf{$k< k_0$, $\|\bF(\hat\bx)\|\le D\cdot \calL$}
	{
		\For(\tcp*[f]{subdivision to blocks of half side-length}){$\by\in \bx+\{0,2^{-(k+1)}\}^d$} 
		{RES = MSS($\ba,\by,\calL,\calS,k+1,k_0$)\tcp*{recursive calls}
			\If{{\rm RES} $\neq$ {\rm FALSE}}{\KwRet $\bx^{\salg}$=RES\,\;}}
		\KwRet FALSE\tcp*{non of the recursive calls is successful}
	}
	\uElseIf{$k=k_0$, $\|\bF(\hat\bx)\|\le D\cdot \calL$,  $s_{\min}(\bD\bF(\hat\bx)\big|_{\Ts_{\hat\bx}})\ge \calS$}
	{\KwRet $\bx^{\salg}=\hat\bx$\,\;}
	\Else{\KwRet FALSE\,\;}
	
\end{algorithm}\vspace{.2cm}

\subsection{Smale's 17th problem in the complex case}
\label{subsec:complexSmale17}

The study of Smale's 17th problem actually started in Shub and Smale's  `B\'{e}zout  series' \cite{ShubSmaleI,ShubSmaleII,ShubSmaleIII,ShubSmaleIV,ShubSmaleV}, several years before it was posed in \cite{smale1998mathematical}.
This sequence of papers studies a \emph{homotopy  continuation method}, which in turn builds
on numerical analysis approaches that were studied since the seventies \cite{Drexler1977,GarciaZangwill1979,keller1978global}. 
The basic idea of this approach is to define an interpolation $t\bF +(1-t)\bar\bF$, $t\in[0,1]$, 
between the system one wishes to solve $\bF$ and another system $\bar \bF$ of homogeneous polynomials
of the same degrees with a root $\bx_0$ which we know of. If the system $\bar \bF$ is not singular, 
then a.s. 
the number of roots in the projective space of each of the  systems $\bF=0$ and $\bar \bF=0$ is exactly  the B\'{e}zout number $\prod_{i=1}^{N-1} p_i$. On the unit sphere the number of solutions is double. 
These roots come in pairs $(\bx(0),\bx(1))$, where $\bx(0)$ solves
$\bar\bF(\bx(0))=0$ and $\bx(1)$ solves $\bF(\bx(1)) =0$. Each pair is connected by a smooth path 
$[0,1]\ni t\mapsto \bx(t)$ of solutions to $t\bF +(1-t)\bar\bF=0$.
Homotopy based algorithms start at time $t=0$ with $\bar \bF$ and $\bx(0)$ and iteratively increase the time parameter $t$ by a small step while trying to keep track of the root by an iteration of Newton's method, in order to end at time $t=1$ with $\bF$ and an approximate root $\bx_1$. The size of the step was analyzed in terms of a \emph{condition number} in the B\'{e}zout series and in  \cite{BeltranPardoFastLinear,BurgisserCucker,ShubVI,ArmentanoEtAl}. However, major difficulties  remained in how to choose the system $\bar\bF$ and root $\bx_0$ to start from. 

In the last paper in the  B\'{e}zout  series \cite{ShubSmaleV}, Shub and Smale proved that there exists a good system and a root
to start from, but their argument was not constructive and no practical way (an algorithm) to pick such a system and a root was provided. 
In major breakthroughs \cite{BeltranPardo2,BeltranPardo1}, Beltr\'{a}n and Pardo had an ingenious idea for how to choose a good system and a starting solution  at random, which led to a randomized polynomial time algorithm (on average). This solved Smale's problem, up to the point of using randomization as opposed to a deterministic algorithm. Another breakthrough was made by B\"{u}rgisser
and Cucker \cite{BurgisserCucker} who performed a smoothed analysis of the algorithm of Beltr\'{a}n and Pardo, leading to a deterministic algorithm of (almost polynomial) complexity $N^{O(\log\log N)}$. Finally, Lairez \cite{Lairez1} found a remarkable way to de-randomize the algorithm of Beltr\'{a}n and Pardo by using the random input itself to generate the initial pair. The combination of those works gave a solution to Smale's problem in the complex case. Important improvements to the latter algorithms were recently established in \cite{BurgisserCuckerLairez2023,Lairez2}.

\section{Preliminaries}
\label{sec:Preliminaries}
For any $\ell$, let $\bG^{(\ell)}= (G^{(\ell)}_{i_1,\dots,i_{p_{\ell}}})_{i_1,\dots,i_{p_{\ell}}\le d}$
be a tensor with i.i.d. entries $G^{(\ell)}_{i_1,\dots,i_{p_{\ell}}}\sim\normal(0,1)$. We will also assume that $\bG^{(\ell)}$ are independent for different values of $\ell$. We define 
$\obG^{(\ell)} =(\overline{G}^{(\ell)}_{i_1,\dots,i_{p_{\ell}}})_{i_1,\dots,i_{p_{\ell}}\le d}$ the symmetrization of $\bG^{(\ell)}$ by 
\begin{equation}\label{eq:symmetrization}
\overline{G}^{(\ell)}_{i_1,\dots,i_{p_{\ell}}}:=(p_\ell!)^{-1}\sum_{\pi\in S_{p_\ell}}
G^{(\ell)}_{i_{\pi(1)},\dots,i_{\pi(p_\ell)}}\,.
\end{equation}
Observe that
\begin{equation}\label{eq:coefficients}
\Big(\frac{p_\ell!}{k_1!\cdots k_d!}\Big)^{\frac12}a^{(\ell)}_{k_1,\ldots,k_d}
\mbox{\quad and\quad}
\sum_{\cJ(k_1,\dots,k_d)} \overline{G}^{(\ell)}_{i_1,\dots,i_{p_{\ell}}}=\sum_{\cJ(k_1,\dots,k_d)}
G^{(\ell)}_{i_1,\dots,i_{p_{\ell}}}
\end{equation}
have the same law,
where 
$\cJ(k_1,\dots,k_d)$ is the set of
 indices $i_1,\ldots,i_{p_\ell}$
such that $x_{i_1}\cdots x_{i_{p_\ell}}=x_1^{k_1}\cdots x_d^{k_d}$. Hence we can assume that they are defined on the same probability space such that both sides of \eqref{eq:coefficients} are equal and thus
\begin{align}
F_{\ell}(\bx) = \<\bG^{(\ell)},\bx^{\otimes p_{\ell}}\>= \<\obG^{(\ell)},\bx^{\otimes p_{\ell}}\>\, .
\end{align}
We will often work with this representation for the random polynomials. As the representation in \eqref{eq:Fi}  requires less variables to describe $F_\ell(\bx)$, it is used to define the input in the algorithmic problem. Yet another equivalent way to describe the polynomials is through their covariance function
\begin{equation}\label{eq:Fcov}
	\E\big[F_i(\bx^1)F_j(\bx^2)\big] = \delta_{ij}\xi_i\big(\langle \bx^1,\bx^2\rangle\big)\, ,\quad\mbox{where }\xi_i(t):=t^{p_i}\,.
\end{equation} 

Recall that we denote by $\bD\bF(\bx):=(\partial_jF_i(\bx))_{ij}\in\R^{n\times d}$ the differential or Jacobian matrix and by $\Ts_{\bx}\subset \R^d$ the orthogonal space to $\bx$, which we identify with the tangent space to the sphere of radius 
$\sqrt{q}$, $\S^{d-1}(\sqrt{q})$ at $\bx$, when $\|\bx\|_2^2=q> 0$. 
We often view $\nabla^kF_1(\bx)= (\partial_{i_1,\dots,i_k}F_1(\bx))_{i_1\dots i_k\le d}$ 
as a $k$-th order tensor $\nabla^kF_1(\bx)\in(\reals^d)^{\otimes k}$ and
 $\nabla^k\bF(\bx)= (\partial_{i_1,\dots,i_k}F_{\ell}(\bx))_{\ell\le n,i_1\dots i_k\le d}$
 as a $(k+1)$-th order tensor $\nabla^k\bF(\bx)\in\reals^n\otimes (\reals^d)^{\otimes k}$.
 The operator norm of a tensor $\bT\in \R^{n_1}\otimes \cdots\otimes \R^{n_k}$ is
 \begin{align}
 \|\bT\|_{\op} :=\max_{\bv_1\in\S^{n_1-1}}\cdots\max_{\bv_k\in\S^{n_k-1}}
 \<\bT,\bv_1\otimes\cdots\otimes \bv_k\>\, .
 \end{align}
Finally, we denote by $\Ball^d(r):=\{\bx\in\R^d:\,\|\bx\|\leq r\}$ the ball of radius $r$, and let
$\Ball^d(r_1,r_2):=\Ball^d(r_2)\setminus\Ball^d(r_1)$.

Throughout, we will write $\bW\sim \GOE(N)$ if $\bW=\bW^{\sT}$ and $(W_{ij})_{i\le j\le N}$
are independent with $W_{ii}\sim\normal(0,2)$,  $W_{ij}\sim\normal(0,1)$ for $i<j$.
We write   $\bZ\sim \GOE(M,N)$ if $(Z_{ij})_{i\le M,j\le N}$
are independent with $Z_{ij}\sim\normal(0,1)$.
It is useful to recall that (under the present setting) $F_\ell(\bx)$ and its derivatives are jointly Gaussian and 
\begin{equation}
    \E\Big\{ \frac{\partial}{\partial x_{i_1}}\cdots\frac{\partial}{\partial x_{i_k}}F_\ell(\bx) \frac{\partial}{\partial y_{j_1}}\cdots\frac{\partial}{\partial y_{j_m}}F_\ell(\by)\Big\} 
    = 
    \frac{\partial}{\partial x_{i_1}}\cdots\frac{\partial}{\partial x_{i_k}} \frac{\partial}{\partial y_{j_1}}\cdots\frac{\partial}{\partial y_{j_m}}\E\big\{ F_\ell(\bx)F_\ell(\by)\big\}\,.
\end{equation}

Our first preliminary result provides  bounds on the norm of $\bF(\bx)$ and its derivatives.

\begin{lemma}\label{rmk:MaxHessian}
Assume that $n\le d$.
 Then, for any $k_{\max}\ge 2$ there exist constants $C_*, C_1$
 depending on $k_{\max}$, such that, with probability at least
 $1-C_*\exp(-d/C_*)$
\begin{align}
\max_{\bx\in\Ball^d(1)}\|\bF(\bx)\|_{2}&\le C_1\sqrt{d\log(\pmax)}\, ,\label{eq:MaxF}\\
\max_{\bx\in\Ball^d(1)}\|\bD \bF(\bx)\|_{\op}& \le C_1\pmax
\sqrt{d\log(\pmax)}\, ,\label{eq:MaxDF}\\
\max_{\bx\in\Ball^d(1)}
\big\|\nabla^k\bF(\bx)\big\|_{\op}& \le C_1\pmax^{k}
\sqrt{d\log\pmax}\,, \quad\forall 2\le k\leq k_{\max}\,,\label{eq:MaxDkF}
\\
\max_{\bx\in\Ball^d(1)}
\big\|\nabla^k H(\bx)\big\|_{\op}& \le C_1\pmax^{k}
d\log \pmax\,. \quad\forall 1\le k\leq k_{\max}\,.\label{eq:MaxDDH}
\end{align}
\end{lemma}
In fact, it will be clear from the proof that the last two inequalities follow for $2\le k\le K$ with general $K$, if one allows the constants $C_*, C_1$ to depend on $K$.
\begin{proof}
It will be enough to prove a bound on the probability for each of the 
inequalities \eqref{eq:MaxF}-\eqref{eq:MaxDDH} separately, and conclude the lemma by a union bound. 

\vspace{0.2cm}

\noindent{\bf Proof of Eq.~\eqref{eq:MaxF}.} Consider the following Gaussian process indexed by 
$\bu\in\S^{n-1}$, $\bx\in \S^{d-1}$:
\begin{align}
 Z_0(\bu,\bx):= \<\bu, \bF(\bx)\>\, ,
 \end{align}
 and notice that, using homogeneity, 
 \begin{align}
\max_{\bx\in\Ball^d(1)}\|\bF(\bx)\|_{2}  = \max_{\bu\in\S^{n-1},\bx\in\S^{d-1}} Z_0(\bu,\bx)\, .
\end{align}
With $\be_1=(1,0,\ldots,0)$ denoting the standard basis element, define the spherical cap $\Cap_d(\pi/4):=\{\bx\in\S^{d-1}: \<\bx,\be_1\>\ge 1/\sqrt{2}\}$ and note that for any $\bx_1,\bx_2\in\Cap_d(\pi/4)$, $\<\bx_1,\bx_2\>\ge 0$. We claim that to prove \eqref{eq:MaxF}, it is enough to prove that for some absolute $C_1, C_2, C_3>0$,
\begin{align}\label{eq:M0}
\P\left(M_0:=\max_{\bu\in\S^{n-1},\bx\in\Cap_d(\pi/4)} Z_0(\bu,\bx)>C_1\pmax\sqrt{d\log(\pmax)}\right)<C_2\exp(-C_3 d/C_*)\, .   
\end{align}
Indeed, by the Borell-TIS inequality, by increasing $C_1$ we can make sure that the same bound holds with $C_3$ is as large as we wish. Thus, noting that
$\S^{d-1}$ can be covered by $e^{C'd}$ rotations of $\Cap_d(\pi/4)$ for some absolute constant $C'>0$,  by rotational invariance and a union bound we obtain the required bound with $\bx$ maximized over $\S^{d-1}$ instead of $\Cap_d(\pi/4)$.

In order to prove Eq.~\eqref{eq:M0}, we compute the canonical distance of this process to get (using $\xi_i(1)=1$)
\begin{align*}
 d_0(\bu_1,\bx_1;\bu_2, \bx_2)^2 &:=\E\left\{ \big(\<\bu_1, \bF(\bx_1)\> - \<\bu_2, \bF(\bx_2)\>\big)^2 \right\}\\
 &=
 \sum_{i=1}^n\left\{u_{1,i}^2-2u_{1,i}u_{2,i}\,\xi_i(\<\bx_1,\bx_2\>)
 +u_{2,i}^2\right\}\\
 & =\frac{1}{2}\sum_{i=1}^n(u_{1,i}^2+u_{2,i}^2)\big(2-2\xi_i(\<\bx_1,\bx_2\>)\big)
 +\sum_{i=1}^n(u_{1,i}-u_{2,i})^2\xi_i(\<\bx_1,\bx_2\>)\\
 &\le \max_{i\le n}\big(2-2\xi_i(\<\bx_1,\bx_2\>) \big) + \|\bu_1-\bu_2\|^2_2\\
 &\le \xi_{\min}(\|\bx_1\|^2)-2\xi_{\min}(\<\bx_1,\bx_2\>)+\xi_{\min}(\|\bx_2\|^2)
 +\|\bu_1-\bu_2\|^2_2\, ,
\end{align*}
where $\xi_{\min}(q):=\min_{i\le n}\xi_i(q)=q^{\pmax}$ and we assume that $\bx_1,\bx_2\in\Cap_d(\pi/4)$ so that $\<\bx_1,\bx_2\>\ge0$.
We recognize that the right-hand side is the canonical distance of the process
$\oZ(\bu,\bx):= \<\bg,\bu\>+F_{\max}(\bx)$, where $\bg\sim\normal(0,\id_n)$ and $F_{\max}$ is the centered Gaussian process 
with covariance $\<\bx_1,\bx_2\>^{\pmax}$.  

Using \cite[Proposition A1]{montanari2023solving}, we get
\begin{align}
\E\max_{\bu\in\S^{n-1}\bx\in\S^{d-1}} \oZ(\bu,\bx) \le 
\E\|\bg\|+ \E\max_{\bx\in\S^{d-1}} F_{\max}(\bx)\le C\sqrt{n}+C\sqrt{d\log \pmax}\, .
\end{align}
Hence, for $M_0$ as in \eqref{eq:M0}, using the Sudakov-Fernique inequality we get, for $n\le d$,
\begin{align}\label{eq:EmaxF}
\E M_0\le C\sqrt{d\log \pmax}\, .
\end{align}
On the other hand,  $Z_0(\bu,\bx)$ is a Lipschitz function of $\bG$, hence so is
$M_0$, whence, by Gaussian concentration 
\begin{align}
\prob\big(M_0\ge \E M_0+ t\big) \le e^{-t^2/2}\, .
\end{align}
Equation \eqref{eq:M0}, and thus \eqref{eq:MaxF}, follow from the last two displays.

\vspace{0.2cm}

\noindent{\bf Proof of Eq.~\eqref{eq:MaxDF}.} 
We have
\begin{align}
\max_{\bx\in \Ball^d(1)}\|\bD \bF(\bx)\|_{\op}&\stackrel{(a)}{\le} 
\max_{\bx\in \Ball^d(1)}\max_{\bv\in\S^{d-1}:\<\bv,\bx\>=0}\|\bD \bF(\bx)\bv\|
+\max_{\bx\in \Ball^d(1)\setminus\bfzero}\frac{\|\bD \bF(\bx)\bx\|}{\|\bx\|}\nonumber\\
& \stackrel{(b)}{\le} \max_{\bx,\bv\in\S^{d-1}:\<\bv,\bx\>=0}\|\bD \bF(\bx)\bv\|
+\max_{\bx\in \S^{d-1}}\|\bD \bF(\bx)\bx\|\nonumber\\
& \stackrel{(c)}{\le}  \max_{(\bu,\bx,\bv)\in\cA_{n,d}}Z_{1}(\bu,\bx,\bv)
+\max_{\bx\in \S^{d-1}}\|\bD \bF(\bx)\bx\|\, ,\label{eq:BoundDF1}
\end{align}
where $(a)$ follows by triangle inequality, $(b)$ by homogeneity, and in
$(c)$ we defined:
\begin{align}
\cA_{n,d}& :=\big\{(\bu,\bx,\bv)\in\S^{n-1}\times (\S^{d-1})^2\; :\;\;\<\bv,\bx\>=0\big\}\, ,\\
 Z_1(\bu,\bx,\bv)&:= \<\bu, \bD\bF(\bx)\bv\>\, .
 \end{align}
 The second term in the upper bound \eqref{eq:BoundDF1} is easily treated since
 \begin{align*}
 \big(\bD \bF(\bx)\bx\big)_{i} = p_i\<\bG^{(i)},\bx^{\otimes p_i}\> = p_i\,  F_i(\bx)\, ,
 \end{align*}
 and therefore 
 \begin{align}
\big\|\bD \bF(\bx)\bx\big\| \le \pmax \|\bF(\bx)\|\, .
 \label{eq:BoundDFRadial}
 \end{align}
 It is therefore sufficient to bound the first term in Eq.~\eqref{eq:BoundDF1}.

 We next compute the canonical distance associated to the process $Z_1$.
 For $(\bu_s,\bx_s,\bv_s)\in\cA_{n,d}$, $s\in\{1,2\}$, we get
 \begin{align}
 d_1(\bu_1,\bv_1,\bx_1;\bu_2,\bv_2, \bx_2)^2 &= \Delta_1(\bu_1,\bv_1,\bx_1;\bu_2,\bv_2, \bx_2)+
 \Delta_2(\bu_1,\bv_1,\bx_1;\bu_2,\bv_2, \bx_2)\, ,\\
 \Delta_1(\bu_1,\bv_1,\bx_1;\bu_2,\bv_2, \bx_2)&:= 
 -2\<\bx_1,\bv_2\>\<\bx_2,\bv_1\>\sum_{i=1}^n
 u_{1,i}u_{2,i}\xi_i''(\<\bx_1,\bx_2\>)\, ,\nonumber\\
 \Delta_2(\bu_1,\bv_1,\bx_1;\bu_2,\bv_2, \bx_2)&:= \sum_{i=1}^n
 u_{1,i}^2\xi_i'(1)+\sum_{i=1}^n
 u_{2,i}^2\xi_i'(1)
 -2\<\bv_1,\bv_2\>\sum_{i=1}^n
 u_{1,i}u_{2,i}\xi_i'(\<\bx_1,\bx_1\>)\, .
 \end{align}
 Denoting by $\bP^{\perp}_s$ the projector orthogonal to $\bx_s$, we have
\begin{equation}
\label{eq:Delta1bd}    
 \begin{aligned}
 \Delta_1(\bu_1,\bv_1,\bx_1;\bu_2,\bv_2, \bx_2)&\le 2\, \max_{i\le n}\xi_i''(1)\|\bP^{\perp}_1\bx_2\|\,\|\bP^{\perp}_2\bx_1\|\\
 & = 2 \pmax(\pmax-1) \big(1-\<\bx_1,\bx_2\>^2\big)\\
 & \le 2 \pmax^2\|\bx_1-\bx_2\|^2\, .
 \end{aligned}
\end{equation}
 On the other hand, 
 \begin{align*}
 \Delta_2(\bu_1,\bv_1,\bx_1;\bu_2,\bv_2, \bx_2)&\le \sum_{i=1}^n
 u_{1,i}^2\xi_i'(1)+\sum_{i=1}^n
 u_{2,i}^2\xi_i'(1)
 -2\sum_{i=1}^n
 u_{1,i}u_{2,i}\xi_i'(\<\bx_1,\bx_1\>)
 +\pmax\|\bv_1-\bv_2\|^2\, .
 \end{align*}
 Putting the above bounds together and using Sudakov-Fernique inequality, 
 we see that 
 \begin{align}\label{eq:maxZ1}
     \E\max_{(\bu,\bx,\bv)\in\cA_{n,d}}Z_{1}(\bu,\bx,\bv) &
     \le \E\max_{(\bu,\bx,\bv)\in\cA_{n,d}}\big\{\sqrt{2}\pmax\<\bg,\bx\>+\sqrt{\pmax}
     \<\bh,\bv\>+
     \<\bu,\hbF(\bx)\>\big\}\, ,
 \end{align}
 where $\bg$, $\bh$ are mutually independent standard normal vectors, independent of 
 $\hbF$, a centered Gaussian process (taking values in $\reals^n$) 
 with covariance $\E\{\hF_i(\bx_1)\hF_j(\bx_2)\} = \delta_{ij}\xi_i'(\<\bx_1,\bx_2\>)$.
 By the same argument used for $\bF$, we have $\E\max_{\bx\in\S^{d-1}}\|\hbF(\bx)\|\le
 C\pmax\sqrt{d\log \pmax}$.
 We then have
 \begin{align*}
     \E\max_{(\bu,\bx,\bv)\in\cA_{n,d}}Z_{1}(\bu,\bx,\bv) & \le C\pmax\sqrt{d}
     +\E\max_{\bx\in\S^{d-1}}\|\hbF(\bx)\| \le  C\pmax\sqrt{d\log \pmax}\, .
 \end{align*}
Using this bound together with \eqref{eq:BoundDFRadial},
we thus obtain
\begin{align}
\E\|\bD\bF(\bx)\|_{\op} \le   C\pmax\sqrt{d\log \pmax}\, .
\end{align}
The tail probability is controlled by Gaussian concentration as in the case of 
$\max_{\bx\in\S^{d-1}}\|\bF(\bx)\|$ treated above.

\vspace{0.2cm}

\noindent{\bf Proof of Eq.~\eqref{eq:MaxDkF}.} To generalize the argument
at the previous point, we define the process
\begin{align}\label{eq:Zk_Def}
    Z_k(\bu,\bx,\bv)&:= \sum_{i=1}^n u_i\< \nabla^k F_i(\bx),\bv^{\otimes k}\>\, .
\end{align}
The quantity we want to upper bound is
\begin{align}
\max_{\bx\in\Ball^d(1)}\big\|\nabla^k \bF(\bx)\big\|_{\op} =
\max_{\bx\in\Ball^d(1)}\max_{\bu\in\S^{n-1}}\max_{\bv\in\S^{d-1}}  Z_k(\bu,\bx,\bv)\, .
\end{align}
Note that we can always decompose $\bv = a\bx+\bv_{\perp}$ where $\<\bv_{\perp},\bx\>=0$,
and therefore\footnote{Given symmetric tensors $\bA\in (\reals^n)^{\otimes k}$,
$\bB\in (\reals^n)^{\otimes (k-j)}$, we denote by $\bT = \bA\{\bB\}$ the tensor 
with components $T_{i_1\dots i_j} = \sum_{i_{j+1}\dots i_k} A_{i_1\dots i_k}B_{i_{j+1}\dots i_k}$.}
\begin{align}
  Z_k(\bu,\bx,\bv)&:= \sum_{j=0}^k
  \binom{k}{j}a^{k-j}\sum_{i=1}^n u_i\< \nabla^k F_i(\bx)\{\bx^{k-j}\},\bv_{\perp}^{\otimes j}\>\, .
\end{align}
By homogeneity, $\nabla^k F_i(\bx)\{\bx^{k-j}\} = (p_i-j)\cdots (p_{i}-k+1)
\nabla^j F_i(\bx)$, and therefore all terms except the one with $j=k$ can be controlled by
bounds on lower order derivatives $k$.
We therefore need to bound 
$\E\max_{(\bu,\bx,\bv)\in\cA_{n,d}}Z_{k}(\bu,\bx,\bv)$,
i.e. to consider the case $\<\bv,\bx\>=0$. 

Define the covaraince function
\begin{align*}
   & \bC_k(\bu_1,\bv_1,\bx_1;\bu_2,\bv_2, \bx_2):=\E\big\{
    Z_{k}(\bu_1,\bx_1,\bv_1)Z_{k}(\bu_2,\bx_2,\bv_2)\big\} 
    =\sum_{i=1}^n u_{1,i}u_{2,i} \partial_{\bv_1}^k\partial_{\bv_2}^k\xi_i(\<\bx_1,\bx_2\>) \\
    &\quad= \sum_{j_1,j_2} \hat a_{j_1,j_2,k} \sum_{i=1}^n u_{1,i}u_{2,i} \xi_i^{(j_1+j_2)}(\<\bx_1,\bx_2\>) \<\bv_1,\bx_2\>^{j_1+j_2-k}\<\bv_2,\bx_1\>^{j_1+j_2-k}\<\bv_1,\bv_2\>^{2k-j_1-j_2}\\
    &\quad= \sum_{j} a_{j,k} \sum_{i=1}^n u_{1,i}u_{2,i} \xi_i^{(j)}(\<\bx_1,\bx_2\>) \<\bv_1,\bx_2\>^{j-k}\<\bv_2,\bx_1\>^{j-k}\<\bv_1,\bv_2\>^{2k-j}
    \,,
\end{align*}
where $\partial_{\bv_i}$ denotes the directional derivative corresponding to $\bv_i$ w.r.t. the coordinates of $\bx_i$, $\hat a_{j_1,j_2,k}$ and $a_{j,k}$ are combinatorial factors and the summation is over $j$ or $j_1,j_2$ such that all the exponents above are non-negative.
Define  $\hat\bC_k(\bu_1,\bv_1,\bx_1;\bu_2,\bv_2, \bx_2)$ by the same expression as above with summation only over $j\neq k$  and note that, for $(\bu,\bx,\bv)\in\cA_{n,d}$, we have that
\begin{align*}
   \hat \bC_k(\bu,\bv,\bx;\bu,\bv, \bx) &=0,\\
   \bC_k(\bu,\bv,\bx;\bu,\bv, \bx)&=
       a_{k,k} \sum_{i=1}^n u_{i}^2 \xi_i^{(k)}(1) \,.
\end{align*}
Thus,   for $(\bu_i,\bx_i,\bv_i)\in\cA_{n,d}$,
 the squared canonical distance of $Z_{k}(\bu,\bx,\bv)$ is given  by
\begin{align}
   &d_k(\bu_1,\bv_1,\bx_1;\bu_2,\bv_2, \bx_2)^2 \nonumber\\
   &\quad=\bC_k(\bu_1,\bv_1,\bx_1;\bu_1,\bv_1, \bx_1)-2\bC_k(\bu_1,\bv_1,\bx_1;\bu_2,\bv_2, \bx_2)+\bC_k(\bu_2,\bv_2,\bx_2;\bu_2,\bv_2, \bx_2)
   \nonumber\\
   &\quad=\hat\bC_k(\bu_1,\bv_1,\bx_1;\bu_1,\bv_1, \bx_1)-2\hat\bC_k(\bu_1,\bv_1,\bx_1;\bu_2,\bv_2, \bx_2)+\hat\bC_k(\bu_2,\bv_2,\bx_2;\bu_2,\bv_2, \bx_2) \label{eq:Chat}\\
    &\quad\; +  a_{k,k} \Big( \sum_{i=1}^n u_{1,i}^2 \xi_i^{(k)}(1) +\sum_{i=1}^n u_{2,i}^2 \xi_i^{(k)}(1)-2\sum_{i=1}^n u_{1,i}u_{2,i} \xi_i^{(k)}(\<\bx_1,\bx_2\>) \<\bv_1,\bv_2\>^{k}\Big)
    \,. \label{eq:ahat}
\end{align}

The expression in \eqref{eq:Chat} is bounded by 
\begin{align*}
    A_k \max_{i\le n,j\le k}\xi_i^{(2j)}(1) 
    |\<\bv_1,\bx_2\>\<\bv_2,\bx_1\>| \le A_k \pmax^{2k} \|\bx_1-\bx_2\|^2\,,
\end{align*}
where $A_k$ is a combinatorial factor and we used the same bound as in \eqref{eq:Delta1bd}.
On the other hand, assuming that $\<\bv_1,\bv_2\>\ge0$, the expression in \eqref{eq:ahat} is bounded by
\begin{align*}
    a_{k,k} \Big( \sum_{i=1}^n u_{1,i}^2 \xi_i^{(k)}(1) +\sum_{i=1}^n u_{2,i}^2 \xi_i^{(k)}(1)-2\sum_{i=1}^n u_{1,i}u_{2,i} \xi_i^{(k)}(\<\bx_1,\bx_2\>) \Big) + k\pmax^k \|\bv_1-\bv_2\|^{2}\,.
\end{align*}
Hence, by the Sudakov-Fernique inequality
\begin{align*}
\E\max_{(\bu,\bx,\bv)\in\cA_{n,d}:\bv\in\Cap_d(\pi/4)} Z_k(\bu,\bx,\bv)\le \E\max_{(\bu,\bx,\bv)\in\cA_{n,d}}\big\{\sqrt{A_k}\pmax^k\<\bg,\bx\>+\sqrt{k\pmax^k}
     \<\bh,\bv\>+
     \<\bu,\hbF(\bx)\>\big\}
\, ,
\end{align*}
where $\bg$, $\bh$ are mutually independent standard normal vectors, independent of 
 $\hbF$, a centered Gaussian process (taking values in $\reals^n$) 
 with covariance $\E\{\hF_i(\bx_1)\hF_j(\bx_2)\} = \delta_{ij}\xi_i^{(k)}(\<\bx_1,\bx_2\>)$.
The proof of Eq.~\eqref{eq:MaxDkF} thus follows by a similar argument to that following \eqref{eq:maxZ1} and the argument used in the proof of Eq.~\eqref{eq:MaxF} to move from $\Cap_d(\pi/4)$ to the whole sphere.

 \vspace{0.2cm}

\noindent{\bf Proof of Eq.~\eqref{eq:MaxDDH}.}
We use
\[
\big\|\nabla^k H(\bx)\big\|_{\op} \le 2^k \max_{0\le i\le k} \big\|\nabla^i \bF(\bx)\big\|_{\op}\big\|\nabla^{k-i} \bF(\bx)\big\|_{\op}\,, 
\] 
and therefore this claim follows from the previous ones.
\end{proof}

For any $\bx\in \reals^d\setminus\{0\}$, we let  $\bU_{\bx}\in\reals^{d\times(d-1)}$ be an arbitray matrix whose columns form a basis 
of $\Ts_{\bx}$.
For $\bx_1,\bx_2\in \reals^d\setminus\{0\}$, define 
$\bU_{\bx_1,\bx_2}:= \bR_{\bx_1,\bx_2}\bU_{\bx_1}$, where 
$\bR_{\bx_1,\bx_2}$ is the rotation that keeps unchanged the space orthogonal to $\bx_1$, $\bx_2$,
and maps $\bx_1/\|\bx_1\|_2$ to $\bx_2/\|\bx_2\|_2$. We will need the following geometric fact.
\begin{lemma}\label{lemma:Udiff}
For any non-zero $\bx_1,\bx_2\in\Ball^d(1)$, we have
\begin{align}
\big\|\bU_{\bx_1,\bx_2}-\bU_{\bx_1}\|_{\op}\le \frac{\|\bx_1-\bx_2\|_2}{\|\bx_1\|_2\wedge \|\bx_2\|_2}\, .
\end{align}
\end{lemma}
\begin{proof}
Note that $\|(\bR_{\bx_1,\bx_2}-\id)\bv\|_2$ is constant for any unit vector $\bv\in
{\rm span}(\bx_1,\bx_2)$, and vanishes if $\bv \perp
{\rm span}(\bx_1,\bx_2)$.
By using $\bv=\bv_1:=\bx_1/\|\bx_1\|_2$, 
we then have
\begin{align*}
\big\|\bU_{\bx_1,\bx_2}-\bU_{\bx_1}\|_{\op}& \le \|\bR_{\bx_1,\bx_2}-\id\|_{\op}\\
& = \|(\bR_{\bx_1,\bx_2}-\id)\bv_1\|_{2}\\
& = \Big\|\frac{\bx_1}{\|\bx_1\|}-\frac{\bx_2}{\|\bx_2\|}\Big\|_2\\
& \le \frac{\|\bx_1-\bx_2\|_2}{\|\bx_1\|_2\wedge \|\bx_2\|_2}\, .
\end{align*}
\end{proof}

\begin{definition}\label{def:Lip}
For $\Omega\subseteq \reals^d$, we
define the following Lipschitz constants:
\begin{align}
\Lip(\bF;\Omega) &:= \sup_{\bx_1\neq\bx_2\in\Omega}\frac{\|\bF(\bx_1)-\bF(\bx_2)\|}{\|\bx_1-\bx_2\|_2}\, ,\\
\Lip(\bD\bF;\Omega)& :=\sup_{\bx_1\neq\bx_2\in\Omega}
\frac{\|\bD\bF(\bx_1)-\bD\bF(\bx_2)\|_{\op}}{\|\bx_1-\bx_2\|_2}\, ,\label{eq:LipD}\\
\Lip(\nabla^2\bF;\Omega)& :=\sup_{\bx_1\neq\bx_2\in\Omega}
\max_{\ell\le n}
\frac{\|\nabla^2 F_{\ell}(\bx_1)-\nabla^2 F_{\ell}(\bx_2)\|_{\op}}{\|\bx_1-\bx_2\|_2}\, .
\end{align}
We also define the following Lipschitz constants for projections onto the tangent space:
\begin{align}
\Lip_{\perp}(\bD\bF;\Omega)& :=\sup_{\bx_1\neq\bx_2\in\Omega}
\frac{\|\bD\bF(\bx_1)\bU_{x_1}-\bD\bF(\bx_2)\bU_{\bx_1,\bx_2}\|_{\op}}{\|\bx_1-\bx_2\|_2}\, ,\label{eq:LipDF_perp}\\
\Lip_{\perp}(\nabla^2\bF;\Omega)& :=\sup_{\bx_1\neq\bx_2\in\Omega}
\max_{\ell\le n}
\frac{\|\bU_{\bx_1}^{\sT}\nabla^2 F_{\ell}(\bx_1)\bU_{\bx_1}-
\bU_{\bx_1,\bx_2}^{\sT}\nabla^2 F_{\ell}(\bx_2)\bU_{\bx_1,\bx_2}\|_{\op}}{\|\bx_1-\bx_2\|_2}\, .
\end{align}
\end{definition}

The next lemma is an immediate consequence of Lemma \ref{rmk:MaxHessian} and Lemma \ref{lemma:Udiff}.
\begin{lemma}\label{lemma:BoundLip}
Assume that $n\le d$.
Then there exist absolute constants $C_0$, $C_*$
such that the following hold with probability at least $1-C_*\exp(-d/C_*)$:
\begin{align}
\Lip(\bF;\Ball^d(1)) & \le C_0 \pmax \sqrt{d\log(\pmax)} \, ,
\label{eq:LipF}\\ 
\Lip(\bD\bF;\Ball^d(1))&  \le C_0\pmax^2\sqrt{d\log(\pmax)}
\, ,
\quad \Lip_{\perp}(\bD\bF;\Ball^d(\rho,1))  \le \frac{C_0}{\rho}\pmax^2\sqrt{d\log(\pmax)}\, ,
\label{eq:LipDf}\\ 
\Lip(\nabla^2\bF;\Ball^d(1))& \le C_0\pmax^{3}\sqrt{d\log(\pmax)}\, ,\quad
 \Lip_{\perp}(\nabla^2\bF;\Ball^d(\rho,1))
\le \frac{C_0}{\rho}\pmax^{3}\sqrt{d\log(\pmax)}\, .
\label{eq:LipDDf}
\end{align}
\end{lemma}
\begin{proof}
The bounds on $\Lip(\bF;\Ball^d(1))$, $\Lip(\bD\bF;\Ball^d(1))$,
$\Lip(\nabla^2\bF;\Ball^d(1))$ follow immediately from Lemma \ref{rmk:MaxHessian}.
The bounds on $\Lip_{\perp}(\bD\bF;\Ball^d(1))$, $\Lip_{\perp}(\nabla^2\bF;\Ball^d(1))$
are proved similarly and we limit ourselves to the last one.
Writing $\bU_1:=\bU_{\bx_1}$, $\bU_2:=\bU_{\bx_1,\bx_2}$, and 
assuming without loss of generality $\|\bx_1\|_2\le \|\bx_2\|_2$,
 \begin{align}
 \max_{\ell\le n}\|\bU_{1}^{\sT}\nabla^2 &F_{\ell}(\bx_1)\bU_{1}-
\bU_{2}^{\sT}\nabla^2 F_{\ell}(\bx_2)\bU_{2}\|_{\op}\nonumber\\
 & \le \max_{\ell\le n}\big\|\nabla^2F_{\ell}(\bx_1)-\nabla^2F_{\ell}(\bx_2)\big\|_{\op}+
 3
 \max_{\ell\le n}\|\nabla^2F_{\ell}(\bx_1)\|_{\op}\big\|\bU_1-\bU_2\big\|_{\op}\nonumber\\
 &\le 
  \max_{\ell\le n}\big\|\nabla^2F_{\ell}(\bx_1)-\nabla^2F_{\ell}(\bx_2)\big\|_{\op}+
   3C\max_{\ell\le n}\frac{\|\nabla^2F_{\ell}(\bx_1)\|_{\op}}{\|\bx_1\|_2}
 \|\bx_1-\bx_2\|_2\nonumber\\
 & \le C_0\pmax^3 \sqrt{d\log(\pmax)}\|\bx_1-\bx_2\|_2
 +\frac{C_0}{\rho} \pmax^2 \sqrt{d\log(\pmax)}\|\bx_1-\bx_2\|_2
 \, . \nonumber
 \end{align}
\end{proof}


\section{Bounds on the Hessian}
In this section we derive the main lower bound on the eigenvalues of the Hessian, which will be used in the analysis of the Hessian descent algorithm.

Denote by $\bcH(\bx) :=  \nabla^2H(\bx)|_{\Ts_{\bx}}$ the restriction of the
Hessian on the tangent space. In a matrix representation, this is given by 
\begin{align}
\bcH(\bx)& = \sum_{\ell=1}^nF_{\ell}(\bx)\bU^{\sT}_{\bx}  \nabla^2F_{\ell}(\bx) \bU_{\bx}+
\bU_{\bx}^{\sT}\bD\bF(\bx)^{\sT}\bD\bF(\bx)\bU_{\bx}\label{eq:HessianPerp}\\
&=:\bcH_1(\bx)+\bcH_2(\bx)\, .
\end{align}

We omit the calculation for the next lemma, which is similar e.g. to \cite{ABC,montanari2023solving,subag2017complexity,SubagConcentrationPoly}. Recall that $\xi_i(t):=t^{p_i}$.
\begin{lemma}\label{lemma:HessianDistr}
Define for $k\in \naturals$, $q\in [0,1]$ define $\bS_{k}(q):= 
\diag\Big(\sqrt{\xi^{(k)}_i(q)}:\, i\le n\Big)$, where $\xi^{(k)}_i$ is the $k$-th derivative 
of $\xi_i$.
For a fixed $\bx\in\reals^d$ with $\|\bx\|^2_2=q$,  we have $\bF(\bx) = \bS_0(q)\, \bg$,
$\bD\bF(\bx)\bU_{\bx} = \bS_1(q)\, \bZ$, $\bU_{\bx}^{\sT}\nabla^2 F_{\ell}(\bx) \bU_{\bx}= \sqrt{\xi_{\ell}''(q)}\,\bW_{\ell}$.
Where $\bg, (\bW_{\ell})_{\ell\le n}, \bZ$ are mutually independent with
\begin{align}
\bg\sim\normal(0,\id_n)\, ,\;\;\; \bW_{\ell}\sim\GOE(d-1)\, ,\;\;\; \bZ\sim\GOE(n,d-1)\, .
\end{align}
As a consequence
\begin{align}
\bcH(\bx) & =  \|\bS_{0}(q)\bS_{2}(q)\bg\|_2\bW+\bZ^{\sT}\bS_{1}(q)^2\bZ\, ,\\
H(\bx)    &= \frac{1}{2}\|\bS_{0}(q)\bg\|_2^2\, ,
\end{align}
where
\begin{align}
(\bg,\bW,\bZ)\sim \normal(0,\id_n)\otimes \GOE(d-1)\otimes \GOE(n,d-1)\, .
\end{align}
\end{lemma}

We denote by $\lambda_i(\bM)$ the
$i$-th smallest eigenvalue of a symmetric matrix $\bM$.
\begin{theorem}\label{thm:ImprovedHessian}
For $t\in (0,1)$, define  $\xi'_{\min}(t) := 
\min_{i\le n}\xi'_i(t)$ and $\xi''_{\min}(t) := \min_{i\le n}\xi''_i(t)$. 
Then for any $s\in \naturals$, there exist 
constants $C_*$, $C_{\#}$, and $d_0=d_0(s)$  such that the following holds. Define $A_{\#}, B_{\#}$ via
\begin{align}
    A_{\#} &= C_{\#} \left(1\vee \sqrt{\frac{\log\pmax}{\log d}}\vee  \sqrt{\frac{-\log
    [\xi'_{\min}(\rho^2)\wedge\rho]}{\log d}}\right)\, ,\label{eq:FirstCondEV}\\
     B_{\#} &= C_{\#} \left(\sqrt{\log\pmax}\vee  \sqrt{-\log(\xi''_{\min}(\rho^2)\rho^2\wedge 1)}  \vee \sqrt{\log d}\right)\label{eq:SecCondEV}\, .
\end{align}
If $n\le d-4 A_{\#}\, (d\log d)^{1/2}$,  $n\le d-B_{\#}\, d^{1/2}$, 
then the following holds for all $d\ge d_0(s)$:
\begin{align}
\prob\Big(\forall \bx\in\Ball^d(\rho,1) \;\;
\lambda_{\fix}(\bcH(\bx))\le - \frac{1}{10}\sqrt{(d-n)\xi_{\min}''(\rho^2) H(\bx)}\Big) \ge1-C_*e^{-d/C_*}\, .\label{eq:MainProbBound}
\end{align}
%
\end{theorem}
\begin{proof}
For $\balpha\in\reals^n$, define the symmetric form
$\bM(\bx;\balpha)$ on $\Ts_{\bx}$ by 
\begin{align}
\bM(\bx;\balpha) & = \sum_{\ell=1}^n\alpha_{\ell}\nabla^2 F_{\ell}(\bx)\Big|_{\Ts_{\bx}}\, .
\end{align}
Let $\bv_1(\bx), \dots, \bv_{d-1}(\bx)$ be right singular vectors of $\bD\bF(\bx)|_{\Ts_{\bx}}$
corresponding to singular values $\sigma_1(\bD\bF(\bx)|_{\Ts_{\bx}})\ge \sigma_2
(\bD\bF(\bx)|_{\Ts_{\bx}})\ge \dots \ge \sigma_{d-1}(\bD\bF(\bx)|_{\Ts_{\bx}})$
(whereby we include zero singular values).
For $m\in \{1,\dots,n+1\}$, define
\begin{align}
V_m(\bx) := {\rm span}\big(\bv_{n-m+2}(\bx),\dots,\bv_{d-1}(\bx)\big)\, .\label{eq:VmDef}
\end{align}
Let $\nulls(\bD\bF(\bx)|_{\Ts_{\bx}})\subseteq \Ts_{\bx}$ denote the null space of $\bD\bF(\bx)|_{\Ts_{\bx}}$ (in the basis $\bU_{\bx}$, this is equivalent to the null space of
$\bD\bF(\bx)\bU_{\bx}$).
Since $\dim(\nulls(\bD\bF(\bx)|_{\Ts_{\bx}}))\ge d-1-n$, we have
\begin{align}
V_1(\bx) & \subseteq  \nulls(\bD\bF(\bx)|_{\Ts_{\bx}})\, ,\label{eq:NullInclusion}\\
V_1(\bx)&\subseteq V_2(\bx )\subseteq \dots\subseteq V_{n+1}(\bx)\, .
\end{align}
(with  $V_1(\bx)  = \nulls(\bD\bF(\bx)|_{\Ts_{\bx}})$ if $\bD\bF(\bx)|_{\Ts_{\bx}}$ has full 
row rank.)
We will define $d_m:= \dim(V_m(\bx)) = d-n+m-2$ and, for each $m$,
\begin{align}
L_{k,m}(\bx;\balpha) :=\lambda_k\big(\bM(\bx;\balpha)|_{V_m(\bx)}\big)\, .
\end{align}
where $\lambda_{\ell}(\bA)$ denotes the $\ell$-th smallest eigenvalue of matrix $\bA$.

Defining $\balpha_*(\bx):= \bF(\bx)/\|\bF(\bx)\|_2$, we then have 
\begin{align}
\lambda_{s}(\bcH(\bx))&\le \lambda_{s}(\bcH_1(\bx)|_{\nulls(\bD\bF(\bx)\bU_{\bx})})\\
&= \|\bF(\bx)\|_2\cdot \left.\lambda_s(\bM(\bx;\balpha_*(\bx))\right|_{\nulls(\bD\bF(\bx)\bU_{\bx})})\\
&\le \|\bF(\bx)\|_2\cdot \lambda_s(\bM(\bx;\balpha_*(\bx)|_{V_1(\bx)})\\
& =
\|\bF(\bx)\|_2\cdot L_{s,1}(\bx;\balpha_*(\bx))\, .
\end{align}
Further notice that, by the variational representation of eigenvalues, we have,
for each $k,m$,
\begin{align}
L_{k,m}(\bx;\balpha) \le L_{k+1,m+1}(\bx;\balpha)\, .
\end{align}
To see this, writing for simplicity $\bM=\bM(\bs;\balpha)$ and $V_m=V_m(\bx)$, 
note that
\begin{align*}
\lambda_{k+1}(\bM|_{V_{m+1}}) &= \max_{
\substack{U\subseteq V_{m+1} \\ \dim(U) = d_{m+1}-k}}
\lambda_1(\bM|_U) \\
& \ge  \max_{\substack{U\subseteq V_{m}\\ \dim(U) = d_{m}-(k-1)}}
\lambda_1(\bM|_U)\\
& = \lambda_{k}(\bM|_{V_{m}})\, ,
\end{align*}
where the inequality holds, since $V_m\subseteq V_{m+1}$ and $d_{m+1}-k = d_{m}-(k-1)$.

We then have
\begin{align}
\rP_{\sgood} &:= \prob\Big(\forall \bx\in\Ball^d(\rho,1) \;
\lambda_{s}(\bcH(\bx))\le -\frac{1}{10}\sqrt{ (d-n) \xi''_{\min}(\rho^2) H(\bx)}\Big)\nonumber\\
&\ge 
\prob\Big(
\max_{\bx\in \Ball^d(\rho,1)}\max_{\balpha\in \S^{n-1}}L_{s,1}(\bx;\balpha)
\le - \frac{1}{10}\sqrt{\xi''_{\min}(\rho^2)(d-n)}
\Big)\, .\label{eq:MainProbClaim}
\end{align}

Now note that, for fixed $\balpha,\bx$, with $\|\balpha\|_2=1$, we have
\begin{align}
\big(\bM(\bx;\balpha), \bD\bF(\bx)|_{\Ts_{\bx}}\big)&\ed \big(\nu(\balpha;\|\bx\|_2^2)\cdot\bW, \, 
\bS_1(\|\bx\|_2^2)\cdot \bG\big)\, ,\\
\nu(\balpha;q) & =\Big(\sum_{i=1}^n\alpha_i^2\xi''_i(q)\Big)^{1/2}\, ,
\end{align}
where $\bW\sim\GOE(d-1)$, $\bG\sim\GOE(n,d-1)$ are independent random matrices
(cf. Lemma \ref{lemma:HessianDistr}).
Hence $\bM(\bx;\balpha)|_{V_m(\bx)}\sim \nu(\balpha;q) \GOE(d_m)$, $q=\|\bx\|^2_2$.
Since $\|\balpha\|_2=1$,  we have $\nu(\balpha;q) \ge \sqrt{\xi_{\min}''(q)}$.
Therefore, by Lemma \ref{lemma:LD-Symm-GOE}, there exists an absolute constant $C_2>0$
such that, for any $k\le 3(d-n)/8$, 
\begin{align}
\prob\Big(L_{k,m}(\bx;\balpha)\ge -\frac15\sqrt{  \xi_{\min}''(\|\bx\|_2^2)(d-n)}\Big)\le 
e^{-C_2(d-n)^2}\, .
\end{align}

Define the following events, depending on absolute constants $C_0,C_1>0$,
which will be chosen below
\begin{align}
\cE_1(C_0)&:= \Big\{\big|L_{s,1}(\bx;\balpha_1)-L_{s,1}(\bx;\balpha_2)\big|\le C_0 d\pmax^{3}\|\balpha_1-\balpha_2\|_2
\;\; \forall \bx\in\Ball^d(1)\, \forall\balpha_1,\balpha_2\in\S^{n-1}\Big\}\, ,\label{eq:Event1}\\
\cE_{2,m}(C_1,\rho)&:= \Big\{L_{s,1}(\bx_1;\balpha)\le L_{m+s-1,m}(\bx_0;\balpha)+  \frac{C_1d^2\pmax^5}{\xi'_{\min}(\rho^2)^{1/2}\rho}\|\bx_0-\bx_1\|_2
\;\;\forall \bx_0,\bx_1\in\Ball^d(\rho,1)\, \forall\balpha\in\S^{n-1}\Big\}\, ,\label{eq:Event2}
\end{align}
where we recall that $\xi'_{\min}(t) :=\min_{i\le n}\xi'_i(t)$. 
Further, define
\begin{align}
\cE_{2}(C_1,\rho):=\cE_{2,m=\lfloor(d-n)/4\rfloor}(C_1,\rho).
\end{align}

Let $N^{d}(\eta)$ be an $\eta$-net in $\Ball^d(\rho,1)\times \S^{n-1}$.
Denoting by $(\bx^\eta,\balpha^{\eta})$ the projection of $(\bx,\balpha)\in\Ball^d(\rho,1)\times \S^{n-1}$
onto  $N^{d}(\eta)$, we have that, on $\cE_1(C_0)\cap \cE_2(C_1,\rho)$,
\begin{align}
L_{s,1}(\bx;\balpha)\le L_{m+s-1,m}(\bx^{\eta};\balpha^{\eta}) + \frac{(C_0+C_1)d^2\pmax^5\eta}{1\wedge [\rho\xi'_{\min}(\rho^2)^{1/2}]}\, .
\end{align}
Taking 
\begin{align}
\eta= \frac{1}{10(C_0+C_1)d^2\pmax^5}\xi''_{\min}(\rho^2)^{1/2}\{[\rho\xi'_{\min}(\rho^2)^{1/2}]\wedge 1\}\, ,
\end{align}
this implies 
$L_{s,1}(\bx;\balpha)\le L_{m+s-1,m}(\bx^{\eta};\balpha^{\eta}) +\xi''_{\min}(\rho^2)^{1/2}/10$.
Then continuing from Eq.~\eqref{eq:MainProbClaim}, with $m=\lfloor(d-n)/4\rfloor$ we obtain, for $C, C_2$ absolute constants and $s\le (d-n)/8$, 
\begin{align*}
\rP_{\sgood} &\ge 1- |N^{d}(\eta)|\max_{(\bx,\balpha)\in N^{d}(\eta)}
\prob\big(L_{m+s-1,m}(\bx;\balpha)\ge -\frac{1}{5}\sqrt{ \xi_{\min}''(\rho^2)(d-n)}\big)- 
\prob(\cE^c_1(C_0))-\prob(\cE^c_2(C_1,\rho))\\
& \ge 1- \left(\frac{C}{\eta}\right)^{2d} e^{-C_2(d-n)^2}- 
\prob(\cE^c_1(C_0))-\prob(\cE^c_2(C_1,\rho))\\
& \ge 1-e^{-\Lambda d}- 
C_* \, e^{-d/C_*}\, ,
\end{align*}
where the last step follows from
Lemma \ref{lemma:Event1} and Lemma \ref{lemma:Event2} (since $n\le d- 4A_{\#} (d\log d)^{1/2}$ 
holds by assumption),
and $\Lambda = C_2(d-n)^2/d-2 \log(C/\eta)$ is bounded by  
\begin{align}
\Lambda \ge C_2\frac{(d-n)^2}{d}- C' \left(\log d\vee \log \pmax \vee \log\frac{1}{\xi''_{\min}(\rho^2)\rho^2}\right)\, ,
\end{align}
for some absolute $C'$, where we used the fact that $\pmax\xi'_i(\rho^2)\ge \xi_i''(\rho^2)\rho^2$. The requirement that $n\le d-B_{\#}\, d^{1/2}$ with sufficiently large $C_{\#}$ guarantees that $\Lambda>c$ for some absolute $c>0$.
\end{proof}

\begin{lemma}\label{lemma:Event1}
Under the assumptions of Theorem \ref{thm:ImprovedHessian}, 
there exists an absolute  constant $C_*$ such that, defining $\cE_1$ as per
Eq.~\eqref{eq:Event1}, we have
\begin{align}
\prob(\cE^c_1(C_*)) \le  C_{*}\, e^{-d/C_*}\, .
\end{align}
\end{lemma}
\begin{proof}
By Weyl's inequality, 
\begin{align*}
\big|L_{s,1}(\bx;\balpha_1)-L_{s,1}(\bx;\balpha_2)\big|&\le 
\big\|\bM(\bx;\balpha_1)-\bM(\bx;\balpha_2)\big\|_{\op}\\
& \le \sqrt{n}\|\balpha_1-\balpha_2\|_2\cdot 
\max_{\ell\le n} \big\|\nabla^2 F_{\ell}(\bx)|_{\Ts_{\bx}}\big\|_{\op}\\
& \le C\sqrt{n}\|\balpha_1-\balpha_2\|_2 \pmax^2 \sqrt{d\log\pmax }\, ,
\end{align*}
where the last step holds with probability at least $1-C\exp(-d/C)$ by Lemma \ref{rmk:MaxHessian}.
\end{proof}

\begin{lemma}\label{lemma:Event2}
Under the assumptions of Theorem \ref{thm:ImprovedHessian}, 
there exist absolute constants $C_{\#}, C_*>0$ such that the following holds.
Define $\cE_2$ as per Eq.~\eqref{eq:Event2}, and 
let $A_{\#}$ be defined by
\begin{align}
    A_{\#} = C_{\#} \left(1\vee \sqrt{\frac{\log\pmax}{\log d}}\vee  
    \sqrt{\frac{-\log(\xi'_{\min}(\rho^2)\wedge \rho)}{\log d}} \right)\, .\label{eq:Asharp}
\end{align}
If $n\le d-A_{\#}\, (d\log d)^{1/2}$, $m= \lfloor (d-n)/4\rfloor $
then we have 
\begin{align}
\prob(\cE^c_{2,m}(C_*,\rho)) \le C_*\, e^{-d/C_*}\, .
\end{align}
\end{lemma}

Before proving Lemma \ref{lemma:Event2}, it is useful to state and 
prove an auxiliary result.

\begin{lemma}\label{lemma:SigmaDF}
For $t\in (0,1)$, let  $\xi'_{\min}(t) := \min_{i\le n}\xi'_i(t)$. 
For $\bx\in\Ball^d(1)$,
denote by $\sigma_1(\bx)\ge \sigma_2(\bx)\ge\cdots\ge \sigma_{d-1}(\bx)$ the singular values of 
$\bD\bF(\bx)|_{\Ts_{\bx}}$ (including the vanishing ones). 

Then there exist absolute constants $C_*$, $C_{\#}$, $\eps_0>0$  such that the following holds,
with $A_{\#}$ defined as per Eq.~\eqref{eq:Asharp}.
For any $\eps\in (0,\eps_0]$, 
$\rho\in (0,1)$, 
$m\ge (d-n)\eps$  and $n\le d-A_{\#}\sqrt{\eps^{-1}d(\log d)}$, we have 
\begin{align}
\prob\Big(\forall \bx\in\Ball^d(\rho,1):\;
\sigma_{n-m+1}(\bx) \ge \frac{1}{C_*}\sqrt{\xi'_{\min}(\|\bx\|^2_2)} (\sqrt{d}-\sqrt{n})\Big) \ge
 1-C_* e^{-d/C_*}\,.\label{eq:SigmaMinDF}
\end{align}
\end{lemma}
\begin{proof}
If $\ba_i\in \reals^{(d-1)}$, $i\le n+1$, define $Q_{ij}:=\<\ba_i,\ba_j\>$,
and $\bQ_\ell:= (Q_{ij})_{i,j\le \ell}\in \reals^{\ell\times \ell}$, then, by the variational principle,
for any $i\le n$
\begin{align}
\lambda_i(\bQ_{n+1})\le \lambda_{i}(\bQ_n)\, .
\end{align}
Note that $\sigma_{n-m+1}(\bx)^2=\lambda_{m-1}(\bQ_n)$
for $\bQ_n:=\bD\bF(\bx)\bU_{\bx}(\bD\bF(\bx)\bU_{\bx})^{\sT}$. In particular, removing
one row of $\bD\bF(\bx)$ cannot increase these eigenvalues. 
Hence, without loss of generality we can assume throughout the proof that
$n\ge d/2$. Also note that since $m$ is an integer, we may assume that $\eps(d-n)\ge1$.

First consider a fixed point $\bx\in\Ball^d(1)$. 
Recall that, by Lemma \ref{lemma:HessianDistr}, 
$\bD\bF(\bx)\bU_{\bx} = \bS^{(1)}(\|\bx\|_2^2)\, \bZ$, with $\bZ\sim\GOE(n,d-1)$.
Denoting the eigenvalues of $\bZ\bZ^{\sT}$ by  $\lambda_1(\bZ\bZ^{\sT})\le \lambda_2(\bZ\bZ^{\sT})\le
 \dots\le \lambda_{n}(\bZ\bZ^{\sT})$, we have that, for any $1\le m\le n$,
 \begin{equation}\label{eq:lambda_sigma}
 \xi_{\min}'(\|\bx\|^2)\lambda_{k(\eps)}(\bZ\bZ^{\sT})\le\xi_{\min}'(\|\bx\|^2)\lambda_{m}(\bZ\bZ^{\sT})\le \sigma_{n-m+1}(\bx)^2\,,
 \end{equation}
where we define
 $k(\eps)=\lfloor\eps(d-n)\rfloor$.
 
By Lemma \ref{lem:K-th-Eigengalue-Wishart}, with  
$N=d-1$, $M=n$, $\ell=k(\eps)$,
there exists constants $\eps_0>0$, $\Delta_*>0$, such that for all $\eps\in (0,\eps_0)$,
\begin{align}
\prob\Big(\lambda_{k(\eps)}(\bZ\bZ^{\sT}) \le 4\Delta^2_*(\sqrt{d}-\sqrt{n})^2\Big) 
\le  e^{-\eps(d-n)^2}\, ,
\end{align}
where we used that $\sqrt{d}-\sqrt{n}\le \sqrt{d-1}-\sqrt{n-1}$.
Therefore, for  $n\le d-A d^{1/2}(\log d)^{1/2}$, we have
\begin{align}
\prob\Big(\lambda_{k(\eps)}(\bZ\bZ^{\sT}) \le  4\Delta^2_*(\sqrt{d}-\sqrt{n})^2\Big) 
\le \exp\Big(-A^2\eps (d\log d)\Big)\,.\label{eq:ClaimDF}
\end{align}

Let $N^d(\eta)$ be an $\eta$-net in $\Ball^d(\rho,1)$. For $\bx\in \Ball^d(\rho,1)$, let
$\bx^{\eta}= \argmin_{\by\in N^d(\eta)}\|\bx-\by\|_2$ be its projection onto the net.
Then, using $\sqrt{d}-\sqrt{n}\ge 1/(2\sqrt{d})$,
\begin{align*}
\prob\Big(\exists\bx\in\Ball^d(\rho,1):\;\big|\sigma_{n-m+1}(\bx)-\sigma_{n-m+1}(\bx^{\eta})\big|&
\ge \Delta_*\sqrt{\xi'_{\min}(\|\bx\|^2_2)} (\sqrt{d}-\sqrt{n})\Big) \\
&\le \prob\Big(\Lip_{\perp}(\bD\bF;\Ball^d(1)) \ge \frac{\Delta_*}{2\eta\sqrt{d}}
\sqrt{\xi'_{\min}(\rho^2)}\Big)\\
&\le C_*e^{-d/C_*}\, ,
\end{align*}
where the last inequality holds by Lemma \ref{lemma:BoundLip} for
 $\eta< C\, d^{-1}\pmax^{-3}\rho\sqrt{\xi'_{\min}(\rho^2)}$
with $C$ a sufficiently small constant, and $C_*>0$ an absolute constant. 
For any $m\geq k(\eps)$ we therefore get, using \eqref{eq:lambda_sigma}  and \eqref{eq:ClaimDF},
\begin{align*}
\prob\Big(\exists\bx\in\Ball^d(\rho,1):\; &
\sigma_{n-m+1}(\bx) \le \Delta_*\sqrt{\xi'_{\min}(\|\bx\|_2^2)} (\sqrt{d}-\sqrt{n})\Big) 
\\ 
&\le \prob\Big(\exists\bx\in N^d(\eta):\;
\sigma_{n-m+1}(\bx) \le 2\Delta_*\sqrt{\xi^{(1)}_{\min}(\|\bx\|_2^2)} (\sqrt{d}-\sqrt{n})\Big) +
 C_*e^{-d/C_*}\\
 & \le | N^{d}(\eta)|\exp\Big(-A^2\eps(d\log d)\Big) +
 C_*e^{-d/C_*}\\
 & \le \exp\Big(d\log\frac{10}{\eta}-A^2\eps d\log d\Big) + 
 C_*e^{-d/C_*}\, .
\end{align*}
%
We then choose $\eta= C\, d^{-1}\pmax^{-3}\rho\sqrt{\xi'_{\min}(\rho^2)}$ and $A=A_{\#}/\sqrt{\eps}$,
with $C_{\#}$ a sufficiently large absolute constant
to conclude the proof. 
\end{proof}
We are now in position to prove Lemma \ref{lemma:Event2}.
\begin{proof}[Proof of Lemma \ref{lemma:Event2}]
Denote by $\cG_1$ and $\cG_2$ the high probability events of 
Lemmas \ref{rmk:MaxHessian} and  \ref{lemma:BoundLip}.
Further, let $\cG_3$ be the high probability event of 
Lemma \ref{lemma:SigmaDF} with $m= \lfloor(d-n)/4\rfloor$ and some fixed $\eps<1/4$. 
We suppose that $n\le d-A_{\#}\sqrt{d(\log d)}$. Since $\eps$ is fixed, by increasing $C_{\#}$ if needed, we may assume that the conclusion of Lemma \ref{lemma:SigmaDF} holds. Namely, that $\cG_3$ holds with probability at least $1-C_* e^{-d/C_*}$. To complete the proof of the lemma,
we will prove that the constant $c_2>0$ can be chosen so that 
$\cE_2(c_2,\rho)\supseteq  \cG_1\cap \cG_2\cap \cG_3$.  Hence, we hereafter assume
that events  $\cG_1$, $\cG_2$, $\cG_3$ hold.

In the following we will identify
$\Ts_{\bx}$ with $V^{\perp}(\bx)$ (the orthogonal complement of $\bx$) in the obvious 
way, and therefore identify $\Ts_{\bx_1}$, $\Ts_{\bx_2}$ when $\bx_1=\alpha\bx_2$
for $\alpha \in \reals\setminus\{0\}$. We also recall the rotation 
$\bR_{\bx_1,\bx_2}\in\reals^{d\times d}$
defined Section \ref{sec:Preliminaries}, which maps $\Ts_{\bx_1}$ to $\Ts_{\bx_2}$.
We will use the same notation for the restriction $\bR_{\bx_1,\bx_2}:\Ts_{\bx_1}\to\Ts_{\bx_2}$, which
is the parallel transport on the sphere, along the
geodesic connecting the two points $\bx_1/\|\bx_1\|$ and $\bx_2/\|\bx_2\|$.
Note that $\bR_{\bx_1,\bx_2}= \bR_{\bx_2,\bx_1}^{-1} = \bR_{\bx_2,\bx_1}^{\sT}$

Fix two points $\bx_0,\bx_1\in \Ball^d(1)$, and define the following linear
operators:
\begin{align}
\bD_0 &:= \bD\bF(\bx_0)\big|_{\Ts_{\bx_0}}\, ,\;\;\;\; \bD_1:= \bD\bF(\bx_1)\big|_{\Ts_{\bx_1}}\bR_{\bx_0,\bx_1}\, ,
\label{eq:Ddef}\\
\bM_0 & :=\bM(\bx_0;\balpha)\, ,\;\;\;\;  \bM_1 :=\bR_{\bx_0,\bx_1}^{\sT}
\bM(\bx_1;\balpha)\bR_{\bx_0,\bx_1}\, .\label{eq:Mdef}
\end{align}
 We view $\bD_0,\bD_1$ as linear operators $\bD_0,\bD_1:\Ts_{\bx_0}\to\reals^n$,
 and $\bM_0,\bM_1$ as symmetric forms $\bM_0,\bM_1:\Ts_{\bx_0}\times \Ts_{\bx_0}\to\reals$. We note that below restrictions $M_i|_V$ to a linear subspace $V\subset \Ts_{\bx_0}$ are defined  as $V\times V\to\R$ symmetric forms.

Recall definition \eqref{eq:VmDef}, which we repeat here for the reader's convenience
\begin{align} 
V_m(\bx):={\rm span}\big(\bv_{n-m+2}(\bx),\dots,\bv_{d-1}(\bx)\big)\subseteq \Ts_{\bx}\, ,
\end{align}
where $\bv_{1}(\bx),\dots,\bv_{d-1}(\bx)$ denote the right singular vectors of $\bD\bF(\bx)|_{\Ts_{\bx}}$
 (corresponding to the singular values in decreasing order).
 In particular, letting $\bv_{a,1},\dots,\bv_{a,d-1}$ denote the right singular vectors of $\bD_a$, 
 we  introduce the shorthands
 \begin{align}
 V_{0,m} &:= V_m(\bx_0)= {\rm span}\big(\bv_{0,n-m+2},\dots,\bv_{0,d-1}\big)\, ,\\
V_{1,s} & := {\rm span}\big(\bv_{1,n-s+2},\dots,\bv_{1,d-1}\big)\, .
 \end{align}
 Note that we have $V_{1,s} = \bR_{\bx_1,\bx_0}V_{s}(\bx_1)$ and therefore
 we obtain the following identities:
 \begin{align}
 L_{k,m}(\bx_0;\balpha) & = \lambda_k(\bM_0|_{V_{0,m}}) \, ,\;\;\;
 L_{s,1}(\bx_1;\balpha) =  \lambda_s(\bM_1|_{V_{1,1}})\, .
 \end{align}
 Recall that $d_m:=\dim(V_{0,m}) = d-n+m-2$.
 Let $\bE_0\in\reals^{d\times d_m}$ be a matrix whose columns form an orthonormal basis of 
 $V_{0,m}$, and $\bE_1\in \reals^{d\times d_1}$ a  matrix whose columns 
 form an orthonormal basis of $V_{1,1}$. Then we can rewrite the above formulas as
 \begin{align}
 L_{k,m}(\bx_0;\balpha) & = \lambda_k(\bE_0^{\sT}\bM_0\bE_0) \, ,\;\;\;
 L_{s,1}(\bx_1;\balpha) =  \lambda_s(\bE_1^{\sT}\bM_1\bE_1)\, .\label{eq:LmL1}
 \end{align}
By the variational representation of eigenvalues, for $m>s$, $k\ge s$,
 we have
 \begin{align}
 L_{k,m}(\bx_0;\balpha) & = \max_{\bQ\in \cO(d_m,d_m-k+s)}\lambda_s\big(\bQ^{\sT}\bE_0^{\sT}\bM_0\bE_0\bQ\big) \,,
 \label{eq:LMVar}
 \end{align}
 where the maximization is over the Stiefel manifold
 $\cO(d_m,d_m-k+s)\subseteq\reals^{d_m\times (d_m-k+s)}$ of  matrices whose columns form an orthonormal frame.
 
 Let $\obE_0\in\reals^{d\times (d-1-d_m)}$ be such that $[\bE_0|\obE_0]\in\reals^{d\times(d-1)}$
 is an orthonormal basis of $\Ts_{\bx_0}$. Note that $\bE_0\bE_0^{\sT}$ and $\obE_0\obE_0^{\sT}$ are the projections to the column space of $\bE_0$ and $\obE_0$, respectively. We then have
 \begin{align*}
 \bE_1^{\sT}\bM_1\bE_1 =\; & \bE_1^{\sT}\bE_0\bE_0^{\sT}\bM_1\bE_0\bE_0^{\sT}\bE_1 
 + \bE_1^{\sT}\obE_0\obE_0^{\sT}\bM_1\bE_0\bE_0^{\sT}\bE_1\\
&+\bE_1^{\sT}\bE_0\bE_0^{\sT}\bM_1\obE_0\obE_0^{\sT}\bE_1+
\bE_1^{\sT}\obE_0\obE_0^{\sT}\bM_1\obE_0\obE_0^{\sT}\bE_1\\
 \preceq\; &\bA^{\sT}\bE_0^{\sT}\bM_1\bE_0 \bA+
2\|\bE_1^{\sT}\obE_0\|_{\op}\|\bM_1\|_{\op}\cdot \id + \|\bE_1^{\sT}\obE_0\|^2_{\op}\|\bM_1\|_{\op}
\cdot\id\\
 \preceq\; &\bA^{\sT}\bE_0^{\sT}\bM_1\bE_0 \bA+
3\|\bE_1^{\sT}\obE_0\|_{\op}\|\bM_1\|_{\op}\cdot \id \, ,
\end{align*}
where we defined $\bA := \bE_0^{\sT}\bE_1\in\reals^{d_m\times d_1}$, $\|\bA\|_{\op}\le 1$ and we write $\bB_1\preceq\bB_2$ if $\bB_2-\bB_1$ is non-negative definite.
Continuing from the previous sequence of inequalities,
 \begin{align}
 \bE_1^{\sT}\bM_1\bE_1 \preceq \bA^{\sT}\bE_0^{\sT}\bM_0\bE_0 \bA+
 \|\bM_0-\bM_1\|_{\op}\cdot \id+
3\|\bE_1^{\sT}\obE_0\|_{\op}\|\bM_1\|_{\op}\cdot \id\,.
\label{eq:E1ME1}
\end{align}
%
%
%

Let $\bA = \bQ_*\bS\bO^{\sT}$, $\bQ_*\in\cO(d_m,d_1)$,
 $\bO\in\cO(d_1,d_1)$, $\bS\succeq \bfzero$ diagonal,
 be the reduced singular value 
 decomposition of $\bA$ (if ${\rm rank}(\bA)<d_1$, orthogonal columns
 are added to $\bQ_*$, $\bO$), and define $\bA_0: =  \bQ_*\bO^{\sT}$.
 Noting that $\|\bS\|_{\op}\le 1$ we have
 \begin{align*}
 \|\bA-\bA_0\|_{\op} &= \|\id-\bS\|_{\op}
 \le \|\id-\bS^2\|_{\op} \\
 &= \|\id-\bA^{\sT}\bA\|_{\op} = \|\id-\bE_1^{\sT}\bE_0\bE_0^{\sT}\bE_1\|_{\op} \\
 &= \|\obE_0^{\sT}\bE_1\|^2_{\op}\, ,
 \end{align*}
where for the last equality we used that $\id=\bE_0\bE_0^{\sT}+\obE_0\obE_0^{\sT}+\bx_0\bx_0^{\sT}/\|\bx_0\|^2$ and $\bx_0^{\sT}E_1=\bzero$.
Hence
 \begin{align*}
  \bA^{\sT}\bE_0^{\sT}\bM_0\bE_0 \bA
  \le \bA_0^{\sT}\bE_0^{\sT}\bM_0\bE_0 \bA_0+
  3 \|\obE_0^{\sT}\bE_1\|^2_{\op}\|\bM_0\|_{\op}\cdot\id\, .
 \end{align*}
Substituting this in Eq.~\eqref{eq:E1ME1} and using
 $\|\bE_0^{\sT}\bE_1\|^2_{\op}\le 1$, we get
 \begin{align*}
 \bE_1^{\sT}\bM_1\bE_1 \preceq \bO\bQ_{*}^{\sT}\bE_0^{\sT}\bM_0\bE_0 \bQ_*\bO^{\sT}+
 \|\bM_0-\bM_1\|_{\op}\cdot \id+
3\|\bE_1^{\sT}\obE_0\|_{\op}(\|\bM_0\|_{\op}+\|\bM_1\|_{\op})\cdot \id\,,
\end{align*}
 and therefore
 \begin{align*}
 \lambda_s(\bE_1^{\sT}\bM_1\bE_1) \le  \lambda_s(\bQ_*^{\sT}\bE_0^{\sT}\bM_0\bE_0 \bQ_*)+
 \|\bM_0-\bM_1\|_{\op}+
3\|\bE_1^{\sT}\obE_0\|_{\op}(\|\bM_0\|_{\op}+\|\bM_1\|_{\op}) \, .
\end{align*}
By Eq.~\eqref{eq:LmL1} we have $\lambda_s(\bE_1^{\sT}\bM_1\bE_1)=L_{s,1}(\bx_1;\balpha)$.
Further, by  Eq.~\eqref{eq:LMVar},
$\lambda_s(\bQ_*^{\sT}\bE_0^{\sT}\bM_0\bE_0 \bQ_*)\le L_{m+s-1,m}(\bx_0;\balpha)$. 
We thus obtained
 \begin{align}
 L_{s,1}(\bx_1;\balpha)\le  L_{m+s-1,m}(\bx_0;\balpha)+
 \|\bM_0-\bM_1\|_{\op}+
3\|\bE_1^{\sT}\obE_0\|_{\op}(\|\bM_0\|_{\op}+\|\bM_1\|_{\op})\, .\label{eq:L1LmBound}
\end{align}

We finally bound the error terms on the right-hand side of Eq.~\eqref{eq:L1LmBound}.
Recall the definitions of Eqs.~\eqref{eq:Ddef}, \eqref{eq:Mdef}.  We now choose 
an orthonormal basis on $\Ts_{\bx_0}$, which we write as an orthogonal matrix 
$\bU_{0}\in \cO(d,d-1)$, and an orthonormal basis for $\Ts_{\bx_1}$, given by
$\bU_{1} = \bR_{\bx_0,\bx_1}\bU_{0}$. We then obtain (identifying operators with their matrix representation)
\begin{align}
\bD_0 &:= \bD\bF(\bx_0)\bU_0\, ,\;\;\;\; \bD_1:= \bD\bF(\bx_1)\bR_{\bx_0,\bx_1}\bU_0 = \bD\bF(\bx_1)\bU_1\, ,
\label{eq:Ddef2}\\
\bM_0 & := \sum_{\ell=1}^n\alpha_{\ell}\bU_0^{\sT}\nabla^2F_{\ell}(\bx_0)\bU_0
\, ,\;\;\;\;  \bM_1 :=\sum_{\ell=1}^n\alpha_{\ell}\bU_1^{\sT}\nabla^2F_{\ell}(\bx_1)\bU_1
\, .\label{eq:Mdef2}
\end{align}
 We then have
 \begin{align}
 \|\bM_0-\bM_1\|_{\op} 
 & \le \|\balpha\|_1\max_{\ell\le n}\big\|\bU_0^{\sT}\nabla^2F_{\ell}(\bx_0)\bU_0-
 \bU_1^{\sT}\nabla^2F_{\ell}(\bx_1)\bU_1\big\|_{\op}\\
 & \le \frac{C_0}{\rho} d \pmax^{3} \sqrt{\log(\pmax)} \|\bx_0-\bx_1\|\, , \label{eq:BoundM0M1}
 \end{align}
 where we used $\|\balpha\|_1\le \sqrt{n}\le \sqrt{d}$.
 The last inequality holds on  the event  $\cG_2$
 of Lemma  \ref{lemma:BoundLip}.
 
 By a similar argument,
 \begin{align}
 \|\bM_1\|_{\op} 
 & \le  \|\balpha\|_1\max_{\ell\le n}\|\nabla^2F_{\ell}(\bx_0)\|_{\op}\le C_1 d \pmax^2\sqrt{\log(\pmax)}\, .
 \label{eq:BoundM1}
 \end{align}
 on the event $\cG_1$ of Lemma \ref{rmk:MaxHessian},
 and similarly for $\|\bM_0\|_{\op}$.
  
 Further, 
 \begin{align}
 \|\bD_0-\bD_1\|_{\op}&= \big\|\bD\bF(\bx_0)\bU_0-\bD\bF(\bx_1)\bU_1\big\|_{\op}\label{eq:BoundD0D1}\\
 &\le \frac{C_0}{\rho}\pmax^2\sqrt{d\log(\pmax)}\|\bx_0-\bx_1\|\, ,\nonumber
 \end{align}
 where the last step we used once more the fact that event $\cG_2$
 of Lemma \ref{lemma:BoundLip} holds.

Finally, by Corollary \ref{coro:Wedin} (applied with $\bA_i=\bD_i$, and $d$ replaced by $d-1$),
we get
\begin{align}
\big\|\bE_1^{\sT}\obE_0\big\|_{\op}&\le \frac{1}{\sigma_{n-m+1}(\bD_0)}
\|\bD_0-\bD_1\|_{\op}\nonumber\\
& \stackrel{(a)}{\le} \frac{C_0/C_*}{(\sqrt{d}-\sqrt{n})\xi'_{\min}(\|\bx_0\|^2)^{1/2}\rho}\, \pmax^2\sqrt{d\log(\pmax)}\|\bx_0-\bx_1\|\\
& \stackrel{(b)}{\le} \frac{2C_0/C_* d\pmax^2\sqrt{\log(\pmax)}}{\xi'_{\min}(\|\bx_0\|^2)^{1/2}\rho} \|\bx_0-\bx_1\|\, .
\label{eq:BoundEE}
\end{align}
Here $(a)$ holds by the bound \eqref{eq:BoundD0D1}, on the event $\cG_3$  
of Lemma \ref{lemma:SigmaDF} (with $m= \lfloor(d-n)/4\rfloor$) and $(b)$ since,
we have $\sqrt{d}-\sqrt{n}\ge (d-n)/(2\sqrt{d})\ge 1/(2\sqrt{d})$.
 
Using the bounds \eqref{eq:BoundM0M1}, \eqref{eq:BoundM1},
and \eqref{eq:BoundEE} in Eq.~\eqref{eq:L1LmBound}, we finally obtain
\begin{align*}
 L_{s,1}(\bx_1;\balpha)&\le  L_{m+s-1,m}(\bx_0;\balpha) \frac{C_0}{\rho} d \pmax^{3} \sqrt{\log(\pmax)} \|\bx_0-\bx_1\| 
 + \frac{C d^2 \pmax^4 \log(\pmax) }{\xi'_{\min}(\|\bx_0\|^2)^{1/2}\rho} \|\bx_0-\bx_1\|\\
 &\le L_{m+s-1,m}(\bx_0;\balpha)+ \frac{C'd^2\pmax^5}{\xi'_{\min}(\rho^2)^{1/2}\rho}\|\bx_0-\bx_1\|\, ,
\end{align*}
for some absolute constants $C,C'>0$, which concludes the proof.
\end{proof}

\section{Minimal singular value at solutions}

In Section \ref{sec:Newton} we will prove that w.h.p. any point sufficiently close to a solution is an approximate solution. The distance from the solution will be controlled by Lipschitz constants which we have already analyzed in Section \ref{sec:Preliminaries},
 and the minimal singular value of the Jacobian $\bD\bF(\bx)$ at the solution which we study in this section.
The following are the two main results we prove, treating separately the case of $n=d-1$ and $n< d-1$ (which will be used for Algorithm \ref{algo:MMS} and Algorithm \ref{algo:HD}, respectively).
\begin{proposition}\label{prop:sigmamin_MSS}
    Suppose $n=d-1$. For some  universal constant $C>0$ and any $u>0$,
    \begin{equation}
    \E\Big|\Big\{ \bx\in\S^{d-1}:\,\bF(\bx)=\bzero,\,\sigma_{\min}(\bD\bF(\bx)\bU_{\bx})\le \sqrt{u}d^{-\frac32}\pmax^{-\frac{d}{4}}
    \Big\}\Big| \le C u\,.       
    \end{equation}
\end{proposition} 

\begin{proposition}\label{prop:sigmamin}
    There exist universal constants $C,C'>0$ such that for $d/2\le n< d-1$, defining
    \begin{equation}\label{eq:tau}
        \tau:= \left(e^{C'd} \pmax^{d-\frac{3}{4}n-1}
    \log(\pmax)^{\frac14(d-n-1)}
	d^{\frac{1}{4}(d-n)}\right)^{-1}\,,
    \end{equation}
    we have that
	\begin{equation}\label{eq:sigmaminbd3}
		\P\left(\exists \bx\in \S^{d-1}:\,\bF(\bx)=\bzero,\,\sigma_{\min}(\bD\bF(\bx)\bU_{\bx})\le \tau
		\right) \le  C\exp(-d/C).
	\end{equation}
\end{proposition}

\subsection{\label{subsec:smallsigmaE} Upper bound on volume in expectation}
Consider the set
\[
\Psi(\tau):= \Big\{\bx\in\S^{d-1}:\,\bF(\bx)=0,\, \sigma_{\min}(\bD\bF(\bx)\bU_{\bx})\le \tau  \Big\}\,.
\]
Denote by $\Vol_{k}$ the $k$-dimensional Hausdorff measure. We start by proving an upper bound for the volume of $\Psi(\tau)$ in expectation. This bound will be used in the proof of both propositions above.
\begin{lemma}\label{lem:smallsigmaE}
	Let $\bZ\sim\GOE(n,d-1)$. Then, 
 \begin{equation}\label{eq:Psiub1}
	\E \,\Vol_{d-1-n}(\Psi(\tau)) \le \Vol_{d-1}(\S^{d-1}) 
	\prod_{i\leq n} \sqrt{\frac{p_i}{2\pi }}
	\E \Big(
    |\det(\bZ\bZ^{\top})|^{1/2} \ind\{\sigma_{\min}(\bZ)\leq \tau\}
		\Big)
	\,.
\end{equation}
 Moreover, if we assume that $\tau\le 1/n$, then for some universal constant $C_1>0$,
 \begin{align}
	\E \, \Vol_{d-1-n}(\Psi(\tau)) \le \Vol_{d-1}(\S^{d-1})
	\prod_{i\leq n} \sqrt{\frac{p_i}{2\pi }}
	\E \Big(
    |\det(\bZ\bZ^{\top})|^{1/2}\Big)\cdot C_1 d^{5/2}\tau \P\Big(\sigma_{\min}(\bZ)\leq \tau
		\Big)
	\,.\label{eq:Psiub2}
	\end{align}
\end{lemma}
\begin{proof}
    The expectation $\E  \Vol_{d-1-n}(\Psi(\tau))$ can be expressed by a variant of the Kac-Rice formula, e.g. as in \cite[Theorem 6.8]{AzaisWschebor} (see  \cite[Remark 3]{SubagConcentrationPoly} for a discussion on the regularity conditions required for its application).
	Precisely, abbreviating $\bDel(\bx)=\bD\bF(\bx)\bU_{\bx}$, the  formula gives
	\begin{align*}
		&\E \, \Vol_{d-1-n}(\Psi(\tau))\\ 
		&=\int_{\S^{d-1}}
		\varphi_{\bF(\bx)}(0)\E \Big[
		|\det(\bDel(\bx)\bDel(\bx)^{\top})|^{1/2} \ind\{\sigma_{\min}(\bDel(\bx))\leq \tau\}
		\,\Big|\,\bF(\bx)=0
		\Big]\de \Vol_{d-1}(\bx)\\
		&=\Vol_{d-1}(\S^{d-1})
		(2\pi )^{-n/2}\E \Big[
    |\det(\bDel(\bx)\bDel(\bx)^{\top})|^{1/2} \ind\{\sigma_{\min}(\bDel(\bx))\leq \tau\}
		\,\Big|\,\bF(\bx)=0
		\Big]\,,
	\end{align*}
	where $\varphi_{\bF(\bx)}(0) = (2\pi )^{-n/2}$ is the density of the random vector $\bF(\bx)$ evaluated at the origin and, using symmetry, in the second line $\bx$ is an arbitrary point in $\S^{d-1}$.

	Recall that by Lemma \ref{lemma:HessianDistr}, conditional on $\bF(\bx)=0$, $\bDel(\bx)\ed\bS_{1}(1)\bZ$
	where $\bS_{1}(1):= 
	\diag(\sqrt{p_i}:\, i\le n)$ and $\bZ\sim\GOE(n,d-1)$. 
	Note that $\sigma_{\min}(\bS_{1}(1)\bZ)\ge \sigma_{\min}(\bZ)$, since $\sigma_{\min}(\bA)=\min_{\bv\in\S^{n-1}}\max_{\bu\in\S^{d-2}}\bv^{\top}\bA\bu$ for any $\bA\in \R^{n\times (d-1)}$.  This proves \eqref{eq:Psiub1}.
	
    The random matrix $\bZ\bZ^{\top}$ is a Wishart matrix. Recall that its eigenvalues $\lambda_1\le \cdots\le \lambda_n$ have density $c_d\phi_{n,d-1}$ where $c_d$ is a normalizing constant and
 \begin{equation}\label{eq:Wishartdensity}
    \phi_{n,d-1}(x_1,\ldots,x_n)=\prod_{i} x_i^{(d-n-2)/2}e^{-x_i/2} \prod_{i<j}|x_i-x_j|
 \end{equation}
 on $B_n:=\{(x_1,\ldots,x_n):\,0<x_1<\cdots<x_n\}$. 
 Since 
 \[
 \phi_{n,d-1}(x_1,\ldots,x_n) = x_1^{(d-n-2)/2}e^{-x_1/2} \prod_{i=2}^n|x_i-x_1|\phi_{n-1,d-2}(x_2,\ldots,x_n)\,,
 \]
 by conditioning on $\sigma_{\min}(\bZ)^2=\lambda_1$,
 \begin{equation}\label{eq:detconditional}
 \E \Big[
    |\det(\bZ\bZ^{\top})|^{1/2} \,\Big|\,\sigma_{\min}(\bZ)\leq \tau
		\Big]=\int_0^{\tau^2} \frac{\E \Big[\prod_{i=1}^{n-1}\bar{\lambda}_i^{1/2}|\bar{\lambda}_i-t|\ind\{\bar{\lambda}_1>t\}\Big]}{\E \Big[\prod_{i=1}^{n-1}|\bar{\lambda}_i-t|\ind\{\bar{\lambda}_1>t\}\Big]} f_{n,d-1}(t)\sqrt{t}\, \de t\,,
 \end{equation}
 where $\bar{\lambda}_1\le\cdots\le \bar{\lambda}_{n-1}$ are the eigenvalues of $\bar\bZ\bar\bZ^{\top}$ with $\bar\bZ\sim\GOE(n-1,d-2)$ and $f_{n,d-1}(t)$ is the density of $\lambda_1$ normalized so that $\int_0^{\tau^2} f_{n,d-1}(t)\de t=1$.

 The numerator on the right-hand side of \eqref{eq:detconditional} is bounded from above by 
 \begin{align*}
\E\left[|\det(\bar\bZ\bar\bZ^{\top})|^{3/2}\right]
 \,.
 \end{align*}
Using the density \eqref{eq:Wishartdensity}, we can write the denominator as
\begin{align*}
    &c_{d-1}\int_{B_{n-1}+t} \prod_{i=1}^{n-1}|x_i-t| x_i^{(d-n-2)/2}e^{-x_i/2} \prod_{i<j}|x_i-x_j|\de x_1\cdots \de x_{n-1}\\
    &\ge c_{d-1}\int_{B_{n-1}} \prod_{i=1}^{n-1}|x_i| x_i^{(d-n-2)/2}e^{-(x_i+t)/2} \prod_{i<j}|x_i-x_j|\de x_1\cdots \de x_{n-1} \ge \frac12 \E\Big[|\det(\bar\bZ\bar\bZ^{\top})|\Big]\,,
\end{align*}
where by an abuse of notation we denote by $B_{n-1}+t$ the Minkowski sum of $B_{n-1}$ and $\{(t,\ldots,t)\}$ and the inequality holds if $t<1/n$.

Recall that for any matrix $\bA\in\R^{M\times N}$, $|\det(\bA\bA^{\top})|^{1/2} = \prod_{i\leq M} \|\Theta_i(\bA)\|$, where $\Theta_i(\bA)$ is the projection of the $i$-th row of $\bA$ onto the orthogonal space to its first $i-1$ rows. Hence, $|\det(\bar\bZ\bar\bZ^{\top})|=\prod_{i=1}^{n-1}\chi_{d-1-i}^2$ and $|\det(\bZ\bZ^{\top})|=\prod_{i=1}^{n}\chi_{d-i}^2$ in distribution, 
where $\chi^2_k$ are i.i.d. chi-squared variables of $k$ degrees of freedom. Using this, one can verify that, for absolute $C>0$,
\begin{align*}
\frac{\E\left[|\det(\bar\bZ\bar\bZ^{\top})|^{3/2}\right]}{\E\Big[|\det(\bar\bZ\bar\bZ^{\top})|\Big] \E\Big[|\det(\bar\bZ\bar\bZ^{\top})|^{1/2}\Big]}\le C d^2
 \mbox{\, \, and\, \, }
 \frac{\E\left[|\det(\bar\bZ\bar\bZ^{\top})|^{1/2}\right]}{\E\left[|\det(\bZ\bZ^{\top})|^{1/2}\right]}\le  C/\sqrt d\,.
\end{align*}
 %
Combining the above with \eqref{eq:detconditional} and plugging back into \eqref{eq:Psiub1} completes the proof.
\end{proof}

\subsection{\label{subsec:smallsigma_at_a_point} Lower bound on the volume of a neighborhood}

Recall the definitions of $\Lip(\bD\bF;\Omega)$ and $\Lip_{\perp}(\bD\bF;\Omega)$
in Eqs.~\eqref{eq:LipD}, \eqref{eq:LipDF_perp}.
\begin{lemma}\label{lem:smallsigmaLB}
	Let $c,\tau,r>0$ such that $r<1/10$ and $(1-r)^{\pmax}\tau-rc\sqrt{d}\ge \tau/2$. On the event that
	\begin{equation}\label{eq:LipDBd}
	\Lip_{\perp}(\bD\bF;\S^{d-1}(1)) \vee \Lip(\bD\bF;\Ball^d(1))\leq c \sqrt d,
	\end{equation}
	if there exists a point $\bx\in \S^{d-1}$  such that $\bF(\bx)=\bzero$ and $\sigma_{\min}(\bD\bF(\bx)\bU_{\bx}) \in [\tau,2\tau]$ then
	\begin{equation}\label{eq:Vollb}
	\Vol_{d-n-1} ( \Psi( 3\tau )
	)
	\ge \Big(\frac{r}{2\sqrt 2}\Big)^{d-n-1} \Vol_{d-n-1} \big(
\Ball^{d-n-1}(1)
\big)\,.
	\end{equation}	
\end{lemma}
\begin{proof}
	Assume throughout the proof that the event in \eqref{eq:LipDBd} occurs and let $\bx$ be a point as in the lemma. Since the $F_i$ are homogeneous, $\bD\bF(\bx)\bx=\bzero$. Setting $t=1-r$, for $\bz=t\bx$ we have that $\bD\bF(\bz)=\diag(t^{p_i-1}:\,1\le i\le n)\bD\bF(\bx)$  and $\sigma_{\min}(\bD\bF(\bz))= \sigma_{\min}(\bD\bF(\bz)\bU_{\bz}) \in t^{\pmax}[\tau,2\tau]$.
	

	Let $V:=\{\bD\bF(\bz)^{\top}\cdot \bv:\,\bv\in\R^n\}\subset \Ts_{\bz}$, let $\bT \in \R^{d\times n}$ be a matrix whose columns are an orthonormal basis of $V$ and let $\bK \in \R^{d\times (d-n-1)}$ be a matrix whose columns are an orthonormal basis of $V^\perp\cap \Ts_{\bz}$. For $(\ba,\bb)\in\R^{d-n-1}\times \R^n$, define
	\[
	\hat \bF(\ba,\bb) = \bF(\bz+\bK\cdot \ba+\bT\cdot \bb),
	\]
	and note that
	\begin{align*}
		\bD_1\hat\bF(\ba,\bb) &:= \Big(\frac{\de}{\de a_j} \hat F_i(\ba,\bb)\Big)_{i\le n,\, j\le d-n-1} =  \bD\bF(\bz+\bK\cdot \ba+\bT\cdot \bb)\cdot \bK,\\
		\bD_2\hat\bF(\ba,\bb) &:= \Big(\frac{\de}{\de b_j} \hat F_i(\ba,\bb)\Big)_{i\le n,\, j\le n} =  \bD\bF(\bz+\bK\cdot \ba+\bT\cdot \bb)\cdot \bT.
	\end{align*}

Since $\lambda_{\min}(\bD_2\hat\bF(0,0))=\sigma_{\min}(\bD\bF(\bz))$ and, for $i=1,2$,
\begin{align*}
\|\bD_i\hat\bF(\ba,\bb)-\bD_i\hat\bF(0,0)\|_{\op}&\leq \|\bD \bF(\bz+\bK\cdot \ba+\bT\cdot \bb) - \bD \bF(\bz)\|_{\op}\\
&\le \|\bK\cdot \ba+\bT\cdot \bb\| \Lip(\bD\bF;\Ball^d(1))\le c \sqrt d\, \|(\ba,\bb)\|\,,
\end{align*}
if $\|(\ba,\bb)\|=\|\bK\cdot \ba+\bT\cdot \bb\|\le r$ then, using the assumption
 $(1-r)^{\pmax}\tau-rc\sqrt{d}\ge \tau/2$,
 \begin{align}
 \label{eq:Sigma_min_D2}
\sigma_{\min}(\bD_2\hat\bF(\ba,\bb)),\, \sigma_{\min}(\bD \bF(\bz+\bK\cdot \ba+\bT\cdot \bb)) \ge \tau/2.
\end{align}

Suppose that $(\ba,\bb)$ is a point such that $\|(\ba,\bb)\|< r$ and $\hat \bF(\ba,\bb)=\bzero$. Then by the implicit function theorem, there is a unique continuously differentiable function $g$ from an open neighborhood of $\ba$ to $\R^n$ such that $g(\ba)=\bb$ and for any $\ba'$ in this neighborhood $\hat \bF(\ba',g(\ba'))=\bzero$. From compactness, 
there exists a closed set $A\subset \R^{d-n-1}$ containing the origin and a unique $g:A\to\R^n$ continuous on $A$ and continuously differentiable on its interior $A^o$  such that $g(\bfzero)=\bfzero$ and $\hat \bF(\ba,g(\ba))=\bzero$ for any $\ba\in A$,  and such that $\|(\ba,g(\ba))\| <r$ on $A^o$ and $\|(\ba,g(\ba))\| =r$ on $\partial A=A\setminus A^o$. Moreover, by the implicit function theorem, on $A^o$,
\[
\bD g(\ba):= \Big(\frac{\de}{\de a_j} g_i(\ba)\Big)_{i\le n,\, j\le d-n-1} = - \bD_2\hat\bF(\ba,g(\ba))^{-1} \cdot \bD_1\hat\bF(\ba,g(\ba)).
\]
Since $\bD_1\hat\bF(\bfzero,g(\bfzero))=\bfzero \in \R^{n\times (d-n-1)}$, on $\|(\ba,g(\ba))\|< r$,
$$\|\bD_1\hat\bF(\ba,g(\ba))\|_{\op} \le c \sqrt{d}\cdot \|(\ba,g(\ba))\| < \tau/2$$ 
and therefore, by Eq.~\eqref{eq:Sigma_min_D2},
$$\|\bD g(\ba)\|_{\op} < 1.$$ 

Define $B = \big\{\ba\in \R^{d-n-1}: \|\ba\|< r/\sqrt2\big\}$. Obviously, $A\supset B$ for otherwise we would have some $\ba\in \partial A\cap B$ for which $\|(\ba,g(\ba))\|=r$ while $\|\ba\|\vee \|g(\ba)\|<r/\sqrt{2}$. 
Define $P(\by):=\by/\|\by\|$ and $\by(\ba)=\bz+\bK\cdot \ba+\bT\cdot g(\ba)$. We have that
\begin{equation}
\label{eq:F0supsetB}
\Big\{\by\in \S^{d-1}: \bF(\by)=\bzero\Big\}
\supset
P\Big(\Big\{
\by(\ba):\, \ba\in B
\Big\}\Big).
\end{equation}
We note that (using that $r<1/10)$)  for any $\ba_1,\ba_2\in B$, 
\begin{equation}\label{eq:inflation}
    \|P(\by(\ba_1))-P(\by(\ba_2))\|\ge \frac12\Big\|\frac{\by(\ba_1)}{\|\by(\ba_1)\|\wedge\|\by(\ba_2)\|}-\frac{\by(\ba_2)}{\|\by(\ba_1)\|\wedge\|\by(\ba_2)\|}\Big\| \ge\frac12   
    \|\ba_1-\ba_2\|\,.
\end{equation}
It thus follows that 
\begin{align*}
\Vol_{d-n-1}\left(P\Big(\Big\{
\by(\ba):\, \ba\in B
\Big\}\Big)\right) &\ge 2^{-(d-n-1)}\Vol_{d-n-1}(B)\\
& = \Big(\frac{r}{2\sqrt 2}\Big)^{d-n-1} \Vol_{d-n-1} \big(
\Ball^{d-n-1}(1)
\big)\, .
\end{align*}
Moreover, for any $\ba$ in the set on the right-hand side of \eqref{eq:F0supsetB}, if $\bw=P(\by(\ba))$ then $\|\bw-\bx\|\le 2r$ and thus $\sigma_{\min}(\bD\bF(\bw)\bU_{\bw})\le 2\tau+2rc\sqrt d\leq 3\tau$.
\end{proof}

\subsection{\label{sec:sigmamin_bruteforce} Proof of Proposition \ref{prop:sigmamin_MSS}}
We assume throughout that 
$n=d-1$, and therefore $\Vol_{d-1-n}(\Psi(\tau)) =:|\Psi(\tau)|$
is just the cardinality of the set $\Psi(\tau)$. Formally,
$\Psi(\infty)=\{\bx\in\S^{d-1}:\,\bF(\bx)=\bzero  \}$. By 
Eq. (2.2) of \cite{SubagConcentrationPoly},  
\begin{equation}
	\E \,|\Psi(\infty)| = \Vol_{d-1}(\S^{d-1}) 
	\prod_{i\leq n} \sqrt{\frac{p_i}{2\pi }}
	\E \Big(
    |\det(\bZ\bZ^{\top})|^{1/2} 
		\Big)=\prod_{i=1}^{d-1}p_i\le\pmax^{\frac{d}{2}}
	\,.
\end{equation}%
Using Eq.~\eqref{eq:Psiub2} in Lemma \ref{lem:smallsigmaE} 
and Corollary \ref{cor:improvedRV}, we obtain that, for some absolute constant $C$,
\[
\frac{\E |\Psi(\tau)|}{\E |\Psi(\infty)|}\le C d^{3}\tau^2 \,.
\]
The proposition follows by substituting $\tau=\sqrt{u}d^{-\frac32}\pmax^{-\frac{d}{4}}$.

\subsection{\label{sec:smallsigmaP} Proof of Proposition \ref{prop:sigmamin}}
 
Let $C_0$ and $C_1$ be the  
absolute constants as in Lemmas
\ref{lemma:BoundLip} and \ref{lem:smallsigmaE} and denote by $\calE$ the event that \eqref{eq:LipDf} occurs.  
For any $\tau<1/n$, from Lemmas \ref{lem:smallsigmaE} and \ref{lem:smallsigmaLB}  and Markov's inequality, 
\begin{align*}
\prob\left(\calE\cap\left\{\exists \bx\in\S^{d-1}:\,\bF(\bx)=\bzero,\,\sigma_{\min}(\bD\bF(\bx)\bU_{\bx})\in[\tau,2\tau] \right\}\right)\le M_{d,n}(\tau)\, ,
\end{align*}
where
\begin{equation}\label{eq:sigmaminbd}
M_{d,n}(\tau) :=\frac{\Vol_{d-1}(\S^{d-1})
	\prod_{i\leq n} \sqrt{\frac{p_i}{2\pi }}
	\E \Big(
    |\det(\bZ\bZ^{\top})|^{1/2}\Big)\cdot C_1 d^{5/2}3\tau \P\Big(\sigma_{\min}(\bZ)\leq 3\tau
		\Big)}{\Big(\frac{r}{2\sqrt 2}\Big)^{d-n-1} \Vol_{d-n-1} \big(
\Ball^{d-n-1}(1)
\big)}\,,
\end{equation}
 provided that that  
\begin{equation}
\label{eq:smallr}
r<\frac{1}{10},\quad\quad(1-r)^{\pmax}\tau-rC_0\pmax^2\sqrt{\log(\pmax)d}\ge \frac{\tau}{2}\,.
\end{equation}
To satisfy these two inequalities, we shall assume without loss of generality that $C_0>1$, and set 
\begin{equation}\label{eq:r}
r=
\frac{\tau}{4C_0\pmax^2\sqrt{\log(\pmax)d}}\,.    
\end{equation}

Recall that $\Vol_{k-1}(\S^{k-1})=2\pi^{k/2}/\Gamma(\frac{k}{2})$ and $\Vol_{k} (
\Ball^{k}(1)
)=\pi^{(k/2)}/\Gamma(k/2+1)$ and $\E|\det(\bZ\bZ^{\top})|^{1/2}=2^{n/2}\Gamma(d/2)/\Gamma((d-n)/2)$ (e.g., by the relation to chi-variables as in the proof of Lemma \ref{lem:smallsigmaE}).
Hence,
\[
\frac{\Vol_{d-1}(\S^{d-1})
	\prod_{i\leq n} \sqrt{\frac{1}{2\pi }}
	\E \Big(
    |\det(\bZ\bZ^{\top})|^{1/2}\Big)}{ \Vol_{d-n-1} \big(
\Ball^{d-n-1}(1)
\big)}\le c\sqrt{d-n}
\,,
\]
for some universal $c>0$. By Corollary \ref{cor:improvedRV}, for universal $c'>0$,
\[
\P\Big(\sigma_{\min}(\bZ)\leq 3\tau
		\Big)\le \Big(\frac{c'\tau}{\sqrt{d-1}-\sqrt{n-1}}
		\Big)^{d-n}\,.
\]
Combining the above and substituting $r$ as in \eqref{eq:r}, we have that 
$M_{d,n}(\tau)$, defined in Eq.~\eqref{eq:sigmaminbd}, is bounded as
\begin{align*}
M_{d,n}(\tau)&\le c_0^{d-n}\pmax^{\frac{n}{2}}
	d^{5/2}\sqrt{d-n}\tau\Big(\frac{\pmax^2\sqrt{\log(\pmax)d}}{\tau}\Big)^{d-n-1}  \Big(\frac{\tau}{\sqrt{d-1}-\sqrt{n-1}}
		\Big)^{d-n}\\
  &\le \tau^2\cdot c_0^{d-n}\pmax^{\frac{n}{2}}
	d^{3}\Big(\pmax^2\sqrt{\log(\pmax)d}\Big)^{d-n-1}=:\tau^2 A \,,
\end{align*}
for some universal constant $c_0>0$.

By a union bound, 
\begin{align*}
\prob\left(
\calE\cap\left\{\exists \bx\in\S^{d-1}:\,\bF(\bx)=\bzero,\,\sigma_{\min}(\bD\bF(\bx)\bU_{\bx})\le \tau  \right\}
\right)
\le \sum_{i=1}^{\infty}(2^{-i}\tau)^2 A < \tau^2 A \,.
\end{align*}
Since $\P(\calE) \ge 1-C_*\exp(-d/C_*)$ by Lemma \ref{lemma:BoundLip}, taking $\tau$ as in \eqref{eq:tau} with large enough $C'$, this completes the proof. \qed

\section{Good events}
\label{sec:goodevent}

Our  algorithms are correct on an event of high probability, which we now define precisely. 
Recall Definition \ref{def:Lip} and let
\begin{equation}\label{eq:Lip^(i)}
	\begin{aligned}
		\Lip^{(0)}(\bF)  &:= \Lip(\bF;\Ball^d(1)) \, ,
		\\ 
		\Lip^{(1)}(\bF)  &:= \Lip(\bD\bF;\Ball^d(1))\vee \Lip_{\perp}(\bD\bF;\S^{d-1})  \, ,
		\\ 
		\Lip^{(2)}(\bF) &:= \Lip(\nabla^2\bF;\Ball^d(1))\vee 	\Lip_{\perp}(\nabla^2\bF;\S^{d-1})
		\, .
	\end{aligned}
\end{equation}
With $\rho=1$ and $A_{\#}$ and $B_{\#}$  as in Theorem \ref{thm:ImprovedHessian}, define $A=4A_{\#}\vee B_{\#}$ so that $A(d\log d)^{1/2}\ge 4A_{\#}(d\log d)^{1/2}$ and $A(d\log d)^{1/2}\ge B_{\#}d^{1/2}$.
Let $C'$, $C_0$ and $C_1$  be the absolute constants from Proposition \ref{prop:sigmamin} and  Lemmas \ref{lemma:BoundLip} and \ref{rmk:MaxHessian}.
As in  \eqref{eq:tau}, we define
\begin{equation*}
        \tau:= \left(e^{C'd} \pmax^{d-\frac{3}{4}n-1}
    \log(\pmax)^{\frac14(d-n-1)}
	d^{\frac{1}{4}(d-n+5)}\right)^{-1}\,.
\end{equation*}

\begin{definition}[Good event for Hessian Descent] \label{def:good} 
Suppose that  $\pmax\le d^2$ and $\frac{d}{2}\le n\le d-A (d\log d)^{1/2}$. We denote by $\calE=\calE(d,n)$ the event that: 
\begin{enumerate}
    \item \label{enum:good1} $\displaystyle \max_{\bx\in\Ball^d(1)}\|\bF(\bx)\|_{2}\le C_1\sqrt{d\log(\pmax)}\,$.
    \item \label{enum:good2} $\displaystyle \Lip^{(i)}(\bF)\le C_0\pmax^{1+i}\sqrt{d\log(\pmax)}\,$, for $i=0,1,2$.
    \item \label{enum:good3} $\displaystyle 
		\min\Big\{\sigma_{\min}(\bD\bF(\bx)\bU_{\bx}):\,\bx\in\S^{d-1},\,\bF(\bx)=0 \Big\} \ge \tau$.
    \item \label{enum:good4} $\displaystyle \forall \bx\in\S^{d-1}: \;\;
\lambda_{1}(\bcH(\bx))\le - \frac{1}{10}\sqrt{(d-n) H(\bx)}$. 
    \item \label{enum:good5} $\displaystyle \forall \bx\in\Ball^d(1,1+1/\pmax),\,2\le k\le 4: \;\;
\big\|\nabla^k H(\bx)\big\|_{\op} \le 9 C_1 d \pmax^k \log(\pmax) \, $.
 \item \label{enum:good6} For $\bx^0=(1,0,\ldots,0)$, $H(\bx^0)\le d$. 
\end{enumerate}
\end{definition}

\begin{lemma}\label{lem:goodeventprob}
    Assume the setting of Definition \ref{def:good}. Then, for some universal constant $C>0$,
    \begin{equation}\label{eq:goodeventprob}
    \P\left(\calE(d,n)\right) \ge 1-C\exp(-d/C)\,.
    \end{equation}
\end{lemma}
\begin{proof}
    It is enough to prove the same bound as in \eqref{eq:goodeventprob} for each of the events in Items \ref{enum:good1}-\ref{enum:good6} of Definition \ref{def:good} separately. For Items \ref{enum:good1}-\ref{enum:good5} the bound follows from Lemma \ref{rmk:MaxHessian},
    Lemma \ref{lemma:BoundLip},  Proposition \ref{prop:sigmamin} and Theorem \ref{thm:ImprovedHessian}, where for Item \ref{enum:good5} we used the fact that for $t>1$,
    $
    \big\|\nabla^4 H(t \bx)\big\|_{\op}\le t^{2\pmax} \big\|\nabla^4 H(\bx)\big\|_{\op}
    $, 
    since $\nabla^k F_i(t \bx)=t^{p_i-k}\nabla^k F_i(\bx)$.
    The bound for Item \ref{enum:good6} is trivial. 
\end{proof}

Recall the definition \eqref{eq:Ndknet} of the net $\calN_{d,k}$. 
For any $(\bx,k)$ with $\bx\in \calN_{d,k}$, define
$$\tilde\bx:=\argmin\big\{\|\by\|:\,\by\in\bx_k+\{0,2^{-k}\}^d\big\}$$
and
\[
\hat \calN_{d,k}:=\big\{ \tilde \bx/\|\tilde \bx\|:\,\bx\in\calN_{d,k},\,(\bx+[0,2^{-k}]^d)\cap\S^{d-1}\neq \varnothing\big\}\,.
\]
Below $\tilde C_0$ is a universal constant which will be  determined in the proof of Lemma \ref{lem:goodevent2}.

\begin{definition}[Good event for Multi-Scale Search] \label{def:good2} 
	Suppose that  $n=d-1$. Let  $0<u_2\le 1\le u_1$, $u_3,\delta>0$ be real numbers and $k_0\ge0$ an integer number. 
	Denote by $\calE=\calE(d,u_1,u_2,u_3,k_0)$ the event that: 
	\begin{enumerate}
		\item \label{enum:good_plarge_1} $\displaystyle \Lip^{(i)}(\bF)\le u_1 C_0\pmax^{1+i}\sqrt{d\log(\pmax)}\,$, for $i=0,1$.
		\item \label{enum:good_plarge_2} $\displaystyle \max_{\bx\in\Ball^d(1)}\|\bD \bF(\bx)\|_{\op}  \le u_1 C_0\pmax\sqrt{d\log(\pmax)}$.
		\item \label{enum:good_plarge_3} There are no points $\balpha\in\S^{d-1}$ such that $\bF(\balpha)=\bfzero$ and 
		$\displaystyle 
		\sigma_{\min}(\bD\bF(\balpha)\bU_{\balpha}) < \sqrt{u_2} d^{-\frac32}\pmax^{-\frac{d}{4}}$.
		\item \label{enum:good_plarge_4} For any $0\le k<k_0$, setting $T_0:= \tilde C_0\pmax\sqrt{d\log(\pmax)}$,
		\[
		\left|\left\{\hat\bx\in \hat\calN_{d,k}:\, \|\bF(\hat\bx)\|\le \sqrt{d}2^{-k} \cdot  u_1 C_0\pmax\sqrt{d\log(\pmax)} \right\} \right| \le u_3(u_1 T_0)^{d-1}\,.
		\]
	\end{enumerate}
\end{definition}

\begin{lemma}\label{lem:goodevent2}
	For the good event of Definition \ref{def:good2}, for some universal constant $C>0$,
	\[
	\P(\calE(d,u_1,u_2))\ge 1-2\exp(-d(u_1-1)^2/C)- C u_2 - k_0/u_3.
	\]
\end{lemma}

\begin{proof}
Note that 
\begin{align*}
	\Lip(\bF;\Ball^d(1))&\le \max_{\bx\in\Ball^d(1)}\|\nabla \bF(\bx)\|_{\op} = \max_{\bx\in\Ball^d(1)}\max_{\bv\in\S^{n-1},\bu\in\S^{d-1}}\<\nabla \bF(\bx),\bv\otimes\bu\>\,,\\
	\Lip(\bD\bF;\Ball^d(1))&\le \max_{\bx\in\Ball^d(1)}\|\nabla^2 \bF(\bx)\|_{\op}=\max_{\bx\in\Ball^d(1)}\max_{\bv\in\S^{n-1},\bu\in\S^{d-1}}\<\nabla^2 \bF(\bx),\bv\otimes\bu^{\otimes2}\>\,.
\end{align*}
Using that
\[
\<\nabla F_i(\bx),\bu\>=\<\bG^{(i)},\bu\otimes\bx^{\otimes p_i-1}+\bx\otimes\bu\otimes\bx^{\otimes p_i-2}+\cdots+\bx^{\otimes p_i-1} \otimes\bu\> \,,
\]
we have that the variance of $\<\nabla \bF(\bx),\bv\otimes\bu\>$ is bounded by $\pmax^2$ for $\bx,\bv,\bu$ as above. Similarly, the variance of $\<\nabla^2 \bF(\bx),\bv\otimes\bu^{\otimes2}\>$ is bounded by $\pmax^4$.
Hence, from Lemma \ref{lemma:BoundLip} and the Borell-TIS inequality, 
	\begin{align*}
	\P\Big(\Lip(\bF;\Ball^d(1)) > u_1 C_0\pmax\sqrt{d\log(\pmax)}\Big)&\leq \exp\big(-C_0^2d\log(\pmax)(u_1-1)^2/2\big)\,,\\
	\P\Big(\Lip(\bD\bF;\Ball^d(1)) > u_1 C_0\pmax^{2}\sqrt{d\log(\pmax)}\Big)&\leq \exp\big(-C_0^2d\log(\pmax)(u_1-1)^2/2\big)\,.
	\end{align*}
Similarly to the proof of Lemma \ref{lemma:BoundLip}, for $\bx_1,\bx_2\in\S^{d-1}$, writing $\bU_1:=\bU_{\bx_1}$, $\bU_2:=\bU_{\bx_1,\bx_2}$,
\[
\|\bD\bF(\bx_1)\bU_1-\bD\bF(\bx_2)\bU_2\|_{\op}\le \|\bD\bF(\bx_1)-\bD\bF(\bx_2)\|_{\op}+\|\bD\bF(\bx_2)\|_{\op}\|\bx_1-\bx_2\|\,.
\]
By a union bound,
\begin{align*}
	\P\Big(\Lip_{\perp}(\bD\bF;\S^{d-1}) > u_1 C_0\pmax^{2}\sqrt{d\log(\pmax)}\Big)&\leq 2\exp\big(-C_0^2d\log(\pmax)(u_1-1)^2/8\big)\,.
\end{align*}
Since at the origin $\|\bD \bF(\bzero)\|_{\op}=0$, the bound of Item \ref{enum:good_plarge_2} follows from that of Item \ref{enum:good_plarge_1}.
	
	For Item \ref{enum:good_plarge_3}, from Proposition \ref{prop:sigmamin_MSS} and Markov's inequality, for some universal $C>0$,
	\[
	\P\Big(\exists \bx\in\S^{d-1}:\,\bF(\bx)=\bzero,\,\sigma_{\min}(\bD\bF(\bx)\bU_{\bx})\le \sqrt{u_2}d^{-\frac32}\pmax^{-\frac{d}{4}}
	\Big) \le C u_2\,.       
	\]
	
	Finally, we deal with Item \ref{enum:good_plarge_4}. To this end we first bound the size of the set $|\hat \calN_{d,k}|$. We note that the number of spherical caps of Euclidean diameter $2^{-k-1}$ that can be packed into the unit sphere is bounded by $(2^{k})^{d-1}\Vol(\S^{d-1})/\Vol(\Ball^{d-1})\le c\sqrt{d}2^{k(d-1)}$, for some universal constant $c$. Hence, the sphere can be covered the same number of spherical caps of Euclidean diameter $2^{-k}$, or by $3^d\cdot c\sqrt{d}2^{k(d-1)}$ set of the form $\bx+[0,2^{-k}]^d$ with $\bx\in\calN_{d,k}$.
	 
	Using this, for $X\sim\normal(0,\bI_{d-1})$,
	\begin{align*}
		&\E	\left|\left\{\hat\bx\in \hat\calN_{d,k}:\, \|\bF(\hat\bx)\|\le \sqrt{d}2^{-k} \cdot  u_1 C_0\pmax\sqrt{d\log(\pmax)} \right\} \right| \\
		&\le c \sqrt{d}3^d2^{k(d-1)}\cdot  \P\Big(
		\|X\|\le \sqrt{d}2^{-k} \cdot  u_1 C_0\pmax\sqrt{d\log(\pmax)}
		\Big)\\
		&\le c \sqrt{d}3^d2^{k(d-1)} \cdot(2\pi)^{-\frac{d-1}{2}}\mbox{Vol}(\Ball^{d-1}(1))\big(\sqrt{d}2^{-k} \cdot  u_1 C_0\pmax\sqrt{d\log(\pmax)}\big)^{d-1}\\
		&\le \big(  u_1 \tilde C_0\pmax\sqrt{d\log(\pmax)}\big)^{d-1} = (u_1 T_0)^{d-1}
		\,,
	\end{align*}
where $\tilde C_0>0$ is an appropriate universal constant and we used that, by Stirling's approximation
\[
\mbox{Vol}(\Ball^{d-1}(1))=\frac{\pi^{\frac{d-1}{2}}}{\Gamma((d+1)/2)}=(1+o_d(1))\frac{\pi^{\frac{d-1}{2}}2^{\frac{d}{2}}e^{\frac{d+1}{2}}}{\sqrt{2\pi}(d+1)^{\frac{d}{2}}}\,.
\]

Hence, by Markov's inequality, the probability that the event in Item \ref{enum:good_plarge_4} does not occur is bounded by
$k_0 /u_3$.
\end{proof}

\section{\label{sec:Newton}Sufficient conditions for approximate solutions}

Throughout the section, we consider the Newton iteration introduced in Section 
\ref{subsec:apxsol}, and denoted formally by 
\begin{align}\label{eq:Newton-Repeat-0}
\bx^0=\bx\, ,\;\;\;\; \bx^{k+1}=\NM(\bx^k)\, ,
\end{align}
for a general initialization $\bx\in \S^{d-1}$.
More explicitly, we have $\NM(\bx^k)= (\bx^k+\bv^k)/\|\bx^k+\bv^k\|$, where $\bv^k$ solves
\begin{equation}
\label{eq:Newton-Repeat}
\bF(\bx^k)+\bD\bF(\bx^k)\bv^k=\bzero \hspace{.4cm}\mbox{subject to}\hspace{.4cm}\bv^k\perp \bx^k,\,\bv^k\perp\ker(\bD\bF(\bx^k))\,.
\end{equation}
In this section we will always work on the event that $\sigma_{\min}(\bD\bF(\bx^k)|_{\Ts_{\bx}})>0$, so that a solution to the above indeed exits. 
Assuming that there exists at least one non-zero solution to $\bF(\bx)=\bzero$, 
we will denote by
    \[
    \balpha^k:=\argmin\big\{ \|\bx^k-\balpha\|:\,\balpha\in\S^{d-1},\, \bF(\balpha)=\bzero \big\}
    \]
    the closest solution to $\bx^k$ (if more than one solution with minimal distance exists, choose one arbitrarily). 
We will prove sufficient conditions for a point $\bx\in\S^{d-1}$ to be an approximate solution. We shall treat the two cases $\pmax \le d^2$ and $\pmax >d^2$ separately. In both we will rely on the following.

\begin{proposition}[Quadratic convergence]\label{prop:Newton_explicit_init}
Assume that $\bF(\bx)=\bzero$ has at least one non-zero solution. 
Let $L:=\Lip(\bD\bF;\S^{d-1})\vee \Lip_{\perp}(\bD\bF;\S^{d-1})$ as defined in \eqref{eq:LipD} and \eqref{eq:LipDF_perp},  $M:=\sup_{\bx\in\S^{d-1}}\|\bD\bF(\bx)\bx\|$ and $T:=\inf_{k\ge0} \sigma_{\min}(\bD\bF(\balpha^k)|_{\Ts_{\balpha^k}})$. Assume that $T>0$  and
set
$B:= 3+(8L+4M)/T $.
If the initial point $\bx^0\in\S^{d-1}$ satisfies $\|\bx^0-\balpha^0\|\le 1/B$,
then 
\begin{equation}\label{eq:quadconv}
  \forall k\geq 0:\quad \|\bx^{k+1}-\balpha^{k+1}\|\le \|\bx^{k+1}-\balpha^{k}\|\le B \|\bx^{k}-\balpha^{k}\|^2\,.
\end{equation}
\end{proposition}

\begin{proof}
    The first inequality in \eqref{eq:quadconv} follows by definition. We claim that to prove the proposition, it will be enough to show that
    \begin{equation}\label{eq:Newton_basic_estimate}
        \forall k\geq0,\quad   \|\bx^k+\bv^k-\balpha^k\| \le B \|\bx^{k}-\balpha^{k}\|^2\,.
    \end{equation}
    Indeed,  this will imply that $\|\bx^k+\bv^k-\balpha^k\|\le \|\balpha^k\|= 1$ and  $\langle \balpha^k, \bx^k+\bv^k\rangle\ge 0$, from which one can easily verify that 
    \[
     \|\bx^{k+1}-\balpha^k\|\le \|\bx^k+\bv^k-\balpha^k\|\,.
    \]

     To prove \eqref{eq:Newton_basic_estimate} by induction, we let $k\ge 0$ be an arbitrary integer number. Assuming that \eqref{eq:Newton_basic_estimate} holds for any $0\le i\le k-1$, we will show that it holds for $k$. Note that by assumption, $\|\bx^0-\balpha^0\|\le 1/B$. Since we assume \eqref{eq:Newton_basic_estimate} for $0\le i\le k-1$, we have from the above that for $0\le i\le k$, $\|\bx^i-\balpha^i\|$ is a decreasing sequence and therefore also $\|\bx^k-\balpha^k\|\le 1/B$.

   Denote $V(\bx):=\ker(\bD\bF(\bx)\big|_{\Ts_{\bx}})\subset\Ts_{\bx}$ and let $V(\bx)^{\perp}:=\{\bu\in \Ts_{\bx}:\,\bu\perp V(\bx) \}$ be its orthogonal complement in $\Ts_{\bx}$. Denote by $P_{V(\bx^k)}$, $P^{\perp}_{V(\bx^k)}$ and $P_{\bx^k}$ the orthogonal projection matrices onto $V(\bx^k)$, $V(\bx^k)^{\perp}$ and $\bx^k$.
   We claim that
   \begin{align*}
        \sigma_{\min}(\bD\bF(\bx^k)\big|_{\Ts_{\bx^k}}) &= \sigma_{\min}(\bD\bF(\bx^k)\bU_{\bx^k})=\sigma_{\min}(\bD\bF(\bx^k)P^{\perp}_{V(\bx^k)})\,,\\
       \sigma_{\min}(\bD\bF(\balpha^k)\big|_{\Ts_{\balpha^k}}) &=\sigma_{\min}(\bD\bF(\balpha^k)\bU_{\balpha^k})=\sigma_{\min}(\bD\bF(\balpha^k))\,.
    \end{align*}
    The first equality in both lines follows by definition. From homogeneity of $F_i(\bx)$ and since $\bF(\balpha^k)=\bzero$, the rows of $\bD\bF(\balpha^k)$ are orthogonal to $\balpha^k$. The second equality in the second line follows from this. The second equality in the first line holds whenever $\sigma_{\min}(\bD\bF(\bx^k)\bU_{\bx^k})>0$. Indeed, using $\|\bx^k-\balpha^k\|\le \frac{1}{B}$,
    $\sigma_{\min}(\bD\bF(\bx^k)\bU_{\bx^k})\ge T-L\|\bx^k-\balpha^k\|\ge T/2$. Note that by assumption,
    $\sigma_{\min}(\bD\bF(\balpha^k))\ge T$. 
    
    Hence, 
    to prove \eqref{eq:Newton_basic_estimate}, using that $P_{V(\bx^k)}+P^{\perp}_{V(\bx^k)}+P_{\bx^k}=\id$ it will be sufficient to  show that
    \begin{align}
    \|P^{\perp}_{V(\bx^k)}(\bx^k+\bv^k-\balpha^k)\| & \le \frac{2(L+M)}{\sigma_{\min}(\bD\bF(\bx^k)P^{\perp}_{V(\bx^k)})}\|\bx^k-\balpha^k\|^2\,,\label{eq:Newton_estimate1}\\ 
    \|P_{V(\bx^k)}(\bx^k+\bv^k-\balpha^k)\| & \le \Big( 2+\frac{4L}{\sigma_{\min}(\bD\bF(\balpha^k))}\Big)\|\bx^k-\balpha^k\|^2 \,,\label{eq:Newton_estimate2}\\
    \|P_{\bx^k}(\bx^k+\bv^k-\balpha^k)\| & \le \frac12\|\bx^k-\balpha^k\|^2\,.\label{eq:Newton_estimate3}
    \end{align}

     Denote by $d(\bx,\by)=\arccos(\<\bx,\by\>)$ be the geodesic distance on $\S^{d-1}$ and let $t_k:=d(\balpha^k,\bx^k)$.
     For $t\in[0,t_k]$, let $\bx^k(t)\in\S^{d-1}$ be a geodesic from $\bx^k$ to $\balpha^k$ with constant speed $1$ (so that, in particular, $\bx^k(0)=\bx^k$ and $\bx^k(t_k)=\balpha^k$).
    Using \eqref{eq:Newton-Repeat} we write
    \begin{equation}
    \label{eq:Falphak}
    \begin{aligned}
        0=\bF(\balpha^k)&=\bF(\bx^k)+\int_0^{t_k}\frac{\de \bF(\bx^k(t))}{\de t} \de t\\
        &= -\bD\bF(\bx^k)\bv^k + \int_0^{t_k} \bD\bF(\bx^k(t))\dot{\bx}^k(t) \de t\\
        &= -\bD\bF(\bx^k)P^{\perp}_{V(\bx^k)}\bv^k +\bD\bF(\bx^k)P^{\perp}_{V(\bx^k)}(\balpha^k-\bx^k)\\
        &+\underbrace{\int_0^{t_k} \big[\bD\bF(\bx^k(t))-\bD\bF(\bx^k)\big]\dot{\bx}^k(t) \de t}_{=:\bR_1}\\
        &+\underbrace{\bD\bF(\bx^k)P_{V(\bx^k)}(\balpha^k-\bx^k)}_{=:\bR_2}+\underbrace{\bD\bF(\bx^k)P_{\bx^k}(\balpha^k-\bx^k)}_{=:\bR_3}\,,
    \end{aligned}
    \end{equation}
    where integrals of vector valued functions are carried out elementwise.

    By the definition of $V(\bx)$, $\bD\bF(\bx^k)P_{V(\bx^k)}$ is the zero matrix and $\bR_2=0$. Moreover,
    \[
    \|\bR_1\|\le t_k \sup_{t\le t_k} \|\bD\bF(\bx^k(t))-\bD\bF(\bx^k)\|_{\op} \le d(\balpha^k,\bx^k) \cdot L \|\bx^k-\balpha^k\|\le 2 L \|\bx^k-\balpha^k\|^2
    \]
    and
    \[
    \|\bR_3\|\le \|\bD\bF(\bx^k)\bx^k\| \cdot |\langle\bx^k,\bx^k-\balpha^k\rangle| = \|\bD\bF(\bx^k)\bx^k\| \cdot \| \bx^k-\balpha^k\|^2 /2\le M \| \bx^k-\balpha^k\|^2\,.
    \]

     Continuing from \eqref{eq:Falphak} we have that
     \[
     0=\bD\bF(\bx^k)P^{\perp}_{V(\bx^k)}(\balpha^k-\bx^k-\bv^k)+\bR\,,
     \]
    where $\bR:=\bR_1+\bR_2+\bR_3$ and $\|\bR\|\le2(L+M)\|\bx^k-\balpha^k\|^2$, from which 
    \eqref{eq:Newton_estimate1} follows.

    Next, we prove \eqref{eq:Newton_estimate2}. Since we assume that $T>0$, note that by the implicit function theorem, for any $\bu\in V(\balpha^k)$ there exist small $\eps>0$ and a function $\bvphi:(-\eps,\eps)\to\R^d$ such that $\|\bvphi(t)\|=O(t^2)$ and, for any $t\in(-\eps,\eps)$, 
    \[
    \bgamma(t):=\frac{\balpha^k+t\bu+\bvphi(t)}{\|\balpha^k+t\bu+\bvphi(t)\|}\in \S^{d-1}
    \]
    is a solution (i.e., $\bF(\bgamma(t))=\bzero$). If there existed a vector $\bu\in V(\balpha^k)$ such that $\langle \dot \bx^k(t_k),\bu\rangle<0$, we would have $\|\bgamma(t)-\bx^k\|<\|\balpha^k-\bx^k\|$ for sufficiently small $t>0$. We conclude that no such vector exists, i.e., that $\dot \bx^k(t_k)\perp V(\balpha^k)$ or, equivalently, $P_{V(\balpha^k)}\dot \bx^k(t_k)=0$. From homogeneity of $F_i(\bx)$, at the solution $\balpha^k$ we have that $\ker(\bD\bF(\balpha^k))=V(\balpha^k)\oplus{\rm span}\{\balpha^k\}$. Since $\dot \bx^k(t_k)\perp \balpha^k$, we also have that $P_{\ker(\bD\bF(\balpha^k))}\dot \bx^k(t_k)=0$.

   Using that $\bv^k\perp V(\bx^k)$, we write
    \begin{align*}
        P_{V(\bx^k)}(\bx^k+\bv^k-\balpha^k) &= P_{V(\bx^k)}(\bx^k-\balpha^k) = - \int_0^{t_k} P_{V(\bx^k)}\dot\bx^k(t)\de t \\
        &= - \int_0^{t_k} P_{V(\bx^k)}(\dot\bx^k(t)-\dot\bx^k(t_k))\de t -t_k P_{V(\bx^k)} \dot\bx^k(t_k)\,.
    \end{align*}
    Noting that $\|\dot\bx^k(t)-\dot\bx^k(t_k)\|$ is maximal over $[0,t_k]$ for $t=0$ and recalling that $\|\dot \bx^k(t_k)\|=1$,
    \begin{align*}
        \|P_{V(\bx^k)}(\bx^k+\bv^k-\balpha^k)\| &\le t_k \|\dot\bx^k(0)-\dot\bx^k(t_k)\|+t_k\|P_{V(\bx^k)} (\id - P_{\ker(\bD\bF(\balpha^k))})\|_{\op}\,.
    \end{align*}
    Since $t_k\leq 2\|\bx^k-\balpha^k\|$, the first term is  bounded by 
    \[
      t_k  \|\dot\bx^k(0)-\dot\bx^k(t_k)\| =  t_k\|\bx^k(0)-\bx^k(t_k)\|\le 2\|\bx^k-\balpha^k\|^2 \,.
    \]

    For any $\bx\in\S^{d-1}$, $V(\bx)\subset \ker(\bD\bF(\bx))$. Thus, 
    \begin{equation}\label{eq:Pker}
    \|P_{V(\bx^k)} (\id - P_{\ker(\bD\bF(\balpha^k))})\|_{\op}
    \leq 
    \|P_{\ker(\bD\bF(\bx^k))} (\id - P_{\ker(\bD\bF(\balpha^k))})\|_{\op}
    \,.
    \end{equation}
    Applying Theorem \ref{thm:Wedin} (specifically, Eq.~\eqref{eq:WedinFirst}) with $A_0=\bD\bF(\bx^k)$, $A_1=\bD\bF(\balpha^k)$ and $k(a)={\rm rank}(A_a)$ so that $\Delta=\sigma_{\min}(\bD\bF(\balpha^k))>0$ (by Eq.~\eqref{eq:sigmaminbd3}), we obtain that the right-hand side of \eqref{eq:Pker} 
    is bounded as follows:
    \begin{align*}
    \|P_{\ker(\bD\bF(\bx^k))} (\id - P_{\ker(\bD\bF(\balpha^k))})\|_{\op}&\le 2\frac{\|\bD\bF(\bx^k)-\bD\bF(\balpha^k)\|}{\sigma_{\min}(\bD\bF(\balpha^k))}\\
    &\le \frac{2L\|\bx^k-\balpha^k\|}{\sigma_{\min}(\bD\bF(\balpha^k))}
    \,.
    \end{align*}
    Combining the above and using again that $t_k\leq 2\|\bx^k-\balpha^k\|$, we obtain \eqref{eq:Newton_estimate2}.

    Finally, \eqref{eq:Newton_estimate3} follows since
    \[
    \|P_{\bx^k}(\bx^k+\bv^k-\balpha^k)\|=\|P_{\bx^k}(\bx^k-\balpha^k)\|=1-\langle\bx^k,\balpha^k\rangle = \frac12\|\bx^k-\balpha^k\|^2\,.
    \]
\end{proof}

\begin{corollary}
\label{cor:Newton_x0bd}
Assume that $\bF(\bx)=\bzero$ has at least one non-zero solution. 
Define $L$ and $M$ as in Proposition \ref{prop:Newton_explicit_init} and 
\[
T:= \sigma_{\min}(\bD\bF(\bx^0)\big|_{\Ts_{\bx^0}})-5L\|\bx^0-\balpha^0\|\,.
\]
Assume that $T>0$ and set $B:= 3+(8L+4M)/T$.
If the initial point $\bx^0\in\S^{d-1}$ satisfies $\|\bx^0-\balpha^0\|\le 1/4B$,
then 
\begin{equation}\label{eq:quadconv2}
  \forall k\geq 0:\quad \|\bx^{k+1}-\balpha^{k+1}\|\le \|\bx^{k+1}-\balpha^{k}\|\le B \|\bx^{k}-\balpha^{k}\|^2\,.
\end{equation}    
\end{corollary}
\begin{proof}
    As in Proposition \ref{prop:Newton_explicit_init}, we will prove \eqref{eq:quadconv2} by induction on $k$. Here, however, we do not assume the uniform  bound on the minimal singular value at $\balpha^k$ and $k\geq0$.

    We start with $k=0$. By definition, $\sigma_{\min}(\bD\bF(\balpha^0)\big|_{\Ts_{\balpha^0}})>T$. Observe that the proof of \eqref{eq:quadconv2} in  Proposition \ref{prop:Newton_explicit_init} for $k=0$ only relied the latter bound on the singular value at $\balpha^0$ but did not use the same bound for $\balpha^i$ with $i>0$. Hence, by exactly the same argument as in Proposition \ref{prop:Newton_explicit_init}, \eqref{eq:quadconv2} follows for $k=0$.

    Now, assume that \eqref{eq:quadconv2} holds for any $0\le i\le k-1$. Then,  by induction, for any $1\le i\le k$,
\begin{equation}\label{eq:NewtonDist1}
  \|\bx^{i}-\balpha^{i}\|\le \|\bx^{i}-\balpha^{i-1}\|\le \|\bx^{0}-\balpha^{0}\| \left(B \|\bx^{0}-\balpha^{0}\|\right)^{2^{i}-1}
  \le \|\bx^{0}-\balpha^{0}\| 
  \Big(
  \frac14
  \Big)^{2^{i}-1}
\end{equation}
and therefore
\begin{equation}\label{eq:NewtonDist2}
\|\bx^{i}-\bx^{i-1}\|\le\|\bx^{i}-\balpha^{i-1}\| + \|\bx^{i-1}-\balpha^{i-1}\| \le \frac54\|\bx^{0}-\balpha^{0}\| 
  \Big(
  \frac14
  \Big)^{2^{i-1}-1}\,.
\end{equation}
    Thus,
    \begin{align*}  
    \|\balpha^{i}-\bx^{0}\|& \le \|\balpha^{i}-\bx^{i}\|+ \|\bx^{i}-\bx^{0}\|+\|\bx^{0}-\balpha^{0}\|\\
    &\le \|\bx^{0}-\balpha^{0}\|\Big[
    \Big(
    \frac14
    \Big)^{2^{i}-1}+\frac54 \sum_{j=1}^i  \Big(
  \frac14
  \Big)^{2^{j-1}-1}+1
    \Big]
    \le 5 \|\bx^{0}-\balpha^{0}\| \,.
    \end{align*}
    And we conclude that
    \begin{equation}
        \label{eq:bdT}
        \sigma_{\min}(\bD\bF(\balpha^i)\big|_{\Ts_{\balpha^i}})\ge \sigma_{\min}(\bD\bF(\bx^0)\big|_{\Ts_{\bx^0}}) - 5L \|\bx^{0}-\balpha^{0}\|=T\,.
    \end{equation}

   The proof of  \eqref{eq:quadconv2} for $k$ follows by the same induction argument we used in the proof of Proposition \ref{prop:Newton_explicit_init}, which relied  on the bound $\sigma_{\min}(\bD\bF(\balpha^i)\big|_{\Ts_{\balpha^i}})\ge T$  with $i\le k$ only and not on the same bound for $i>k$.
\end{proof}

\subsection{Approximate solutions for Hessian Descent}
The following sufficient condition for an approximate solution will be used in the analysis of the Hessian Descent algorithm (Algorithm \ref{algo:HD}).

\begin{proposition}
    \label{thm:apxsol_smallp}
    Assume the same setting and notation as in Definition \ref{def:good} and that  points \ref{enum:good1}, \ref{enum:good2} and \ref{enum:good3} in the definition hold. Let $\tau$ be defined 
    as in \eqref{eq:tau} and let
    \begin{equation}\label{eq:B}
    B:= \frac{c\pmax^2\sqrt{d\log \pmax}}{\tau},\qquad c:=15(C_0\vee C_1 \vee 1) \,.
    \end{equation} 
    If $\|\bx-\balpha\|\le \frac1{4B}$ and $\bF(\balpha)=\bfzero$ for some $\bx,\balpha\in \S^{d-1}$,  then $\bx$ is an approximate solution.
\end{proposition}

\begin{proof}
Let  $\bx^0=\bx$ and define $\bx^{i+1}=\NM(\bx^i)$ by Newton's method 
as in Eqs.~\eqref{eq:Newton-Repeat-0}, \eqref{eq:Newton-Repeat}.
Let $L$, $M$ and $T$ be as defined in Proposition \ref{prop:Newton_explicit_init}.  By Points \ref{enum:good2} and \ref{enum:good3} of Definition \ref{def:good}, $L\le C_0\pmax^2\sqrt{d\log(\pmax)}$ and $T\ge \tau$.
Note that for any $\bx\in\S^{d-1}$, the $i$-th element of $\bD\bF(\bx)\bx$ is the radial derivative of $F_i$ at $\bx$ which is equal to $p_i F_i(\bx)$, since $F_i$ is homogeneous. Thus, by Point \ref{enum:good1} of Definition \ref{def:good}, $M\le C_1 \pmax \sqrt{d\log(\pmax)}$. Hence, for $B$ as defined in \eqref{eq:B}, $\frac{8L+4M}{T} + 3\le B$. By Proposition \ref{prop:Newton_explicit_init}, it follows that \eqref{eq:quadconv} holds.

Therefore, as in \eqref{eq:NewtonDist1} and \eqref{eq:NewtonDist2},
\begin{equation*}
  \|\bx^{k}-\balpha^{k}\|
  \le \|\bx^{0}-\balpha^{0}\| 
  \Big(
  \frac14
  \Big)^{2^{k}-1}\,\,\,\mbox{and}\,\,\,
  \|\bx^{k}-\bx^{k-1}\| \le \frac54\|\bx^{0}-\balpha^{0}\| 
  \Big(
  \frac14
  \Big)^{2^{k-1}-1}\,.
\end{equation*}
Thus, the limit $\bar\balpha=\lim_{k\to\infty}\bx^k\lim_{k\to\infty}\balpha^k$ is well-defined,  $\bF(\bar\balpha)=0$ and, for $k\ge1$,
\[
\|\bx^{k}-\bar\balpha\|\le
\sum_{j=k}^{\infty}\|\bx^{j+1}-\bx^{j}\|\le \frac54 \|\bx^{0}-\balpha^{0}\| 
  \sum_{j=k}^{\infty}4^{1-2^{j}}\le \Big(\frac54\Big)^2 \|\bx^{0}-\balpha^{0}\| 
  4^{1-2^{k}}
  \le  \|\bx^{0}-\bar\balpha\| 
  2^{1-2^{k}}\,.
\]
\end{proof}

\subsection{Approximate solutions for Multi-Scale Search}
The following sufficient condition for an approximate solution will be used in the analysis of the Multi-Scale Search algorithm (Algorithm \ref{algo:MMS}).

\begin{corollary}\label{cor:apxsol_largep}
	Suppose that $n= d-1$ and let $u_2\le 1\le u_1$. Then there exists $\delta=\delta(C_0)$ such that, assuming that the bounds of Points  \ref{enum:good_plarge_1} and \ref{enum:good_plarge_2}   in Definition \ref{def:good2} occur with the same universal constant $C_0$ as there, any $\bx\in\S^{d-1}$ that satisfies 
	\begin{equation}\label{eq:goodnetpt}
		\|\bF(\bx)\|\le \frac{\delta u_2/u_1}{d^{7/2}\pmax^{d/2+2}\sqrt{\log(\pmax)}} \mbox{\,\, and \,\, } \sigma_{\min}(\bD\bF(\bx)\bU_{\bx})\ge \frac{\sqrt{u_2}}{4d^{\frac32}\pmax^{\frac{d}{4}}} 
	\end{equation}
	is an approximate solution. 
\end{corollary}

\begin{proof}
	Assume that $\bx\in\S^{d-1}$ satisfies \eqref{eq:goodnetpt}, for some $0<\delta$ which  will be assumed to be small enough whenever needed and may depend on $C_0$. 
	First, we will show that there exists a solution close to $\bx$ by appealing to  \cite[Lemma E.1]{montanari2023solving} about gradient flow.  
	
	Since, assuming $\delta$ is small,
	\begin{align*}
		\Lip_{\perp}(\bD\bF;\S^{d-1})\|\bF(\bx)\| &\le u_1 C_0\pmax^{2}\sqrt{d\log(\pmax)} \cdot \frac{\delta u_2/u_1}{d^{7/2}\pmax^{d/2+2}\sqrt{\log(\pmax)}} \\
		&\le \frac1{64}u_2 d^{-3}\pmax^{-\frac{d}{2}} \le \frac1{4} \big(\sigma_{\min}(\bD\bF(\bx)\bU_{\bx})\big)^2\,,
	\end{align*}
	by  \cite[Lemma E.1]{montanari2023solving}, gradient flow $\bx(t)$ on $\S^{d-1}$ started from $\bx(0)=\bx$ converges to a solution $\balpha:=\lim_{t\to\infty}\bx(t)$. Namely, $\bF(\balpha)=\bfzero$. By Eq. (162) of \cite{montanari2023solving},
	\[
	\|\bx-\balpha\|\le 2\frac{\|\bF(\bx)\|}{\sigma_{\min}(\bD\bF(\bx)\bU_{\bx})}\le 
	\frac{8 \delta \sqrt{u_2}/u_1}{d^{2}\pmax^{d/4+2}\sqrt{\log(\pmax)}} 
	\,.
	\]
	
	For Newton's method started from $\bx^0=\bx$, in the notation of Corollary \ref{cor:Newton_x0bd}, $L\le \Lip^{(1)}(\bF)\le u_1 C_0\pmax^{2}\sqrt{d\log(\pmax)}$, $M\le u_1 C_0\pmax\sqrt{d\log(\pmax)}$
	(see the proof of Proposition \ref{thm:apxsol_smallp} for the last inequality),
	\begin{align*}
		T&\ge \frac{\sqrt{u_2}}{4d^{\frac32}\pmax^{\frac{d}{4}}}  - 5 u_1 C_0\pmax^{2}\sqrt{d\log(\pmax)} \cdot \frac{8\delta \sqrt{u_2}/u_1}{d^{2}\pmax^{d/4+2}\sqrt{\log(\pmax)}}\\
		& 
		\ge 
		\sqrt{u_2} d^{-\frac32}\pmax^{-\frac{d}{4}}(1/4 - 40 C_0\delta)
		\, .
	\end{align*}
	Therefore, for small $\delta$, we have $T>0$ and
	\[
	B\le 12u_1 C_0  \pmax^{2}\sqrt{d\log(\pmax)} \cdot 8 
	d^{\frac{3}{2}}\pmax^{\frac{d}{4}}/\sqrt{u_2}
	+3
	\]
	and 
	\[
	\|\bx^0-\balpha\|\le\frac{1}{4B}\,.
	\]

	Hence, \eqref{eq:quadconv2} follows, by Corollary \ref{cor:Newton_x0bd}. The fact that $\bx$ is an approximate solution follows by the same argument in the last paragraph of the proof of Proposition \ref{thm:apxsol_smallp}.
\end{proof}

\section{\label{sec:subroutines}Sub-routines and their complexity}
In this section we define and analyze sub-routines that are used by the main algorithms. The first is for finding vectors $\bv_i$ as in Algorithm \ref{algo:HD}.  Its pseudo-code  is given below. Here $\be_i$ denotes the standard basis of $\R^d$ and $C_1$ and $c_0$ are  constants as in  Lemma \ref{rmk:MaxHessian} and Algorithm \ref{algo:HD}. (The value of $c_0$ will be determined in the proof of Theorem \ref{thm:main1}.) The absolute constant $c=c(C_1,c_0)>0$ in the pseudo-code will be determined in the proof of Lemma \ref{lem:find}.\vspace{.3cm}

\begin{algorithm}[H]
	\caption{Find a good direction}\label{algo:Find}
	\SetAlgoLined
	\KwIn{The coefficients $\ba$, a point $\bx\in\S^{d-1}$}
	\KwOut{A unit vector $\bv\in \S^{d-1}$ orthogonal to $\bx$}
        $\mu=9C_1d\pmax^2 \log(\pmax)$\,\;
        $\bA = (\bI-\bx\bx^{\top})(\mu\bI-\nabla^2H(\bx))(\bI-\bx\bx^{\top})$\,\;
        $k = 1$\,\;
        \While{$k<e^{c(d+\log\pmax})$}
        {$\bA = \bA^2$,\, $k=2k$\,\;}
        $i_0=\argmin_{i\le d} \<\nabla^2H(\bx),(\bA\be_i/\|\bA\be_i\|)^{\otimes2}\>$\,\;
	\KwRet $\bv=\bA\be_{i_0}/\|\bA\be_{i_0}\|$
\end{algorithm}\vspace{.3cm}

In our running time analysis below, we compare the number of operation with
the input size $N$ (the number of coefficients to specify the function $\bF$). Recall that
\begin{align}
N=\sum_{i=1}^{n}\binom{d+p_i-1}{p_i}\, .
\end{align}
 
\begin{lemma}\label{lem:find}
   The running time of Algorithm \ref{algo:Find} is  $O(Nd^3 \pmax )$. On the event that Points \ref{enum:good4} and \ref{enum:good5} of Definition \ref{def:good} hold, for any $\bx\in \S^{d-1}$ such
    that  $H(\bx)\geq \pmax^{-c_0d}$, the algorithm
    outputs a vector $\bv$ such that $\bv\perp \bx$, $\|\bv\|=1$ and
    \begin{equation}\label{eq:apxminev}
    \langle\nabla^2H(\bx_i),\bv^{\otimes2}\rangle\leq \frac{1}{2}\min_{\bu\perp\bx,\|\bu\|=1}\langle\nabla^2H(\bx),\bu^{\otimes2}\rangle\,.
    \end{equation}
 \end{lemma}
\begin{proof} 
    Let $\bv_1,\ldots,\bv_{d-1}$ and $\lambda_1\le \cdots\le \lambda_{d-1}$ be the eigenvectors and eigenvalues of $\nabla^2H(\bx)|_{\Ts_{\bx}}$. On  the event as in the statement of the lemma, $|\lambda_i|\le \mu$ and $\lambda_1\le -\frac{1}{10}\sqrt{ (d-n) H(\bx)}$.
    We have that
    \begin{align*}
       \bA^k\be_i = \sum_{j\le d-1} (\mu-\lambda_j)^k\<\be_i,\bv_j\>\bv_j 
    \end{align*}
    and
    \begin{align*}
        \<\nabla^2H(\bx),(\bA^k\be_i/\|\bA^k\be_i\|)^{\otimes2}\> &= \frac{\sum_{j\le d-1} (\mu-\lambda_j)^{2k}\<\be_i,\bv_j\>^2\lambda_j}{\sum_{j\le d-1} (\mu-\lambda_j)^{2k}\<\be_i,\bv_j\>^2} \\
        & \le \frac{\sum_{j\in I} \Big(\frac{\mu-\lambda_j}{\mu-\lambda_1}\Big)^{2k}\<\be_i,\bv_j\>^2\lambda_j + \mu e^{-\frac{k|\lambda_1|}{4\mu}}}{\sum_{j\in I} \Big(\frac{\mu-\lambda_j}{\mu-\lambda_1}\Big)^{2k}\<\be_i,\bv_j\>^2+e^{-\frac{k|\lambda_1|}{4\mu}}}\,,
    \end{align*}
    provided that the numerator in the second ratio is negative, where $I$ is the set of indices such that $\lambda_j\le \frac{3}{4}\lambda_1$ and we used  that, for $j\notin I$,
    \[
    \frac{\mu-\lambda_j}{\mu-\lambda_1}\le 1-\frac{|\lambda_1|}{8\mu}\le e^{-\frac{|\lambda_1|}{8\mu}}\,.
    \]
    
    For some $i$, $\<\be_i,\bv_1\>^2\ge 1/d$. Hence, \eqref{eq:apxminev} follows since $|\lambda_1|/(10d)\ge \mu e^{-\frac{k|\lambda_1|}{4\mu}}$ and $1/(10d)\ge e^{-\frac{k|\lambda_1|}{4\mu}}$ if $k<e^{c(d+\log\pmax)}$ and $c=c(C_1,c_0)$ is sufficiently large.

     The value $F_i(\bx)$ can be computed from the coefficients $\ba$ by $c_1 N\pmax$ operations, for some constant $c_1$. Any first or second order derivative of $F_i(\bx)$ can be obtained with the same  running time. Therefore, $\bA$ can be computed in $c_1'(Nd^3\pmax)$ operations.
     The number of iterations in the while-loop is bounded by $1+c(d+\log\pmax)/\log2$ and each iteration requires $O(d^3)$ operations. The complexity of running  the while-loop can therefore be absorbed in the previous term.     
\end{proof}

The next sub-routines  approximate the maximal and minimal singular value of a matrix.
Denote by $\omega_L$ the time complexity for $L\times L$ matrix multiplication.

\begin{lemma}\label{lem:maxsv}
    There are an absolute constant $c>0$ and an algorithm that, given a matrix $\bA\in \R^{M\times L}$,  computes in time complexity $c(\log\log M\cdot\omega_M+\omega_L))$   a value 
    \begin{equation}\label{eq:smax}
        s_{\max}\in\Big[\,\frac12\sigma_{\max}(\bA),\sigma_{\max}(\bA)\,\Big]\,.
    \end{equation}
\end{lemma}
\begin{proof}
    Define $\bB=\bA\bA^\top$ and denote its eigenvalues by $0\le \lambda_1\le\cdots\le \lambda_M$. Let $k=2^s$ with $s=\lceil\log_2\log M\rceil$.
    We claim that $s_{\max}:=\max_{i\le N}\|\bB^k \be_i\|^{1/k}$ satisfies
    $\mu\in[\frac12\lambda_M,\lambda_M]$.
    Indeed, denoting by $\bv_j$ the eigenvector corresponding to $\lambda_j$, $\lambda_{M}^{2k}\ge \|\bB^k \be_i\|^2=\sum_j\lambda_j^{2k}\<\be_i,\bv_j\>^2\ge \lambda_{M}^{2k}\<\be_i,\bv_{M}\>^2$.
    The claim thus follows since $\max_i \<\be_i,\bv_{M}\>^2\ge 1/M$.
     Exploiting the fact that $k$ is a power of $2$ as in Algorithm \ref{algo:Find}, the time complexity for computing $\mu$ is as stated in the lemma. We note that computing the $k$-th root for $k$ a power of two can be done by iterative computing the square root, which we assume as part of our computational model.
\end{proof}

\begin{lemma}\label{lem:minsv}
    There are an absolute constant $c>0$ and an algorithm that, given a symmetric matrix $\bA\in \R^{M\times L}$ and $\kappa>1$,  computes in time complexity $c((\log\kappa+\log\log M)\omega_M+\omega_L))$   a value 
    \begin{equation}\label{eq:smin}
        s_{\min}\in\Big[\sigma_{\min}(\bA),\sigma_{\min}(\bA)+\sqrt{2\sigma_{\max}(\bA)(1-M^{-1/2\kappa})}\,\Big]\,.
    \end{equation}
\end{lemma}
\begin{proof}
    Define $\bB=\bA\bA^\top$, $\lambda_i$ and $s_{\max}$ as in the previous proof.
    Define $\bD=2s_{\max}\bI-\bB$ and set $\kappa'=2^{\lceil\log_2\kappa\rceil}$. Similarly to the previous proof, $S:=\max_i \|\bD^{\kappa'}\be_i\|^{1/\kappa'}\in [M^{-1/2\kappa'}(2s_{\max}-\lambda_1),2s_{\max}-\lambda_1]$. Therefore, $s_{\min}:=\sqrt{2s_{\max}-S}$ satisfies \eqref{eq:smin}.

    The time complexity for computing $s_{\max}$ is $O(\log\log M\cdot\omega_M+\omega_L))$. Exploiting the fact that $\kappa'$ is a power of $2$ as in Algorithm \ref{algo:Find}, the running time for computing $s_{\min}$ given $s_{\max}$ is $O(\log(\kappa)\omega_M)$.
\end{proof}

\section{\label{sec:HD}Analysis of Hessian descent: Proof of Theorem \ref{thm:main1}}

We start with an analysis of the Hessian Descent algorithm 
(Algorithm \ref{algo:HD}). In the theorem below we will assume that $(\bx_i,\by_i,\bv_i,\delta_i)_{i\ge0}$ is an infinite sequence as follows. For $i=0$, $\bx_0=\by_0=(1,0,\ldots,0)$. For any $i\ge0$, $\bv_i$ is a unit vector such that $\bv_i\perp\bx_i$,
\begin{equation}\label{eq:HDcond1}
    \langle\nabla^2H(\bx_i),\bv_i^{\otimes2}\rangle\leq -\frac{1}{20}\sqrt{ (d-n) H(\bx_i)},
 \end{equation}
$\by_{i+1}=\bx_i\pm\delta_i\bv_i$, with the sign chosen such that $H(\by_{i+1})=\min\{H(\bx_i+\delta_i\bv_i),H(\bx_i-\delta_i\bv_i)\}$, and $\bx_{i+1}=\by_{i+1}/\|\by_{i+1}\|$. $\delta_i$ is as defined in Algorithm \ref{algo:HD}.

\begin{theorem}\label{thm:Algorithm}
    Suppose the sequence $(\bx_i,\by_i,\bv_i,\delta_i)_{i\ge0}$ is as above and that the event in Point  \ref{enum:good5} of Definition \ref{def:good} occurs. For some universal constants $c_1,c_2$ we have the following. If 
    \begin{equation}\label{eq:deltamax}
        \sqrt{ (d-n) H(\bx_i)}\ge 30C_1 d\pmax^3\log(\pmax) 
    \end{equation}
     then 
    \begin{equation}
    \label{eq:Hdiff}
    H(\bx_{i+1})-H(\bx_{i})\le  -\frac38C_1 d\pmax^2\log(\pmax)\,.
    \end{equation}
    Otherwise, if \eqref{eq:deltamax} does not hold, for any $j\ge i$,
		\begin{equation}\label{eq:HBd}
		H(\bx_j)\le H(\bx_i) \cdot \Big(1- \frac{c_1}{\pmax^4\log(\pmax)} \frac{d-n}{d}\Big)^{j-i}\, 
		\end{equation}
  and
    \begin{equation}\label{eq:distBd}
			\big\|\bx_i-\bx_j\big\|_2\le c_2H(\bx_i)^{\frac14} \pmax^2\sqrt{d\log(\pmax)}\, .
		\end{equation}
\end{theorem}	
\begin{proof}
 By Taylor's theorem,
\begin{equation}\label{eq:Hdiff2}
\begin{aligned}
	H(\by_{i+1})-H(\bx_{i})&\overset{(a)}{\le}\sum_{k=1}^{3}\frac{\delta_i^k}{k!}\langle \nabla^kH(\bx_i),\bv_i^{\otimes k}\rangle  +\frac{\delta_i^4}{4!} 3C_1d \pmax^4 \log(\pmax)
	\\
	&\overset{(b)}{\le} \frac{\delta_i^2}{2}\langle \nabla^2H(\bx_i),\bv_i^{\otimes 2}\rangle  +\frac{\delta_i^4}{4!} 3C_1 d \pmax^4 \log(\pmax) \\
	&\overset{(c)}{\le}-\frac{\delta_i^2}{40}\sqrt{(d-n) H(\bx_i)}+\frac{\delta_i^4}{4!} 3C_1 d \pmax^4 \log(\pmax)\\
	&\overset{(d)}{\le}-\frac{\delta_i^2}{80}\sqrt{(d-n) H(\bx_i)}\,,
\end{aligned}
\end{equation}
where for $(a)$ we used that $\delta_i\le \sqrt{1/\pmax}$ and Point \ref{enum:good5} of Definition \ref{def:good} and $C_1$ is as there,
for $(b)$ we used the fact that the bound of $(a)$ holds whether we take $\by_{i+1}=\bx_i+\bv_i$ or $\by_{i+1}=\bx_i-\bv_i$ and we choose the sign that minimizes $H(\by_{i+1})$, for  $(c)$ we used \eqref{eq:HDcond1}, and $(d)$ follows from the choice of $\delta_i$ as in Algorithm \ref{algo:HD}.

We have that  $\delta_i=\sqrt{1/\pmax}$ if and only if \eqref{eq:deltamax} holds,
in which case, using homogeneity for the first inequality,
\[
    H(\bx_{i+1})-H(\bx_{i})\le H(\by_{i+1})-H(\bx_{i})\le  -\frac38C_1 d\pmax^2\log(\pmax)\,,
\]
and \eqref{eq:Hdiff} follows.

Now, suppose \eqref{eq:deltamax} does not hold for some given $i$. Then, using \eqref{eq:Hdiff2} and the definition of $\delta_i$,
\[
    H(\bx_{i+1})-H(\bx_{i})\le H(\by_{i+1})-H(\bx_{i})\le -\frac{1}{2400C_1}\frac{1}{\pmax^4\log(\pmax)}\frac{d-n}{d} H(\bx_i)\,,
\]
and for $c_1=1/2400C_1$,
\[
    H(\bx_{i+1})\le H(\bx_i)\Big(1-\frac{c_1}{\pmax^4\log(\pmax)}\frac{d-n}{d} \Big)\,,
\]
which completes the proof of \eqref{eq:HBd}.

For any $j>i$,
\begin{align*}
    \|\bx_j-\bx_i\|&\le \sum_{k=i}^{j-1}\|\bx_{k+1}-\bx_k\| \le \sum_{k=i}^{j-1}\|\by_{k+1}-\bx_k\| = \sum_{k=i}^{j-1}\delta_k \\
    &= \bigg( \frac{1}{30C_1}\frac{1}{\pmax^4\log(\pmax)} \frac{ \sqrt{ (d-n) }}{d}\bigg)^{\frac12} \sum_{k=i}^{j-1} H(\bx_k)^{\frac14}\\
    &\le H(\bx_i)^{\frac14}\bigg( \frac{1}{30C_1}\frac{1}{\pmax^4\log(\pmax)} \frac{ \sqrt{ (d-n) }}{d}\bigg)^{\frac12} \bigg(1-\Big(1-\frac{c_1}{\pmax^4\log(\pmax)}\frac{d-n}{d}\Big)^{\frac14} \bigg)^{-1}\\
    &\le H(\bx_i)^{\frac14}\bigg( \frac{1}{30C_1}\frac{1}{\pmax^4\log(\pmax)} \frac{ \sqrt{ (d-n) }}{d}\bigg)^{\frac12} \bigg(\frac14\frac{c_1}{\pmax^4\log(\pmax)}\frac{d-n}{d}\bigg)^{-1}\,,
\end{align*}
which proves \eqref{eq:distBd} with $c_2=\sqrt{8/(15C_1c_1^2)}$.
\end{proof}

\subsection{Proof of Theorem \ref{thm:main1}}

Consider the Hessian Descent algorithm (Algorithm \ref{algo:HD}), with $k=C_0' d^{3/2}\pmax^4\log(\pmax)^2$ iterations for some absolute $C_0'>0$ to be determined below. Recall that to find a vector $\bv_i$ we use the sub-routine from Section \ref{sec:subroutines}. By Lemma \ref{lem:find}, the time complexity for a single iteration is $O(Nd^3 \pmax )$ and the total time complexity is 
\[
\chi \le C_0 Nd^{9/2}\pmax^5  \log(\pmax)^2
\]
for appropriate $C_0$, which proves \eqref{eq:BoundComplexity}.

Next, we will assume that the good event $\calE(d,n)$ of Definition \ref{def:good} occurs and prove that the output is a point
$\bx^{\sHD}\in\S^{d-1}$ (and not FALSE) which is an approximate solution (recall Definition \ref{def:apxsol}). Indeed, in light of the lower bound on $\P(\calE(d,n))$ in Lemma \ref{lem:goodeventprob}, this will complete the proof of the theorem.

Suppose that $i_0\leq k$ is the largest index for which $\bx_{i_0}$ is computed  before the algorithm terminates.
By Lemma 
\ref{lem:find}, on $\calE(d,n)$, for all indices $i<i_0$, since $H(\bx_i)\ge \pmax^{-c_0d}$, 
\[
\langle\nabla^2H(\bx_i),\bv_i^{\otimes2}\rangle\leq {\displaystyle\frac{1}{2}\min_{\bu\perp\bx_i,\|\bu\|=1}}\langle\nabla^2H(\bx_i),\bu^{\otimes2}\rangle.
\]
From the bound of Point \ref{enum:good4} in Definition \ref{def:good}, $\bv_i$ therefore also satisfies \eqref{eq:HDcond1}.
By the same point in Definition \ref{def:good}, for $i\ge i_0$ there exists $\bv_i$ such that \eqref{eq:HDcond1} holds. Hence, we may extend the sequence
$(\bx_i,\by_i,\bv_i,\delta_i)_{i\le i_0}$ produced by the Algorithm \ref{algo:HD} to an infinite sequence $(\bx_i,\by_i,\bv_i,\delta_i)_{i\ge0}$ as in the setting of Theorem \ref{thm:Algorithm}.

From \eqref{eq:Hdiff} and the bound $H(\bx_0)\leq d$ as in Point \ref{enum:good6} of Definition \ref{def:good}, for some absolute constant $c>0$, for any $i\geq cd$ the condition in \eqref{eq:deltamax} does not hold.


By \eqref{eq:HBd} and \eqref{eq:distBd}, $\bx_i$ is a Cauchy sequence, and $\balpha:=\lim_{i\to\infty}\bx_i\in\S^{d-1}$ is an exact solution, $\bF(\balpha)=\bzero$. Moreover, for any $i\geq cd$,
\begin{equation*}
    H(\bx_i)\le d \cdot \Big(1- \frac{c_1}{\pmax^4\log(\pmax)} \frac{d-n}{d}\Big)^{i-cd}\, .
\end{equation*}
Therefore, if the number of steps $k$ is as above with sufficiently large $C_0'$, then of some $i\le k$ we have that $H(\bx_i)\le \pmax^{-c_0d}$, and the algorithm outputs the point $\bx^{\sHD}=\bx_i\in\S^{d-1}$ (and does not output FALSE). By \eqref{eq:distBd}, 
\begin{align*}
\|\bx^{\sHD}-\balpha\| &\le c_2 \pmax^{-c_0d/4}\pmax^2\sqrt{d\log(\pmax)}
\,.
\end{align*}
If $c_0$ is sufficiently large, then in the notation of Proposition \ref{thm:apxsol_smallp} we have that $\|\bx^{\sHD}-\balpha\| \le 1/4B$ and $\bx^{\sHD}$ is therefore an approximate solution.   
\qed

\section{\label{sec:MMS}Multi-Scale Search: proof of Theorem \ref{thm:main2}}

	Let $0<u_2\le 1\le u_1$ and $u_3>0$ and assume that $n= d-1$. With 
	$C_0$ and $\tilde C_0$ the universal constants  from Definition \ref{def:good2}, let $\delta=\delta(C_0)$ be the small constant from Corollary \ref{cor:apxsol_largep} which WLOG we assume to be less than $1$.  Here we consider the Multi-Scale Search algorithm (Algorithm \ref{algo:MMS}) with $\bx=(-1,\ldots,-1)$, $k=-1$,
\begin{align*}
	k_0 &=\Big\lceil \frac1{\log2}  \log\Big(\frac{u_1^2}{\delta u_2} C_0\pmax^{d/2+3} d^{9/2}\log(\pmax) \Big) \Big\rceil,\label{eq:k0}\\
	\calL &= u_1 C_0\pmax\sqrt{d\log(\pmax)}, \\
	\calS &= \frac12\frac{\sqrt{u_2}}{d^{\frac32}\pmax^{\frac{d}{4}}}.
\end{align*}

To make precise the computation of the approximation $s_{\min}(\bD\bF(\hat\bx)\big|_{\Ts_{\hat\bx}})=s_{\min}(\bD\bF(\hat\bx) \bU_{\hat\bx})$ of the minimal singular value in Algorithm \ref{algo:MMS}, we suppose that it is implemented using the algorithm of Lemma \ref{lem:minsv} with 
\begin{equation}\label{eq:kappa}
\kappa=16 \frac{u_1}{u_2}  C_0 d^{\frac72}\log(d)\pmax^{\frac{d}{2}+1}\sqrt{\log(\pmax)}\vee 1\,.
\end{equation}

Recall Definition \ref{def:good2} of the event $\calE=\calE(d,u_1,u_2,u_3, k_0)$. 
The fact that $\P(\calE)$ is upper bounded by the expression in \eqref{eq:MSS_complexitybound}  follows from Lemma \ref{lem:goodevent2} and the choice of $k_0$ as above. Henceforth we assume that $\calE$ occurs. To complete the proof of Theorem \ref{thm:main2}, in  Parts I and II below we prove the correctness of the algorithm on $\calE$, from which the second part of the theorem follows. In Part III we prove the complexity bound as in the first part of the theorem.\vspace{.2cm}

\noindent\textbf{Part I: return value is $\bx^{\salg}\in\S^{d-1}\implies \bx^{\salg}$ is an approximate solution.} Assume that the output is $\bx^{\salg}\in\S^{d-1}$ and not FALSE. Note that the only way for this to happen is if there is a sequence of recursive calls  $(\mbox{MSS($\ba,\bx_k,\calL,\calS,k,k_0$)})_{k=-1}^{k_0}$ with $\bx_{k}\in\calN_{d,k}$ such that
\begin{equation*}
	\bx^{\salg}\in \bx_{k_0}+[0,2^{-k_0}]^d\subset \bx_{k_0-1}+[0,2^{-k_0+1}]^d\subset\cdots\subset \bx_{-1}+[0,2]^d=[-1,1]^d
\end{equation*}
and such that 
the return value of all those calls is the same. Hence, by considering the return value of the last call with $k=k_0$,
\[
\|\bF(\bx^{\salg})\|\le \sqrt{d}2^{-k_0}\calL\le \frac{\delta u_2/u_1}{d^{7/2}\pmax^{d/2+2}\sqrt{\log(\pmax)}} \quad {\rm and}\,\quad s_{\min}(\bD\bF(\bx^{\salg})\big|_{\Ts_{\bx^{\salg}}})\ge \calS\,.
\]

We note that Point \ref{enum:good_plarge_2} of Definition \ref{def:good2} implies that  $\sigma_{\max}(\bD\bF(\bx) \bU_{\bx})\le u_1 C_0\pmax\sqrt{d\log(\pmax)}$ on $\calE$ for any $\bx\in\S^{d-1}$. Using that $1-d^{-1/2\kappa}=1-e^{-\log(d)/2\kappa}\le\log(d)/2\kappa$, from \eqref{eq:smin} one can verify that therefore
\begin{equation}\label{eq:sminDFU}
	s_{\min}(\bD\bF(\bx) \bU_{\bx}) \in\Big[\sigma_{\min}(\bD\bF(\bx) \bU_{\bx}),\sigma_{\min}(\bD\bF(\bx) \bU_{\bx})+\frac12\calS  
	\,\Big]\,.
\end{equation}
Thus,
\[
\sigma_{\min}(\bD\bF(\bx^{\salg})\big|_{\Ts_{\bx^{\salg}}})\ge \frac12\calS=\frac14 
\frac{\sqrt{u_2}}{d^{\frac32}\pmax^{\frac{d}{4}}} \,.
\]
Therefore, $\bx^{\salg}$ is an approximate solution by Corollary  \ref{cor:apxsol_largep}.\vspace{.15cm}

\noindent\textbf{Part II: $\exists\balpha\in\S^{d-1}:\,\bF(\balpha)=\bzero  \implies$  return value is $\bx^{\salg}\in\S^{d-1}$.} Assume that $\balpha\in\S^{d-1}$ and $\bF(\balpha)=\bzero $. Let $\bx_k\in \calN_{d,k}$, $k=-1,0,\ldots,k_0$, be a sequence (a.s. unique) such that 
\begin{equation*}
	\balpha\in \bx_{k_0}+[0,2^{-k_0}]^d\subset \bx_{k_0-1}+[0,2^{-k_0+1}]^d\subset\cdots\subset \bx_{-1}+[0,2]^d=[-1,1]^d\,.
\end{equation*}
As in the MSS pseudo-code, we define $$\tilde\bx_k=\argmin\big\{\|\by\|:\,\by\in\bx_k+\{0,2^{-k}\}^d\big\}$$
the corner of the block $\bx_k+[0,2^{-k}]^d$ closest to the origin, and its projection to the sphere
$\hat\bx_k=\tilde\bx_k/\|\tilde\bx_k\|$. Note that
\begin{equation}\label{eq:nonepmty}
	\balpha,\,\hat\bx_k\in\big(\bx_k+[0,2^{-k}]^d\big)\cap \S^{d-1}\,.
\end{equation} 

We will prove by induction that the output of MSS($\ba,\bx_k,\calL,\calS,k,k_0$) is not {\rm FALSE} for any $k=-1,0,\ldots,k_0$. Our base case is $k=k_0$. Since $\big(\bx_{k_0}+[0,2^{-k_0}]^d\big)\cap \S^{d-1}\neq\varnothing$ by \eqref{eq:nonepmty}, to show that the output is not {\rm FALSE} we need to show that
\begin{equation}\label{eq:k0bounds}
\|\bF(\hat\bx_{k_0})\|\le \sqrt{d}2^{-k_0}\calL,  \qquad s_{\min}(\bD\bF(\hat\bx_{k_0})\big|_{\Ts_{\hat\bx_{k_0}}})\ge \calS\,.
\end{equation}
Indeed, the first inequality follows from \eqref{eq:nonepmty} since $\|\bF(\balpha)\|=0$, $\mbox{diam}(\bx_{k_0}+[0,2^{-{k_0}}]^d)=\sqrt{d}2^{-k_0}$ and  $\Lip(\bF;\Ball^d(1))\le \calL$ by Point \ref{enum:good_plarge_1}   of Definition \ref{def:good2}.
By Points \ref{enum:good_plarge_1} and \ref{enum:good_plarge_3}  of Definition \ref{def:good2},
\begin{align*}
	\sigma_{\min}(\bD\bF(\hat\bx_{k_0})\big|_{\Ts_{\hat\bx_{k_0}}})&\ge \sigma_{\min}(\bD\bF(\balpha)\big|_{\Ts_{\balpha}})-\mbox{diam}([0,2^{-k_0}]^d)\cdot \Lip_{\perp}(\bD\bF;\S^{d-1})\\
	&\ge\sqrt{u_2} d^{-\frac32}\pmax^{-\frac{d}{4}}-\sqrt{d}2^{-k_0}\cdot u_1 C_0\pmax^2\sqrt{d\log(\pmax)}\ge \calS\,,
\end{align*}
where the last inequality assumes that $\delta\le 1$. The second bound in \eqref{eq:k0bounds} follows by \eqref{eq:sminDFU}.

We proceed to the induction step. Suppose that the output of MSS($\ba,\bx_{k+1},\calL,\calS,k+1,k_0$) is not FALSE for some $0\le k+1\leq k_0$. When running the pseudo-code for MSS($\ba,\bx_{k},\calL,\calS,k,k_0$) the first IF condition is not satisfied by \eqref{eq:nonepmty}. As $k<k_0$, to complete the proof we only need to show that $\|\bF(\hat\bx_k)\|\le \sqrt{d}2^{-k}\calL$, since then one of the recursive calls is MSS($\ba,\bx_{k+1},\calL,\calS,k+1,k_0$). Indeed, this follows by the same argument we used to justify the first bound of \eqref{eq:k0bounds}.\vspace{.15cm}

\noindent\textbf{Part III: Complexity analysis.} We will prove the complexity bound as in Theorem \ref{thm:main2} assuming that $\calE=\calE(d,u_1,u_2,u_3,\delta, k_0)$ occurs. Indeed, since we prove the correctness of the algorithm only on this event, we may stop the algorithm if the running time exceeds this bound without affecting the proofs in the previous two parts.

We wish to analyze the MSS algorithm with $\bx=(-1,\ldots,-1)$, $k=-1$ and $\calL$, $\calS$ and $k_0$ as above. 
Of course, its time complexity $\chi$ depends on the running
 time of the recursive calls with $k=-1,0,\ldots,k_0$ and $\bx\in\calN_{d,k}$. For each $k$, define $\chi_k$ as the total running time spent in all calls of the form MSS($\ba,\bx,\calL,\calS,k,k_0$) with some $\bx\in\calN_{d,k}$, excluding the time for running the for-loop in each such call (if we reach it) which itself includes the recursive calls of the form MSS($\ba,\bx,\calL,\calS,k+1,k_0$) for $\bx\in\calN_{d,k+1}$.  
 We claim that, up to negligible terms,
 \[
 \chi = \chi_{-1}+\chi_{0}+\cdots+\chi_{k_0}\,.
 \]
What we neglect here is the running time for all commands in the for-loops excluding the recursive calls to MSS, which is smaller than the time to run the recursive calls themselves.

First, we discuss the time complexity of a single call to MSS with $k<k_0$ excluding the running time of its for-loop, which we denote by $\bar\chi_k$. In such a call to MSS, the vector $\tilde \bx$ can be computed using $\tilde x_i=\argmin\{ |t|:\,t\in \{x_i,x_i+2^{-k}\}\}$  in $O(d)$ operations. Precisely, the big O notation that we use throughout means in $Cd$ operations for some universal constant $C$.\footnote{In principle, we also need to compute $2^{-k}$. This can be done for all $k=-1, ,\ldots,k_0$ in the first call to MSS with $k=-1$ in $O(k_0)$ operations, instead of repeating the calculation of $2^{-k}$ in each recursive call. Doing so would make the complexity of this computation negligible relative to the bound we obtain on $\chi$ so we may safely neglect this point in our current discussion.} 
Computing the farthest corner from the origin $\tilde \bx'$ with $\tilde x'_i=\argmax\{ |t|:\,t\in \{x_i,x_i+2^{-k}\}\}$ is done with the same complexity. And by checking whether $\|\tilde\bx\|\le 1\le \|\tilde\bx'\|$, we can decide if $\big(\bx+[0,2^{-k}]^d\big)\cap \S^{d-1}\neq\varnothing$ or not. 
Recall that computing $F_i(\bx)$ for a given $\bx$ can be done in $O(N\pmax)$ operations and therefore computing $\|\bF(\hat\bx)\|$ can be done in $O(Nd\pmax)$  operations. Combining the above we have that $\bar\chi_k=O(Nd\pmax)$.

Denote by $\bar\chi_{k_0}$ the running time of a single call to MSS with $k=k_0$. Here, the for-loop is irrelevant. Since we can compute each derivative $\partial_j F_i(\bx)$ in $O(N\pmax)$ operations, we can compute $\bD\bF(\hat \bx)\bU_{\hat \bx}$ in $O(Nd^3\pmax)$ operations. $s_{\min}(\bD\bF(\hat\bx)\big|_{\Ts_{\hat\bx}})=s_{\min}(\bD\bF(\hat\bx) \bU_{\hat\bx})$ is computed from $\bD\bF(\hat \bx)\bU_{\hat \bx}$ using the algorithm of Lemma \ref{lem:minsv} with $\kappa$ as in \eqref{eq:kappa} by $O(d^3\log \kappa)$ additional operations. Hence, $\bar\chi_{k_0}=O(Nd^3\pmax+d^3\log \kappa)$.

Finally, note that each $(k,\bx)$ with $\bx\in\calN_{d,k}$, there is at most one
recursive call to MSS with $0\le k\le k_0$ and $\bx$. Hence,
by Point \ref{enum:good_plarge_4} of Definition \ref{def:good2}, the number of recursive calls with a given value of $k$ is bounded by
\[
2^d\cdot u_3(u_1 \tilde C_0\pmax\sqrt{d\log(\pmax)} )^{d-1}\,.
\]
The $2^d$ factor accounts for the number of iterations in the for-loop and the other factor bounds the number of calls with $k-1$ and some $\bx\in \calN_{d,k-1}$ such that both $\big(\bx+[0,2^{-(k-1)}]^d\big)\cap \S^{d-1}=\varnothing$ and $\|\bF(\hat\bx)\|>D\cdot \calL$.

Combining all of the above, the total time complexity is bounded by
\begin{align*}
\chi&\le 2^d u_3\big(u_1 \tilde C_0\pmax\sqrt{d\log(\pmax)} \big)^{d-1}\cdot O\big( k_0 Nd\pmax + Nd^3\pmax+d^3\log \kappa\big)\\
&\le O\Big(N u_3\big(u_1 \tilde C_0'\pmax\sqrt{d\log(\pmax)} \big)^{d-1}\big(
\log u_1\vee \log(1/u_2) \vee \log(1/\delta)\vee d\log(\pmax) 
\big)
\Big)\,.
\end{align*}
Since $\delta$ is a constant (depending only on $C_0$), we may omit the $\log(1/\delta)$ term. \qed

\section{\label{sec:BF}Proof of Theorem \ref{thm:main0}}

To prove the theorem, we define an algorithm that either uses the Hessian Descent or Multi-Scale Search algorithms, (Algorithm \ref{algo:HD} and Algorithm \ref{algo:MMS}) 
depending on the values of $d$ and $\pmax$. For the Multi-Scale Search algorithm the choice of $u_1,u_2,u_3$ will also depend on $d$ and $\pmax$. 

Recall that we analyzed the two algorithms in Theorems \ref{thm:main1} and \ref{thm:main2}. 
Assume that $C$ and $C_0$ are the maximum of the same universal constants  in the latter two theorems, so that both hold with these maximal values. Let $A$ be the absolute constant from Theorem \ref{thm:main1}.

Let $\delta>0$ and define the following  sets 
\begin{align*}
	\Lambda_1&:=\big\{(d,p):\, p<d^2,\, Ce^{-d/C}<\delta\big\}\,,\\
	\Lambda_2&:=\big\{(d,p):\, p<d^2,\, Ce^{-d/C}\ge\delta\big\}\,,\\
	\Lambda_3&:=\big\{(d,p):\, p\ge d^2,\, \pmax^{-d}<\delta\big\}\,,\\
	\Lambda_4&:=\big\{(d,p):\, p\ge d^2,\, \pmax^{-d}\ge \delta\big\}\,.
\end{align*} 

Note that $\Lambda_2$ and $\Lambda_4$ are finite sets. Hence, we can choose some $u_1,u_2,u_3$ such that for all $(d,\pmax)\in\Lambda_2\cup\Lambda_4$, the expression in \eqref{eq:MSS_complexitybound} is at least $1-\delta$.
We use the Multi-Scale Search algorithm with this choice of $u_1,u_2,u_3$, for any $(d,\pmax)\in\Lambda_2\cup\Lambda_4$. By construction, this implies the required bound on the probability that the algorithm solves $\bF(\bx)=\bzero$. Using again the fact that $\Lambda_2\cup\Lambda_4$ is a finite set, the time complexity for $(d,\pmax)$ in this set is bounded by some constant $C_0(\delta)$ (which implicitly depends on $u_i$).

If $(d,\pmax)\in\Lambda_1$, we use the Hessian Descent algorithm. Recall that by Footnote \ref{ft:delta0}, assuming that $\delta_0$ is sufficiently small guarantees that $n=\lfloor d- A(d\log d)^{1/2}\rfloor\ge 1$. Here we need to show that the algorithm solves $\bF(\bx)=\bzero$ at least with the required probability of $1-Ce^{-d/C}$, and this indeed follows by Theorem \ref{thm:main1}. Moreover, by the same theorem, the complexity $\chi$ is bounded as in \eqref{eq:BoundComplexity}, 
\begin{align*}
	\chi \le C_0 Nd^{9/2}\pmax^5  (\log\pmax)^2\overset{(a)}{\le} 8C_0 N^{3.25}\pmax^5  (\log\pmax)^2\overset{(b)}{\le} C_0'N^{4.25}\, .
\end{align*}
where $(a)$ follows since $N\ge \binom{d+1}{2}\ge d^2/2$ and $(b)$, with appropriate absolute constant $C_0'$,  since 
\[
N\ge\binom{d+\pmax-1}{\pmax} \ge\Big(\frac{d+\pmax-1}{\pmax}\Big)^{\pmax}\Big(\frac{d+\pmax-1}{d-1}\Big)^{d-1}
\]
and therefore, for any $t>0$,
\begin{equation}\label{eq:pmaxvsN}
\lim_{d\to\infty}\sup_{\pmax\ge2}\frac{\pmax}{N^{t}}=0\,.
\end{equation}

Finally, if $(d,\pmax)\in\Lambda_3$,   we use the Multi-Scale Search algorithm with 
\begin{align*}
u_1&:=1+\sqrt{ C(\log6/d +\log\pmax)} \le C'\sqrt{ \log\pmax}\,,\\
u_2&:=\pmax^{-d}/3C\,,\quad u_3:= C\pmax^{d}(d\log\pmax+ C'')/3\,,
\end{align*}
where $C',C''>0$ are universal constants. Assuming that $C''$ is sufficiently large, note that this choice implies that each of the three terms subtracted in \eqref{eq:MSS_complexitybound} is bounded by $\pmax^{-d}/3$, and therefore, with $n=d-1$, the algorithm solves $\bF(\bx)=\bzero$ with probability at least $1-\pmax^{-d}$. 

By \eqref{eq:BoundComplexity2}, the complexity in this case is bounded by
\begin{align*}
	\chi&\le N C \pmax^d \big(C' C_0\sqrt{d}\pmax\log(\pmax) \big)^{d-1}(d\log\pmax+ C'')^2
	\\
	&\le N^3 c^d d^{\frac52d+\frac12} \pmax (\log\pmax)^{d+1} 
	\le N^{5.5} c_0^d d^{\frac{17}{4}} \pmax (\log\pmax)^{d+1}
	\,,
\end{align*}
where $c,c_0,c'>0$ are universal constants and we used that by Stirling's approximation
\[
N\ge \binom{d+\pmax-1}{\pmax}\ge c'\frac{1}{\sqrt{d}} \frac{\pmax^{d-1}}{(d-1)^{d-1}}\ge
c' d^{d-\frac32} \,.
\]
By the last inequality, we also have that
\[
\lim_{d\to\infty}\sup_{\pmax\ge d^2}\frac{c_0^d d^{\frac{17}{4}}}{N^{\delta'}}=0\,,
\]
which combined with \eqref{eq:pmaxvsN} yields that
\[
\pushQED{\qed} 
\chi\le C_0(\delta,\delta')N^{5.5+\delta'}\,.\qedhere
\popQED
\]

\section*{Acknowledgements}

AM was supported by the NSF through award DMS-2031883, the Simons Foundation through Award
814639 for the Collaboration on the Theoretical Foundations of Deep Learning, the NSF grant CCF2006489 and the ONR grant N00014-18-1-2729. 
ES was supported by the Israel Science Foundation theough grant
No. 2055/21 and a research grant from the Center for Scientific Excellence at the Weizmann Institute of Science. ES is the incumbent of the Skirball Chair in New Scientists.

\newpage

\appendix

\section{Tools: Linear algebra}

\begin{theorem}[Wedin \cite{wedin1972perturbation}]\label{thm:Wedin}
Let $\bA_0, \bA_1\in\reals^{m\times n}$ have singular value decomposition (for $a\in\{0,1\}$)
\begin{align}
\bA_a = \bU_a\bSigma_a\bV_a^{\sT}\, , 
\end{align}
with $\bSigma_a$ containing the singular values in decreasing order.
Further let $\bU_{a,+}\in\reals^{m\times k(a)}$, $\bV_{a,+}\in\reals^{n\times k(a)}$,   be formed by the first $k(a)$ columns of  
$\bU_a$, $\bV_a$, so that 
\begin{align}
\bU_a= \big[\bU_{a,+}\big|\bU_{a,-}\big]\, ,\;\;\;\;\;\;\bV_a= \big[\bV_{a,+}\big|\bV_{a,-}\big]\, .
\end{align}
Finally assume $\Delta\equiv \sigma_{k(1)}(\bA_1)-\sigma_{k(0)+1}(\bA_0)>0$. 
Let $\bP_{a} = \bV_{a,+}\bV_{a,+}^{\sT}$ (respectively $\bQ_{a} = \bU_{a,+}\bU_{a,+}^{\sT}$) denote the projector onto the right singular space 
(left singular space) corresponding to large singular values of $\bA_a$. Then we have
\begin{align}
\big\|(\id_n-\bP_0)\bP_1\big\|_{\op}\le \frac{1}{\Delta}\,
\left\{\big\|(\id-\bQ_0)(\bA_0-\bA_1)\bP_1\big\|_{\op}\vee \big\|\bQ_1(\bA_0-\bA_1)(\id -\bP_0)\big\|_{\op} \right\}\, ,\label{eq:WedinFirst}
\end{align}
If instead we have  $\Delta\equiv \sigma_{k(0)}(\bA_0)- \sigma_{k(1)+1}(\bA_1)>0$, then
\begin{align}
\big\|\bP_0(\id_n-\bP_1)\big\|_{\op}\le \frac{1}{\Delta}\,
\left\{\big\|(\id-\bQ_1)(\bA_0-\bA_1)\bP_0\big\|_{\op}\vee \big\|\bQ_0(\bA_0-\bA_1)(\id -\bP_1)\big\|_{\op} \right\}\, . \label{eq:WedinSecond}
\end{align}
\end{theorem}

\begin{corollary}\label{coro:Wedin}
Let $\bA_0, \bA_1\in\reals^{n\times d}$ and denote by $\bv_1(a),\dots, \bv_n(a)$,
$a\in \{0,1\}$ two bases of right singular vectors of these matrices,
with associated singular values $\sigma_1(\bA_a)\ge \cdots\ge \sigma_d(\bA_a)$ 
(including vanishing singular values). Assume  $n<d$  and define, for $m\ge 1$,
\begin{align}
V_{0,m} &= {\rm span}(\bv_{n-m+2}(0),\dots, \bv_d(0))\, ,\\
V_{1,1} &= {\rm span}(\bv_{n+1}(1),\dots, \bv_d(1))\, .
\end{align}
(In particular, $V_{1,1}$ is a subspace of the null space of $\bA_1$.)
Finally, denote by $\bE_1\in\reals^{d\times (d-n)}$  a matrix whose columns form an 
orthonormal basis of $V_{1,1}$, and by $\obE_0\in\reals^{d\times (d-n+m-1)}$ 
a matrix whose columns form an orthonormal basis of $V_{0,m}^{\perp}$. 

Then we have 
\begin{align}
\big\|\bE_1^{\sT}\obE_0\big\|_{\op}\le \frac{1}{\sigma_{n-m+1}(\bA_0)}
\|\bA_0-\bA_1\|_{\op}\, .
\end{align}
\end{corollary}

\section{Tools: Random matrix theory}

Given the symmetric matrix $\bM\in\reals^{n\times n}$, we denote by
 $\lambda_1(\bM)\le \lambda_2(\bM)\le \dots\le \lambda_n(\bM)$ its the eigenvalues 
  in increasing order.
We denote by $F_{\rm sc}(t)= \frac{1}{2\pi}\int_{-2}^{(t\wedge 2) \vee (-2)}\sqrt{4-x^2}\, \de x$ the semicircle distribution.
\begin{lemma}\label{lemma:LD-Symm-GOE}
For any $t\in(-2,2)$ and $\eps>0$ there exists a constant $C_0=C_0(t,\eps)$ such that for  $\bW\sim \GOE(N)$ and any $k$ such that $k/N\le F_{\rm sc}(t)-\eps$,
\begin{align}
\prob\Big(\lambda_k(\bW)> t\sqrt{N}\Big)\le C_0\, e^{-N^2/C_0}\, .
\end{align}
\end{lemma}
\begin{proof}
Since $\lambda_k(\bW)>t\sqrt{N}$ if and only if $\#\{\lambda_i(\bW)\le t\sqrt{N}\}< k$, the lemma follows from the large deviation principle for the empirical measure of eigenvalues of a GOE matrix proved in \cite[Theorem 2.1.1]{benarous1997large}. 
\end{proof}

\begin{lemma}\label{lem:K-th-Eigengalue-Wishart}
For all $c_0>0$, there exists  $C_*(c_0)$ such that the following holds.
If $\bZ\sim\GOE(M,N)$, $c_0N< M\le N$  and  $1\le \ell\le M/2$
then, for any $\Delta>0$,
\begin{align}
\prob\Big(\lambda_{\ell}(\bZ\bZ^{\sT})\le \Delta^2\big(\sqrt{N}-\sqrt{M-1}\big)^2\Big) \le (C_*(c_0)\Delta)^{\ell(N-M+\ell)}\, .
\end{align}
\end{lemma}

\begin{proof}
Let $f(\bu) = f(u_1,\dots,u_M)$ denote the joint density 
of the ordered eigenvalues
$\lambda_1\le \dots\le \lambda_M$, $\lambda_i= \lambda_i(\bZ\bZ^{\sT})$. We have \cite[Proposition 4.1.3]{Guionnet}
\begin{align}
f(\bu)  = \frac{1}{Z_{M,N}}\prod_{i=1}^Mu_i^{(N-M-1)/2}e^{-u_i/2}
\prod_{i<j}(u_j-u_i)\, \bfone_{\bbJ_M}(\bu) \, ,
\end{align}
where $\bbJ_M=\bbJ_M(0)$ and $\bbJ_M(a)\subseteq \reals^M$ is the set
$\bbJ_M(a):= \{\bu:\; a\le u_1\le \cdots\le u_M\}$, 
and
\begin{align}
 Z_{M,N} = \left(\frac{2^N}{\pi}\right)^{M/2}
 \prod_{i=1}^M\Gamma\Big(\frac{N-i+1}{2}\Big)
 \Gamma\Big(\frac{M-i+1}{2}\Big)\, .
 \end{align}
 We denote joint density of the lowest $\ell$ eigenvalues 
 $\bu_{\le \ell}=(u_1,\dots,u_\ell)$ by $f_\ell$. 
 It is given by
\begin{align*}
&f_{\ell}(\bu_{\le \ell})  = \frac{F_{\ell}(\bu_{\le \ell})}{Z_{M,N}}
 \int_{\bbJ_{M-\ell}(u_\ell)}\, 
\prod_{i=\ell+1}^M\prod_{m=1}^\ell(u_i-u_m)
\prod_{i=\ell+1}^M
u_i^{(N-M-1)/2}e^{-u_i/2}\prod_{\ell<i<j }(u_j-u_i) 
\de\bu_{> \ell}\, ,\\
&F_{\ell}(\bu_{\le \ell}): = \prod_{i=1}^\ell
u_i^{\frac{N-M-1}{2}}e^{-\frac{u_i}{2}}
\prod_{i<j\le \ell }(u_j-u_i)\, \bfone_{\bbJ_\ell}(\bu_{\le \ell})\, .
\end{align*}
We can bound the integral as
\begin{align*}
f_{\ell}(\bu_{\le \ell})  &\le \frac{F_{\ell}(\bu_{\le \ell})}{Z_{M,N}}
 \int_{\bbJ_{M-\ell}(u_\ell)}\, 
\prod_{i=\ell+1}^M
u_i^{(N-M+2\ell-1)/2}e^{-u_i/2}\prod_{\ell<i<j }(u_j-u_i) 
\de\bu_{> \ell}\, \\
& \le \frac{F_{\ell}(\bu_{\le \ell})}{Z_{M,N}}
 \int_{\bbJ_{M-\ell}(0)}\, 
\prod_{i=\ell+1}^M
u_i^{(N-M+2\ell-1)/2}e^{-u_i/2}\prod_{\ell<i<j }(u_j-u_i) 
\de\bu_{> \ell}\\
& = \frac{ Z_{M-\ell,N+\ell}}{Z_{M,N}}F_{\ell}(\bu_{\le \ell})\,.
\end{align*}

Therefore
\begin{align*}
\P\Big(\lambda_\ell(\bZ\bZ^{\sT})\le \delta\Big)
&\le  \frac{ Z_{M-\ell,N+\ell}}{Z_{M,N}}\int F_{\ell}(\bu_{\le \ell})
\bfone_{u_\ell\le \delta}\de\bu_{\le \ell}\\
&\le  \frac{ Z_{M-\ell,N+\ell}}{Z_{M,N}}\frac{1}{\ell!}\int 
\prod_{i=1}^\ell
u_i^{\frac{N-M-1}{2}}
\prod_{i<j\le \ell }|u_j-u_i|
\prod_{i=1}^{\ell}\bfone_{0\le u_i\le \delta}
\de\bu_{\le \ell}\\
&\le  \frac{ Z_{M-\ell,N+\ell}}{Z_{M,N}}
\frac{1}{\ell!}\cuS_{\ell}\Big(\frac{N-M+1}{2},1,\frac{1}{2}\Big)
\cdot \delta^{\ell(N-M+\ell)/2}\, ,
\end{align*}
where $\cuS_{\ell}(\alpha,\beta,\gamma)$ is Selberg's integral
\cite[Theorem 2,5.8]{Guionnet}. In particular, we have 
\begin{align*}
\cuS_{\ell}(\alpha,1,1/2) =\prod_{i=0}^{\ell-1}
\frac{\Gamma(\alpha+(i/2))\Gamma(1+(i/2))\Gamma((3+i)/2)}
{\Gamma(\alpha+1+(\ell+i-1)/2)\Gamma(3/2)}\, .
\end{align*}
After reordering we can write the above bound as
\begin{align}
\P\Big(\lambda_\ell(\bZ\bZ^{\sT})\le \delta\Big)
\le \frac{1}{\Gamma(3/2)^\ell \ell!}\frac{Q_{M,N,\ell}}{R_{M,N,\ell}} S_{M,N,\ell}\, T_{M,N,\ell}
\cdot \delta^{\ell(N-M+\ell)/2}\, ,\label{eq:CompleteBound-KTH}
\end{align}
where 
\begin{align*}
Q_{M,N,\ell} & := 2^{-\frac{\ell}{2}(N-M)-\frac{\ell^2}{2}}\pi^{\ell/2}\, ,\\
R_{M,N,\ell} & :=\prod_{i=1}^{2\ell}\Gamma\big((N-M+i)/2\big)\, ,\\
S_{M,N,\ell} & := \prod_{i=1}^{\ell}
\frac{\Gamma\big((N+i)/2\big)}{\Gamma\big((M-\ell+i)/2\big)}
\, ,\\
T_{M,N,\ell} & := 
\prod_{i=1}^\ell\frac{\Gamma\big((N-M+i)/2\big)}{\Gamma\big((N-M+\ell+1+i)/2\big)}
\cdot\prod_{i=1}^\ell\Gamma\big((i+1)/2\big)\Gamma\big((i+2)/2\big)\, .
\end{align*}
%
We estimate the above quantities in the following
lemma.

\begin{lemma}\label{lemma:Constants}
For all $c_0>0$, there exists a constant $C_1=C_1(c_0)$ such that the following holds.
If $c_0N< M\le N$ and  $1\le \ell\le M/2$,
then
\begin{align}
\log \frac{T_{M,N,\ell}}{R_{M,N,\ell}} &\le -\ell(N-M+\ell)\log (N-M+\ell)+C_1\ell(N-M+\ell)\, ,\label{eq:logToverR}\\
\log S_{M,N,\ell} & \le \ell\frac{N-M+\ell}{2}\log M +
C_1\ell(N-M+1)\, .\label{eq:logS}
\end{align}
\end{lemma}

The proof of this lemma is postponed. We will now use it to prove the claim of the present one. 
Continuing from Eq.~\eqref{eq:CompleteBound-KTH}
(and denoting by $C$ constants that depend on $c_0$
and can change from line to line):
\begin{align*}
\prob\Big(\lambda_{\ell}(\bZ\bZ^{\sT})&\le \Delta^2\big(\sqrt{N}-\sqrt{M-1}\big)^2\Big) \le \prob\Big(\lambda_{\ell}(\bZ\bZ^{\sT})\le 
C\Delta^2\frac{(N-M)^2}{N}\Big) \\
& \le \frac{T_{M,N,\ell}}{R_{M,N,\ell}} S_{M,N,\ell}\cdot
(C\Delta)^{\ell(N-M+\ell)}\cdot
\left(\frac{(N-M)^2}{N}\right)^{\ell(N-M+\ell)/2}\\
& \le (C\Delta)^{\ell(N-M+\ell)} \exp(E_{M,N,\ell})\, ,
\end{align*}
where, for $C_1,C_2$ depending on $c_0$,
\begin{align*}
E_{M,N,\ell} :=&
 -\ell(N-M+\ell)\log (N-M+\ell)
+\ell\frac{N-M+\ell}{2}\,\log M+C_1\ell(N-M+\ell)\\
&+\ell(N-M+\ell)
\log(N-M) -\frac{1}{2}\ell(N-M+\ell)
\log N \\
\le &\, C_2\ell(N-M+\ell)\, .
\end{align*}
Substituting above, and adjusting the constant $C$ yields the claim.
\end{proof}

We now complete the last proof by proving Lemma \ref{lemma:Constants}.

\begin{proof}[Proof of Lemma \ref{lemma:Constants}]

Recalling the Legendre duplication formula,
$\Gamma (z)\Gamma \left(z+{\tfrac {1}{2}}\right)=2^{1-2z}{\sqrt {\pi }}\;\Gamma (2z)$, we have that
\begin{align*}
\frac{T_{M,N,\ell}}{R_{M,N,\ell}}&=\prod_{i=1}^\ell\frac{\Gamma\big((i+1)/2\big)\Gamma\big((i+2)/2\big)}{\Gamma\big((N-M+\ell+1+i)/2\big)\Gamma\big((N-M+\ell+i)/2\big)}\\
&=2^{\ell(N-M+\ell-1)}
\prod_{i=1}^\ell\frac{\Gamma(i+1)}{\Gamma(N-M+\ell+i)}\,.
\end{align*}
Using $|\log(n!)-n\log n| <Cn$ for some absolute constant $C>0$, \eqref{eq:logToverR} follows since
\begin{align*}
\log\left(\frac{T_{M,N,\ell}}{R_{M,N,\ell}}\right)
&\le 2(C+\log2) \ell (N-M+2\ell)\\
&+\sum_{i=1}^\ell\Big(
i\log(i)-(N-M+\ell+i-1)\log(N-M+\ell+i-1)
\Big)
\\
&\le 2(C+\log2) \ell (N-M+2\ell) -\ell(N-M+\ell-1)\log(N-M+\ell)
\,.
\end{align*}

In order to prove  Eq.~\eqref{eq:logS}, we use Stirling's 
formula to get, for some absolute constant $C>0$,
\begin{align*}
\log S_{M,N,\ell}&\le \sum_{i=1}^\ell\left\{
\frac{N+i}{2}\log\frac{N+i}{2}-\frac{N+i}{2}-
\frac{1}{2}\log\frac{N+i}{2}\right.\\
&\phantom{= \sum_{i=1}^\ell}\left.-\frac{M-\ell+i}{2}\log\frac{M-\ell+i}{2}+\frac{M-\ell+i}{2}+
\frac{1}{2}\log\frac{M-\ell+i}{2}
\right\}+C\ell\\
&= \sum_{i=1}^\ell\left\{\frac{N+i}{2}\log\frac{N+i}{M-\ell+i}+\frac{N-M+\ell}{2}\log\frac{M-\ell+i}{2}\right\}\\
&\phantom{===}- \sum_{i=1}^\ell\left\{\frac{N-M+\ell}{2}+\frac{1}{2}\log\frac{N+i}{M-\ell+i}\right\}+C\ell\\
&\overset{(a)}{\le} -\sum_{i=1}^\ell\frac{N+i}{2}\log\left(1-\frac{N-M+\ell}{N+i}\right)+\ell\frac{N-M+\ell}{2}\log M+C\ell\\
&\overset{(b)}{\le} \ell\frac{N-M+\ell}{2}\log M +C\ell+C_1(c_0)\ell (N-M)\, ,
\end{align*}
where in $(a)$ we omitted the negative sum
and for $(b)$ we 
applied the inequality $-\log(1-x)\le Cx$ for $0\le x\le 1-C^{-1}$
with $C=(1-x)^{-1}$ to $x= (N-M+\ell)/(N+i)$.
\end{proof}

Since $(\sigma_{\min}(\bZ))^2=\lambda_{1}(\bZ\bZ^{\sT})$, the following is an immediate consequence of Lemma \ref{lem:K-th-Eigengalue-Wishart}.

\begin{corollary}\label{cor:improvedRV}
In the setting of Lemma \ref{lem:K-th-Eigengalue-Wishart},
\begin{align}
\prob\Big(\sigma_{\min}(\bZ)\le \eps\big(\sqrt{N}-\sqrt{M-1}\big)\Big) \le (C_*(c_0)\eps)^{N-M+1}\, .
\end{align}
\end{corollary}


\newcommand{\etalchar}[1]{$^{#1}$}
\providecommand{\bysame}{\leavevmode\hbox to3em{\hrulefill}\thinspace}
\providecommand{\MR}{\relax\ifhmode\unskip\space\fi MR }
\providecommand{\MRhref}[2]{%
  \href{http://www.ams.org/mathscinet-getitem?mr=#1}{#2}
}
\providecommand{\href}[2]{#2}

\end{document}